\documentclass[a4paper]{elsarticle}

\usepackage{hyperref}
\usepackage{amsmath}
\usepackage{amssymb}
\usepackage{graphicx}
\usepackage{longtable}
\usepackage{url}

\renewcommand{\And}{\wedge}

\newcommand{\Union}{\bigcup}
\newcommand{\pow}[1]{{\sf pow}\ #1}
\renewcommand{\emptyset}{\{\}}
\newcommand{\emptylist}{[\,]}
\newcommand{\cons}[2]{#1\;\#\;#2} 

\newcommand{\distinct}[1]{{\sf distinct}\ #1}
\newcommand{\remdups}[1]{{\sf remdups}\ #1}

\newcommand{\last}[1]{{\sf last}\ #1}
\newcommand{\butlast}[1]{{\sf butlast}\ #1}
\newcommand{\declast}[1]{{\sf dec\_last}\ #1}
\newcommand{\incnth}[2]{{\sf inc\_nth}\ #1\ #2}

\newcommand{\sumfam}[2]{#1 \uplus #2}
\newcommand{\uc}[1]{{\sf uc}\ #1}
\newcommand{\uca}[2]{{\sf uc}_{#1}\ #2}
\newcommand{\closurea}[2]{\left\langle#2\right\rangle_{#1}}
\newcommand{\ic}[2]{{\sf ic}\ #1\ #2}
\newcommand{\ica}[3]{{\sf ic}_{#1}\ #2\ #3}
\newcommand{\icaa}[4]{{\sf ic}_{#1}^{#4}\ #2\ #3}
\newcommand{\frankl}[1]{{\sf frankl}\ #1}

\newcommand{\wf}[2]{{\sf wf}_{#2}\ #1} 
\newcommand{\sw}[2]{#1(#2)} 
\newcommand{\fw}[2]{#1(#2)} 
\renewcommand{\ss}[3]{\bar{#2}_{#3}(#1)} 
\newcommand{\fs}[3]{\bar{#2}_{#3}(#1)}  
\newcommand{\hc}[2]{{\sf hc}_{#1}^{#2}}  
\newcommand{\hs}[5]{\bar{#4}^{#2}_{#1#5}(#3)} 
\newcommand{\hcprj}[3]{{\sf hc}_{#1}^{#2}\left\lfloor{#3}\right\rfloor} 
\newcommand{\uce}[1]{{\sf uce}\ #1} 

\newcommand{\ssn}[2]{{\sf ssn}\ #1\ #2} 
\newcommand{\ssnaux}[5]{{\sf ssn}^{#3,#4}\ #1\ #2} 
\newcommand{\ssnauxa}[4]{{\sf ssn}^{#3,#4}\ #1\ #2} 

\newcommand{\card}[1]{|#1|}
\newcommand{\length}[1]{|#1|}
\newcommand{\nth}[2]{#1_{#2}}
\newcommand{\mult}[3]{#1 \odot_{#3} #2}
\newcommand{\union}{\cup}
\newcommand{\inter}{\cap}
\newcommand{\expressible}[2]{\mathit{depends}\ #1\ #2}
\newcommand{\setn}[1]{\{\overline{#1}\}}        
\newcommand{\listn}[1]{[\overline{#1}]}        

\newcommand{\indep}[1]{\mathit{ir}\ #1}
\newcommand{\iso}[2]{#1 \cong #2}
\newcommand{\mapfam}[2]{#1 ` #2}
\newcommand{\mapset}[2]{#1 ` #2}
\newcommand{\FC}{\mathcal{F}_c}
\newcommand{\nonFC}{\mathcal{N}_c}
\newcommand{\FCsix}{\mathcal{F}_c^6}
\newcommand{\nonFCsix}{\mathcal{N}_c^6}
\newcommand{\closure}[1]{\left\langle#1\right\rangle}
\newcommand{\nfam}[1]{\{\!\{\overline{#1}\}\!\}}
\newcommand{\cnt}[2]{\#_{#1}#2} 
\newcommand{\FCcovered}[2]{#2\vdash#1}
\newcommand{\nonFCcovered}[2]{#2\Vdash#1}
\newcommand{\covered}[3]{(#2, #3)\vDash#1}
\newcommand{\families}[2]{\binom{\setn{#1}}{#2}}
\newcommand{\LnP}[3]{{#1}_{#2}^{#3}}
\newcommand{\LnPirsix}{\LnP{L}{6}{ir}}
\newcommand{\take}[2]{#2^{[#1]}}
\newcommand{\base}[2]{iso\_reduce_{#2}\ #1}

\newcommand{\enumrecname}{{\sf enum\_rec}}
\newcommand{\enumrec}[3]{\enumrecname^{#2\,#3}\ #1}

\newcommand{\enumdpname}{{\sf enum\_dp}}
\newcommand{\enumdp}[4]{\enumdpname\ {#4\ #3\ #2\ #1}}
\newcommand{\enumdppname}{{\sf enum\_dp\_aux}}
\newcommand{\enumdpp}[7]{\enumdppname^{#7\,#6\,#5}\ #1\ #3\ #2\ #4}

\newcommand{\Pinc}[1]{\overline{P}}

\newdefinition{definition}{Definition}
\newtheorem{theorem}{Theorem}
\newtheorem{example}{Example}
\newtheorem{lemma}{Lemma}
\newtheorem{proposition}{Proposition}
\newproof{proof}{Proof}

\title{Fully Automatic, Verified Classification of all Frankl-Complete (FC(6)) Set Families\tnoteref{t1}} 
\tnotetext[t1]{The first and the third author were partially
    supported by the Serbian Ministry of Education and Science grant
    174021 and the second author by the Serbian Ministry of Education
    and Science grant 044006 (III).}
\author[fm]{Filip Mari\' c}
\ead{filip@matf.bg.ac.rs}
\author[bv]{Bojan Vu\v ckovi\' c}
\ead{mm99030@alas.matf.bg.ac.rs}
\author[mz]{Miodrag \v Zivkovi\' c}
\ead{ezivkovm@alas.matf.bg.ac.rs}
\address[fm]{Faculty of Mathematics, University of Belgrade}
\address[bv]{Mathematical Institute of the Serbian Academy of Sciences and Arts}
\address[mz]{Faculty of Mathematics, University of Belgrade}

\begin{document}
\begin{abstract}
  The \emph{Frankl's conjecture}, formulated in 1979. and still open,
  states that in every family of sets closed for unions there is an
  element contained in at least half of the sets. A family $F_c$ is
  called \emph{Frankl-complete} (or FC-family) if in every
  union-closed family $F \supseteq F_c$, one of the elements of
  $\Union F_c$ occurs in at least half of the elements of $F$ (so $F$
  satisfies the Frankl's condition). FC-families play an important
  role in attacking the Frankl's conjecture, since they enable
  significant search space pruning. We extend previous work by giving
  a total characterization of all FC-families over a 6-element
  universe, by defining and enumerating all minimal FC and maximal
  nonFC-families. We use a fully automated, computer assisted
  approach, formally verified within the proof-assistant Isabelle/HOL.
\end{abstract}

\maketitle

\section{Introduction}
\label{sec:introduction}

\emph{Union-closed set conjecture}, an elementary and fundamental
statement formulated by P\' eter Frankl in 1979. (therefore also
called \emph{Frankl's conjecture}), states that for every family of
sets closed under unions, there is an element contained in at least
half of the sets (or, dually, in every family of sets closed under
intersections, there is an element contained in at most half of the
sets). Up to the best of our knowledge, the problem is still open, and
that is not because of the lack of interest --- a recent survey by
Bruhn and Schaudt lists more than 50 published research articles on
the topic \cite{frankl-survey}.

The conjecture has been confirmed for many finite special cases. For
example, Bo\v snjak and Markovi\' c \cite{frankl-bosnjak-markovic}
proved that the conjecture holds for families such that their union
has at most $m=11$ elements and \v Zivkovi\' c and Vu\v ckovi\' c
\cite{frankl-zivkovic-vuckovic} describes the use of computer programs
to check the case of $m=12$ elements. Lo Faro \cite{frankl-lo-faro}
establishes the connection between the size of the union and the size
of the minimal counter-example, proving that for any $m$ the minimal
counter-example has at least $4m-1$ sets. Using results of Zivkovi\' c
and Vu\v ckovi\' c the conjecture is true for every family containing
$n \le 48$ sets.

It can easily be shown that if a union-closed family contains a
one-element set, then that element is abundant (occurs in at least
half of the sets). Similarly, one of the elements of a two-element set
in a family is abundant. Unfortunately, as first shown by Renaud and
Sarvate \cite{frankl-sarvate-renaud}, the pattern breaks for a
three-element set. This motivates the search for good \emph{local
  configurations} as they enable significant search space
pruning. Following Vaughan, these are sometimes called
\emph{Frankl-complete} families (or just \emph{FC-families}). A family
$F_c$ is an FC-family if in every union-closed family
$F \supseteq F_c$, one of the elements of $\Union F_c$ is abundant. A
FC family is called FC($n$) if its union is an $n$-element set. Most
effort has been put on investigating \emph{uniform} families, where
all members have the same number of elements. The number FC($k$, $n$)
is the minimal number $m$ such that any family containing $m$
$k$-element sets whose union is an $n$-element set is an
FC-family. Poonen gives a necessary and sufficient conditions for a
family to be FC \cite{frankl-poonen}.

As it is usually the case in finite combinatorics, even for small
values of $n$, a combinatorial explosion occurs and assistance of a
computer is welcome for case-analysis within proofs. The corresponding
paradigm is sometimes called \emph{proof-by-evaluation} or
\emph{proof-by-computation}. Since these are not classical
mathematical results, these proofs sometimes raise controversies. We
support this criticism, and advocate that the use of computer programs
in classical mathematical proofs should be allowed only if the
programs are formally verified.

In our previous work \cite{frankl-cicm}, we have applied
proof-by-computation techniques and developed a fully verified
algorithm that can formally prove that a given family is FC and have
applied it to confirm some known uniform FC-families and to discover a
new FC-family (we have shown that each family containing a four
3-element sets contained in a 7-element set is FC, i.e., that FC($3$,
$7$) $\le 4$, which, together with the lower bound on the number of
3-sets of Morris \cite{frankl-morris} gives that FC($3$, $7$) $=4$).

In this paper we extend these results by giving a fully automated and
mechanically verified (within a proof assistant) characterization of
all FC($n$) families for $n \le 6$. Such characterization requires
three components:
\begin{enumerate}
\item a method to prove (within a proof-assistant) that some families
  are FC (the technique relies on the Poonen's Theorem
  \cite{frankl-poonen} and was already formalized in our previous work
  \cite{frankl-cicm}),
\item a method to prove (within a proof-assistant) that some families
  are not FC (the technique also relies on the Poonen's Theorem
  \cite{frankl-poonen}, but this is the first time that it is
  formalized),
\item finding a list $\FC$ of FC and a list $\nonFC$ of nonFC-families
  that are characteristic in some sense, formally verifying (within a
  proof-assistant) their FC-status (i.e., proving if a family is FC or
  nonFC), enumerating (within a proof-assistant) all relevant families
  from a 6-element universe and proving that all of them are in some
  sense covered by some of those characteristic families i.e., that
  their FC-status directly follows from the status of the covering
  family (this technique is novel).
\end{enumerate}

Finding a list of characteristic FC and nonFC-families requires lot of
experimenting and checking the FC-status of many candidate
families. It has recently been shown that this process can be fully
automated\footnote{All FC-families classification results in the
  present paper were obtained prior to Pulaj's algorithm \cite{pulaj}
  and for determining the FC status of various families we used a
  semi-automated procedure that is in spirit somewhat similar to
  Pulaj's technique. Afterwards we fully automated the procedure, and
  confirmed previous results.}. Namely, Pulaj recently proposed a
fully automated method for determining the FC-status of an arbitrary
given family \cite{pulaj}. The method is based on linear integer
programming, and, although not integrated within a proof-assistant, it
is very reliable, as it uses exact arithmetic. Even with the fully
automated FC-status checking procedure, the third point requires
nontrivial effort and is the main contribution of this paper (although
there are well-known algorithms for exhaustive generation of
non-isomorphic objects \cite{isomorph-free}).

Apart from the significance of this core result, an important
contribution of this paper is to demonstrate that in the field of
finite combinatorics it is possible to use computer programs to push
the bounds and simplify proofs, but in a way that does not jeopardize
proof correctness. On the contrary, since all statements and
algorithms have been verified within the theorem prover Isabelle/HOL,
the trust in our results is significantly higher than most classical
pen-and-paper proofs previously published on this topic. We emphasize
that many experiments may be performed by unverified tools, and only
the final results need to be checked within proof-assistants (e.g., we
find the list of characteristic FC and nonFC-families using unverified
tools, and verify only the final list using Isabelle/HOL).

\paragraph{Overview of the paper}
The paper is organized as follows. In the rest of this section we
describe proofs by computation and discuss some related work.  In
Section \ref{sec:background} we describe Isabelle/HOL and notation
that is going to be used in the paper.  In Section \ref{sec:basic} we
formally introduce basic definitions related to the Frankl's
conjecture (Frankl's condition, FC and nonFC-families, etc.). In
Section \ref{sec:provefc} we give a theorem (based on Poonen's theorem
\cite{frankl-poonen}) that can be used to formally prove that a family
is FC and describe two different approaches for checking the
conditions of that theorem (one based on a specialized, verified
procedure, and one based on linear integer programming). In Section
\ref{sec:provenotfc} we give a theorem (also based on Poonen's theorem
\cite{frankl-poonen}) that can be used to formally prove that a family
is nonFC. In Section \ref{sec:procedureFCstatus} we describe a fully
automated (unverified) procedure for checking if an arbitrary given
family is FC. In Section \ref{sec:characteristic} we define the notion
of covering and describe properties that our characteristic families
should satisfy. In Section \ref{sec:enum} we describe methods for
enumerating all families with certain properties that is used both
within an automated (unverified) procedure for finding all
characteristic families, and to formally show that all families are
covered by the given characteristic families. In Section \ref{sec:fc6}
we give a full characterization of FC(6) families, by listing all
found characteristic families, and formally proving that they cover
all families in $\nfam{6}$. In Section \ref{sec:conclusions} we draw
final conclusions and discuss possible further work.

\paragraph{ITPs and Proofs by computation}
Interactive theorem provers (sometimes called proof assistants), like
Coq, Isabelle/HOL, HOL Light, etc., have made great progress in recent
years. Many classical mathematical theorems have been formally proved
and proof assistants have been intensively used in hardware and
software verification. Several of the most important results in formal
theorem proving are for the problems that require proofs with much
computational content. These proofs are usually highly complex (and
therefore often require justifications by formal means) since they
combine classical mathematical statements with complex computing
machinery (usually computer implementation of combinatorial
algorithms). The corresponding paradigm is sometimes referred to as
\emph{proof-by-evaluation} or \emph{proof-by-computation}. Probably,
the most famous examples of this approach are the proofs of the
Four-Color Theorem \cite{gonthier-notices} and the Kepler's conjecture
\cite{flyspeck}. One of the authors of this paper, recently used a
proof-by-computation technique to give a formal proof of the
Erd\"os-Szekeres conjecture for hexagons \cite{maric-erdos} within
Isabelle/HOL.

\paragraph{Related work}
Bruhn and Schaudt give a detailed survey of the Frankl's conjecture
\cite{frankl-survey}.

The Frankl's conjecture has also been formulated and studied as a
question in lattice theory
\cite{frankl-lattices-reinhold,frankl-lattices-abe}, and in the graph
theory \cite{frankl-graphs}.

FC-families have been introduced by Poonen \cite{frankl-poonen} who
gave a necessary and a sufficient condition for a family to be FC
(based on weight functions). The term FC-family was coined by Vaughan
\cite{frankl-vaughan-1}, and they were further studied by Gao and Yu
\cite{frankl-gao-yu}, Vaughan
\cite{frankl-vaughan-1,frankl-vaughan-2,frankl-vaughan-3}, Morris
\cite{frankl-morris}, Markovi\'c \cite{frankl-markovic}, Bo\v snjak
and Markovi\'c \cite{frankl-bosnjak-markovic}, and \v Zivkovi\'c and
Vu\v ckovi\'c \cite{frankl-zivkovic-vuckovic}.  Poonen
\cite{frankl-poonen} proved that FC($3$, $4$) = $3$. Vaughan
\cite{frankl-vaughan-1,frankl-vaughan-2,frankl-vaughan-3} showed that
FC($4$, $5$) $\le 5$ and FC($4$, $6$) $\le 10$. Morris
\cite{frankl-morris} gives a full characterization of all
FC($5$)-families. He proves that FC($3$, $5$)=$3$, FC($4$,
$5$)=$5$. Also, he proves that FC($3$, $6$)=$4$ and $7 \le$ FC($4$,
$6$) $\le 8$. His proofs rely on computer programs, but these are not
verified and not even presented in the article (as they are ,,fairly
simple-minded''). In our previous work \cite{frankl-cicm} we formally
confirmed all these results within a theorem prover, additionally
formally proving that FC($3$, $7$) $\leq 4$.

Computer-assisted computational approach was applied by Morris
\cite{frankl-morris} and \v Zivkovi\' c and Vu\v ckovi\' c
\cite{frankl-zivkovic-vuckovic} for solving special cases of the
Frankl's conjecture. In the latter case, computations are performed by
unverified Java programs.

\section{Background and notation}
\label{sec:background}

Logic and the notation given in this paper will follow Isabelle/HOL,
with some minor simplifications to make it approachable to wider
audience.  Isabelle/HOL \cite{isabelle} is a development of Higher
Order Logic (HOL), and it conforms largely to everyday mathematical
notation. Embedded in a theory are the types, terms and formulae of
HOL. The basic types include truth values ($\mathit{bool}$), natural
numbers ($\mathit{nat}$) and integers ($\mathit{int}$).

Terms are formed as in functional programming by applying functions to
arguments. Following the trandition of functional programming,
functions are curried. For example, $f\ x\ y$ denotes the function $f$
applied to the arguments $x$ and then $y$ (in classical mathematics
notation, this would usually be denoted by $f(x, y)$).Terms may also
contain $\lambda$-abstractions. For example, $\lambda x.\ x+1$ is the
function that takes an argument $x$ and returns
$x+1$. Let-expressions, if-expressions, and case-expressions are also
supported in terms. Let expressions are of the form "$\mathrm{let}\
x_1 = t_1; \ldots; x_n = t_n\ \mathrm{in}\ t$". This expressions is
equivalent to the one obtained from the term $t$ by substituting all
free occurrences of the variable $x_i$ by the $t_i$. For example
"$\mathrm{let}\ x=0\ \mathrm{in}\ x+x$" is equivalent to "$0+0$". If
expression is of the form "$\mathit{if}\ b\ \mathrm{then}\ t_1\
\mathrm{else}\ t_2$". Case expressions are of the form
"$\mathrm{case}\ e\ \mathrm{of}\ \mathit{pat_1} \Rightarrow e_1\ |\
\ldots\ |\ \mathit{pat_m} \Rightarrow e_m$". This is equivalent to
$e_i$ if $e$ matches the pattern $pat_i$.

Formule are terms of the type $\mathit{bool}$. Standard logical
connectives ($\neg$, $\wedge$, $\vee$, $\rightarrow$ and
$\longrightarrow$) are supported. Quanfiers are written using
dot-notation, as $\forall x.\ P$, and $\exists x.\ P$.  

New functions can be defined by recursion (either primitive or
general).

Sets over type $\alpha$, type $\alpha\,\mathit{set}$, follow the usual
mathematical conventions\footnote{In a strict type setting, sets
  containing elements of mixed types are not allowed.}. In the
presentation we use the term \emph{set} for sets of numbers and
denote these by $A$, $A'$, \ldots, the term \emph{family} for sets of
sets (i.e., object of the type $\alpha\,\mathit{set}\,\mathit{set}$)
of numbers and denote these by $F$, $F'$, \ldots and the term
\emph{collection} for sets of families (i.e., object of the type
$\alpha\,\mathit{set}\,\mathit{set}\,\mathit{set}$) and denote these
by $\mathcal{F}$, $\mathcal{F}'$, \ldots. The powerset (set of all
subset) of a set $A$ will denoted by $\pow{A}$. Union of sets $A$ and
$B$ is denoted by $A \union B$, and the union of all sets in a family
$F$ is denoted by $\Union{F}$. Image of a set $A$ under a function $f$
is denoted by $\mapset{f}{A}$. In this paper, the number of elements
in a set will be denoted by $\card{A}$. The set
$\{0, 1, \ldots, n-1\}$ will be denoted by $\setn{n}$.

Lists over type $\alpha$, type $\alpha\,\mathit{list}$, come with the
empty list $\emptylist$, and the infix prepend constructor $\#$ (every
list is either $\emptylist$ or is of the form $\cons{x}{xs}$ and these
two cases are usually considered when defining recursive functions
over lists). Standard higher order functions ${\sf map}$,
${\sf filter}$, ${\sf foldl}$ are supported and very often used for
defining list operations (for details see \cite{isabelle}). In this
paper, the N-th element of a list $l$ will be denoted by $\nth{l}{n}$
(positions are zero-based). $\butlast{l}$ denotes the list obtained
from $l$ by removing its last argument. If $l$ contains natural
numbers, $\declast{l}$ is the list obtained from $l$ be decreasing its
last element, and $\incnth{l}{n}$ is the list obtained from $l$ by
increasing its $n$-th element. The predicate $\distinct{l}$ checks if
the list $l$ has no repated elements, and the function $\remdups{l}$
removes duplicates from the list $l$.  List $[0, 1, \ldots, n-1]$ will
be denoted by $\listn{n}$.

\medskip All definitions and statements given in this paper are
formalized within Isabelle/HOL\footnote{Formal proofs are available at
  \url{http://argo.matf.bg.ac.rs/downloads/formalizations/FCFamilies.zip}}. However, in order
to make the text accessible to a more general audience not familiar
with Isabelle/HOL, many minor details are omitted and some
imprecisions are introduced. For example, we use standard symbolic
notation common in related work, although it is clear that some symbols
are ambiguous. Also, in the paper some notions will be defined by only
using sets, while in the formalization they are defined by using lists
(to obtain executability). Statements are grouped into propositions,
lemmas, and theorems. Propositions usually express simple, technical
results and are printed here without proofs, while the proofs of
lemmas and theorems are given in the Appendix. All sets and families
are considered to be finite and this assumptions (present in
Isabelle/HOL formalization) will not be explicitly stated in the rest
of the paper.

\section{Basic notions}
\label{sec:basic}

Since we are only dealing with finite sets and families, without loss
of generality we can restrict the domain only to natural number
domains.

\begin{definition}
  A family $F$ over $\setn{n}$ is a collection of sets such that
  $\Union F \subseteq \setn{n}$. The collection of all families over
  $\setn{n}$ will be denoted by $\nfam{n}$.
\end{definition}

\subsection{Union-Closed Families}
\label{sec:unionclosed}

First we give basic definitions of union-closed families, closure
under unions, and operations used to incrementally obtain closed
families. Let
$\sumfam{F_1}{F_2} = \{A \union B.\ A \in F_1 \wedge B \in F_2 \}$.

\begin{definition}
  Let $F$ and $F_c$ be families.

  A family $F$ is \emph{union-closed}, denoted by $\uc{F}$, iff
  $\sumfam{F}{F} = F$,
  (i.e. $\forall A \in F.\ \forall B \in F.\ A \union B \in F$). 
  A family $F$ is \emph{union-closed for $F_c$}, denoted by
  $\uca{F_c}{F}$, iff $\uc{F} \wedge (\sumfam{F}{F_c} \subseteq F)$,
  (i.e.
  $\uc{F} \wedge (\forall A \in F.\ \forall B \in F_c.\ A \union B \in
  F)$).

  \emph{Union-closure of $F$ (abbr.~closure)}, denoted by
  $\closure{F}$, is the minimal family of sets (in sense of inclusion)
  that contains $F$ and is union-closed.

  \emph{Union-closure of $F$ for $F_c$ (abbr.~closure for
    $F_c$)}, denoted by $\closurea{F_c}{F}$, is the minimal family of
  sets (in sense of inclusion) that contains $F$ and is union-closed
  for $F_c$.

  \emph{Insert and close operation} of set $A$ to family $F$, denoted
  by $\ic{A}{F}$, is the family $F \union \{A\} \union (\sumfam{F}{\{A\}})$.
  \emph{Insert and close operation for $F_c$} of set $A$ to family
  $F$, denoted by $\ica{F_c}{A}{F}$, is the family
  $F \union \{A\} \union (\sumfam{F}{\{A\}}) \union
  (\sumfam{F_c}{\{A\}})$.
\end{definition}

The following proposition gives some trivial properties of these
notions.

\begin{proposition}
\label{prop:closure}\hfill
\vspace{-2mm}
\begin{enumerate}
\item $\closure{F} = \{\Union F'.\ F' \in \pow{F} - \{\emptyset\}\}$
\item $\closure{F \union \{A\}} = \ic{A}{\closure{F}}$, \quad $\closurea{F_c}{F \union \{A\}} = \ica{F_c}{A}{\closure{F}}$
\item If $F \subseteq \pow{\Union F_c}$ and $\uca{F_c}{F}$ then
  $\uca{\closure{F_c}}{F}$.
\item If $\uc{F'}$ and $F \subseteq F'$ then $\closure{F} \subseteq
  F'$.
\end{enumerate}
\end{proposition}

\subsection{The Frankl's Condition}
\label{sec:franklcond}
The next definition formalizes the Frankl's condition and the notion
of FC-family.

\begin{definition}
  Family of sets $F$ is a \emph{Frankl's family}, denoted by
  $\frankl{F}$, if it contains an element that satisfies the
  \emph{Frankl's condition for $F$}, i.e., that occurs in at least
  half sets in the family $F$. Formally,
  $\frankl{F}\ \equiv\ \exists a.\ a \in \Union F\ \And\ 2 \cdot
  \cnt{a}{F} \ge \card{F}$,
  where $\cnt{a}{F}$ denotes $\card{\{A \in F.\ a \in A\}}$.
  \end{definition}

\subsection{FC-families}
\label{sec:fc}
\begin{definition}
  Family of sets $F_c$ is an \emph{FC-family} if in every union-closed
  family $F$ such that $F \supseteq F_c$ one of the elements of
  $\Union{F_c}$ satisfies the Frankl's condition for $F$.  Every
  family that is not an FC-family is called a \emph{nonFC-family}.
\end{definition}

The next propositions give some properties of FC-families.

\begin{proposition}
  \label{prop:FC_family_mono}
  Any superset of an FC-family is an FC-family. Any subset of a
  nonFC-family is a nonFC-family.
\end{proposition}

\begin{proposition}
  \label{prop:FC_family_empty}
  A family $F_c$ is an FC-family iff the family $F_c \setminus
  \{\emptyset\}$ is an FC-family.
\end{proposition}

\begin{proposition}
  \label{prop:FC_family_closure}
  A family $F$ is an FC-family iff its closure $\closure{F}$ is an
  FC-family.
\end{proposition}

\section{Proving that a Family is FC}
\label{sec:provefc}

In this section we describe techniques that can be used to formally
prove that a given family is FC. Most statements will be given without
proofs, since the proofs are available in \cite{frankl-cicm}.

\subsection{Weight Functions and Shares}
\label{sec:weightsshares}
We describe the central technique for proving that a family is FC,
relying on characterizations of the Frankl's condition using weights
and shares introduced by Poonen \cite{frankl-poonen}, but adapted to
work in a proof-assistant environment.

\begin{definition}
  A function $w: X \rightarrow \mathbb{N}$ is a \emph{weight function
    on} $A \subseteq X$, denoted by $\wf{w}{A}$, iff $\exists a \in
  A.\ w(a) > 0$.  \emph{Weight of a set $A$ wrt.~weight function $w$},
  denoted by $\sw{w}{A}$, is the value $\sum_{a \in A}w(a)$.
  \emph{Weight of a family $F$ wrt.~weight function $w$}, denoted by
  $\fw{w}{F}$, is the value $\sum_{A\in F}\sw{w}{A}$.
\end{definition}

An important technique for checking Frankl's condition is
\emph{averaging} --- family is Frankl's if and only if there is a
weight function such that weighted average of number of occurrences of
all elements exceeds $\card{F} / 2$. A more formal formulation of this
claim (that uses only integers and avoids division) is given by the
following Proposition.
\begin{proposition}
\label{lemma:Frankl_weight}
$\frankl{F} \iff \exists w.\ \wf{w}{(\Union{F})}\ \And\ 2 \cdot \fw{w}{F} \;\ge\; \sw{w}{\Union{F}}\cdot\card{F}$
\end{proposition}

A concept that will enable a slightly more operative formulation of
the previous characterization is the concept of \emph{share} (again,
to avoid rational numbers, definition is different from $\sw{w}{A} -
\sw{w}{X} / 2$ that is used in the literature).
\begin{definition}
  Let $w$ be a weight function.  \emph{Share of a set $A$ wrt.~$w$ and
    a set $X$}, denoted by $\ss{A}{w}{X}$, is the value
  $\sw{w}{A} - \sw{w}{X \setminus A} = 2 \cdot \sw{w}{A} - \sw{w}{X}$.
  \emph{Share of a family $F$ wrt.~$w$ and a set $X$}, denoted by
  $\fs{F}{w}{X}$, is the value $\sum_{A \in F}\ss{A}{w}{X}$.
\end{definition}

\begin{example}
\label{ex:share}
Let $w$ be a function such that $w(a_0) = 1, w(a_1) = 2$, and $w(a) =
0$ for all other elements. $w$ is clearly a weight function. Then,
$\sw{w}{\{a_0, a_1, a_2\}} = 3$ and $\fw{w}{\{\{a_0, a_1\}, \{a_1,
  a_2\}, \{a_1\}\}} = 7$. Also, $\ss{\{a_1, a_2\}}{w}{\{a_0, a_1,
  a_2\}} = 2\cdot \sw{w}{\{a_1, a_2\}} - \sw{w}{\{a_0, a_1, a_2\}} = 4
- 3 = 1,$ and $\fs{\{\{a_0, a_1\}, \{a_1, a_2\},
  \{a_1\}\}}{w}{\{a_0, a_1, a_2\}} = (2\cdot 3 - 3) + (2\cdot 2 - 3) +
(2 \cdot 2 - 3) = 5.$
\end{example}

\begin{proposition}
\label{prop:Family_share}
$\fs{F}{w}{X} = 2 \cdot \fw{w}{F} - \sw{w}{X}\cdot\card{F}$
\end{proposition}

\begin{proposition}
\label{thm:Frankl_Family_share_ge_0}
$\frankl{F} \iff \exists w.\ \wf{w}{(\Union{F})}\ \And\ \fs{F}{w}{(\Union{F})} \ge 0$
\end{proposition}

\paragraph{Union-closed extensions}
The next definition introduces an important notion for checking
FC-families.
\begin{definition}
  \emph{Union-closed extensions} of a family $F_c$ are families that
  are created from elements of the domain of $F_c$ and are union
  closed for $F_c$.  Collection of all union-closed extensions is
  denoted by $\uce{F_c}$, and defined by
  $\uce{F_c} \equiv \{F.\ F \subseteq \pow{\Union{F_c}} \And
  \uca{F_c}{F}\}$.
\end{definition}

The following theorem corresponds to first direction of Poonen's
theorem (Theorem 1 in \cite{frankl-poonen}). The proof is formalized
within Isabelle/HOL and its informal counterpart is given in the
Appendix.

\begin{theorem}
\label{thm:FC_uce_shares_nonneg}
A family $F_c$ is an FC-family if there is a weight function $w$ such
that shares (wrt.~$w$ and $\Union F_c$) of all union-closed extension
of $F_c$ are nonnegative, i.e.,
$\forall F \in \uce{F_c}.\ \fs{F}{w}{(\Union{F_c})} \ge 0$.
\end{theorem}

\bigskip

In the rest of this section we show two different possibilities for
searching for a union-closed extension with a negative share -- the
first is based on a specialized algorithm, crafted specifically for
this problem, while the other is based on integer linear programming
and employs an integer linear programming (ILP) package or a
satisfiability modulo theory (SMT) solver.

\subsection{Search for Negative Shares}
\label{sec:search}

Theorem \ref{thm:FC_uce_shares_nonneg} inspires a procedure for
verifying FC-families. It should take a weight function on
$\Union{F_c}$ and check that all union-closed extensions of $F_c$ have
nonnegative shares. There are only finitely many union-closed
extensions, so in principle, they can all be checked. However, in
order to have efficient procedure, naive checking procedure will not
suffice and further steps must be taken. We now define a
procedure \emph{SomeShareNegative}, denoted by $\ssn{F_c}{w}$, such
that $\ssn{F_c}{w} = \top$ iff there is an $F \in \uce{F_c}$ such that
$\fs{F}{w}{(\Union{F_c})} < 0$. The procedure is based on a recursive
function $\ssnaux{L_r}{F_t}{F_c}{w}{(\Union{F_c})}$ that preforms a
systematic traversal of all union-closed extensions of $F_c$, but with
pruning that significantly speeds up the search. The procedure has
four parameters ($F_c$, $w$, $L_r$, and $F_t$) that we now
describe. The two fixed parameters of the function (parameters that do
not change troughout the recursive calls) are the family $F_c$ and the
weight function $w$. If a union-closed extension of $F_c$ has a
negative share, it must contain one or more sets with a negative
share. Therefore, a list $L$ of all different subsets of $\Union{F_c}$
with negative shares is formed and each candidate family is determined
by elements of $L$ that it includes. A recursive procedure creates all
candidate families by processing elements of that list sequentially,
either skipping them (in one recursive branch) or including them into
the current candidate family $F_t$ (in the other recursive branch),
maintaining the invariant that the current candidate family $F_t$ is
always from $\uce{F_c}$. The two parameters of the recursive function
$\ssnaux{L_r}{F_t}{F_c}{w}{(\Union{F_c})}$ that change during
recursive calls are the remaining part of the list $L_r$ and the
current candidate family $F_t$. If the current leading element of
$L_r$ has been already included in $F_t$ (by earlier closure
operations required to maintain the invariant) the search can be
pruned. If the sum of (negative) shares of $L_r$ (the remaining
elements of $L$) is less then the (nonnegative) share of the current
$F_t$, then $F_t$ cannot be extended to a family with a negative share
(even in the extreme case when all the remaining elements of $L$ are
included) so, again, the search can be pruned.

\begin{definition}The function
  $\ssnaux{L_r}{F_t}{F_c}{w}{(\Union{F_c})}$ is defined by a primitive
  recursion (over the structure of the list $L_r$):
  \begin{eqnarray*}
   \ssnaux{\emptylist}{F_t}{F_c}{w}{(\Union{F_c})} &\equiv& \fs{F_t}{w}{(\Union{F_c})} < 0 \\
   \ssnaux{(\cons{h}{t})}{F_t}{F_c}{w}{(\Union{F_c})} &\equiv& \mathrm{ if\ } \fs{F_t}{w}{(\Union{F_c})} + \sum_{A \in \cons{h}{t}}\ss{A}{w}{(\Union{F_c})} \ge 0\mathrm{\ then\ } \bot \\
    && \mathrm{else\ if\ }\ssnaux{t}{F_t}{F_c}{w}{(\Union{F_c})} \mathrm{\ then\ } \top\\
    && \mathrm{else\ if\ }h \in F_t \mathrm{\ then\ } \bot\\
    && \mathrm{else\ } \ssnaux{t}{(\ica{F_c}{h}{F_t})}{F_c}{w}{(\Union{F_c})}
  \end{eqnarray*}

  Let $L$ be a distinct list such that its set is
  $\{A.\ A \in \pow{\Union{F_c}}\ \And\ \ss{A}{w}{\Union{F_c}} < 0\}$.
  $$\ssn{F_c}{w} \equiv \ssnaux{L}{\emptyset}{\closure{F_c}}{w}{(\Union{F_c})}$$
\end{definition}

The soundnes of the $\ssn{F_c}{w}$ function is given by the following
propositions.

\begin{proposition}
  \label{lemma:ssnaux_correct}
  If (i) $\ssnaux{L_r}{F_t}{F_c}{w}{(\Union{F_c})} = \bot$, (ii) for
  all elements $A$ in $L_r$ it holds that
  $\ss{A}{w}{\Union{F_c}} < 0$, (iii) for all $A \in F - F_t$, if
  $\ss{A}{w}{\Union{F_c}} < 0$, then $A$ is in $L_r$, (iv)
  $F \supseteq F_t$, and (v) $\uca{F_c}{F}$, then
  $\fs{F}{w}{\Union{F_c}} \ge 0$.
\end{proposition}

\begin{proposition}
  \label{lemma:ssn_correct}
  If $\ssn{F_c}{w} = \bot$ and $F \in \uce{F_c}$ then
  $\fs{F}{w}{(\Union{F_c})} \ge 0$.
\end{proposition}

Apart from being sound, the procedure can also be shown to be
complete. Namely, it could be shown that if $\ssn{F_c}{w} = \top$,
then there is an $F \in \uce{F_c}$ such that
$\fs{F}{w}{(\Union{F_c})} < 0$. This comes from the invariant that the
current family $F_t$ in the search is always in $\uce{F_c}$, which is
maintained by taking the closure $\ica{F_c}{h}{F_t}$ whenever an
element $h$ is added. Since this aspect of the procedure is not
relevant for the rest of the proofs, it will not be formally stated
nor proved. However, this can give a method for finding a
counterexample family for a given weight function, that can be useful
for fully automated classification of a given family (described in
Section \ref{sec:procedureFCstatus}), that we use to find the
minimal FC-families (as described in Section \ref{sec:minimalFC}).

\paragraph{Optimizations}
Important optimization to the basic $\ssn{F_c}{w}$ procedure is to
avoid repeated computations of family shares (both for the elements of
the list $L_r$ and the current family $F_t$). So, instead of accepting
a list of families of sets $L_r$, and the current family of sets
$F_t$, the function is modified to accept a list of ordered pairs
where first component is a corresponding element of $L_r$, and the
second component is its share (wrt.~$w$ and $\Union{F_c}$), and to
accept an ordered pair $(F_t, s_t)$ where $s_t$ is its family share
(wrt.~$w$ and $\Union{F_c}$). The summation of shares of elements in
$L_r$ is also unnecessarily repeated. It can be avoided if the sum
$s_l$ is passed trough the function.

\begin{eqnarray*}
    \ssnaux{(\emptylist, 0)}{(F_t, s_t)}{F_c}{w}{(\Union{F_c})} &\equiv& s_t < 0 \\
    \ssnaux{(\cons{(h, s_h)}{t},\;s_l)}{(F_t,\;s_t)}{F_c}{w}{(\Union{F_c})} &\equiv& \mathrm{ if\ } s_t + s_l \ge 0\mathrm{\ then\ } \bot \\
    && \mathrm{else\ if\ }\ssnaux{(t,\;s_l - s_h)}{(F_t,\;s_t)}{F_c}{w}{(\Union{F_c})} \mathrm{\ then\ } \top\\
    && \mathrm{else\ if\ }h \in F_t \mathrm{\ then\ } \bot\\
    && \mathrm{else\ let\ } F_t' = \ica{F_c}{h}{F_t};\ s_t' = \fs{F_t'}{w}{(\Union{F_c})}\\
    && \qquad \mathrm{\ in}\ \ssnaux{(t, ls - s_h)}{(F_t',s_t'\;)}{F_c}{w}{(\Union{F_c})}
\end{eqnarray*}

Another source of inefficiency is the calculation of
$\fs{F_t'}{w}{(\Union{F_c})}$. If performed directly based on the
definition of family share for $F_t'$, the sum would contain shares of
all elements from $F_t$ and of all elements that are added to $F_t$
when adding $h$ and closing for $F$. However, it is already known that
the sum of shares for elements of $F_t$ is $s_t$ and the
implementation could benefit from this fact. Also, calculating shares
of sets that are added to $F_t$ can be made faster. Namely, it happens
that set share of a same set is calculated over and over again in
different parts of the search space. So, it is much better to
precompute shares of all sets from $\pow{(\Union{F_c})}$ and store
them in a lookup table that will be consulted each time a set share is
needed. Note that in this case there is no more need to pass the
function $w$ itself, nor to calculate the domain $\Union{F_c}$, but
only the lookup table, denoted by $s_w$.

\begin{eqnarray*}
    \ssnauxa{(\emptylist, 0)}{(F_t, s_t)}{F_c}{s_w} &\equiv& s_t < 0 \\
    \ssnauxa{(\cons{(h, s_h)}{t},\;s_l)}{(F_t,\;s_t)}{F_c}{s_w} &\equiv& \mathrm{ if\ } s_t + s_l \ge 0\mathrm{\ then\ } \bot \\
    && \mathrm{else\ if\ }\ssnauxa{(t,\;s_l - s_h)}{(F_t,\;s_t)}{F_c}{s_w} \mathrm{\ then\ } \top\\
    && \mathrm{else\ if\ }h \in F_t \mathrm{\ then\ } \bot\\
    && \mathrm{else\ } \ssnauxa{(t, s_l - s_h)}{(\icaa{F_c}{h}{(F_t, s_t)}{s_w})}{F_c}{s_w}
\end{eqnarray*}

\begin{eqnarray*}
  \icaa{F_c}{h}{(F_t, s_t)}{s_w} &\equiv& \mathrm{let\ }\ add\ = \ \{h\}\ \union\
       (\sumfam{F_t}{\{A\}}) \ \union\ (\sumfam{F_c}{\{A\}});\\
& & \qquad new\ =\ \{A \in add.\ A\notin F_t\}\\
& & \ \textrm{in}\ (new\; \union\;F_t,\ s_t + \sum_{A \in new}s_w\ A)\\
\end{eqnarray*}

We have shown that this implementation is equivalent to the starting,
abstract one (it returns false iff there is a union-closed extension
with a negative share).

\subsection{Integer linear programming}

An alternative to using a specialized, verified procedure
$\ssn{F_c}{w}$ is to encode the existence of a union-closed extension
$F$ with a negative share as a linear integer programming problem and
to employ an existing solver to do the search \cite{pulaj}. In our
case, we need to formally prove (within the Isabelle/HOL) that our
characteristic FC-families are indeed FC, so the SMT solver Z3
integrated within Isabelle/HOL can be used \cite{isabelle-z3}.

Assume that $F_c$ is given and $n$ is such number that
$\Union{F_c} = \setn{n}$. Each subset of $\setn{n}$ can be either
included or excluded from a family $F$. There are $2^n$ such subsets,
wich is significantly less than the number of families which is
bounded above by $2^{2^n}$. For each set $A \subseteq \setn{n}$ we
define a 0-1 integer (or Boolean) variable $x_A$, and its value is 1
iff the set is included in the sought family $F$ i.e.,
$x_A = 1 \longleftrightarrow A \in F$. We must encode that the family
is union-closed, so for every two sets $A \subseteq \setn{n}$ and
$B\subseteq \setn{n}$ it must hold that $A \in F$ and $B \in F$ imply
that $A \union B \in F$, that is
$x_A = 1\ \wedge\ x_B = 1 \longrightarrow x_{A \union B} = 1$, which
can be encoded as $$x_A + x_B \leq 1 + x_{A \union B}.$$

Next we must encode that family is closed for $F_c$, so for every set
$A \in F_c$ and $B \subseteq \setn{n}$ it must hold that
$B \in F \longrightarrow A \union B \in F$, which can be encoded as
$$x_B \leq x_{A \union B}.$$

Finally, we should encode that $F$ has a negative share, i.e.,
$\fs{F}{w}{(\Union{F_c})} < 0$. Since
$\fs{F}{w}{(\Union{F_c})} = \sum_{A \in F} \ss{A}{w}{\setn{n}} =
\sum_{A \subseteq \setn{n}} x_A \cdot \ss{A}{w}{\setn{n}}$,
the condition is equivalent to
$$\sum_{A \subseteq \setn{n}} x_A \cdot \ss{A}{w}{\setn{n}} < 0.$$

The conjunction of the three listed types of linear inequalities is
given to the SMT solver and it returns a model iff there is an
union-closed extension of $F_c$ with a negative share (values of
variables uniquely determine that extension $F$). The result of the
SMT solver (a model, or an unsatisfiability proof) is then verified by
Isabelle/HOL, yielding a fully formally verified proof
\cite{isabelle-z3}.

Note that the problem could be stated as a problem over rational
weights, but in our whole framework we considered only integers, and
it turned out that the search is efficient enough.

\section{Proving that a Family is not FC}
\label{sec:provenotfc}

Proving that a family is not an FC-family is also based on the
Poonen's theorem (Theorem 1 in \cite{frankl-poonen}). The converse of
our Theorem \ref{thm:FC_uce_shares_nonneg} also holds, and if there is
no weight function satisfying the conditions of Theorem
\ref{thm:FC_uce_shares_nonneg}, then the family $F_c$ is not an
FC-familly. However, this is hard to prove formally within
Isabelle/HOL (the original Poonen's proof uses the hyperplane
separation theorem for convex sets), so we formally proved the
following variant that is both easier to prove and more suitable for
further application.
\begin{theorem}
  \label{thm:nonFC}
  Assume that $F_c$ is a union-closed family. If there exists a
  sequence of families $F_0, \ldots, F_k$, and a sequence of natural
  numbers $c_0, \ldots, c_k$ that:
  \begin{enumerate}
  \item for all $0 \leq i \leq k$ it holds that $F_i \in \uce{F_c}$,
  \item for every $a \in \Union{F_c}$ it holds that $$\sum_{i=0}^k c_i \cdot (2 \cdot \cnt{a}{F_i} - \card{F_i}) < 0,$$
  \item not all $c_i$ are zero (i.e., $\exists i.\ 0 \le i \le k \wedge c_i > 0$),
  \end{enumerate}
  then the family $F_c$ is not an FC-family.
\end{theorem}

Major differences between this and Poonen's original formulation are
that instead of real we use only natural numbers, that instead of
considering the whole collection $\uce{F_c}$ we consider only some of
its members, and that instead of showing that there is no weight
function with non-negative shares for those selected union-closed
extensions i.e., showing that the system
$\fs{F_i}{w}{(\Union{F_c})} \geq 0$ i.e.,
$\sum_{a \in \Union{F_c}} w_a \cdot (2 \cdot \cnt{a}{F_i} -
\card{F_i}) \geq 0$,
for all every $0 \leq i \leq k$, has no all-nonnegative, non-all-zero
solutions, we show that the its dual system
$\sum_{i=0}^k c_i \cdot (2 \cdot \cnt{a}{F_i} - \card{F_i}) < 0$, for
every $a \in \Union{F_c}$, has a nontrivial solution
($2 \cdot \cnt{a}{F_i} - \card{F_i}$ equals the difference between the
number of members of $F_i$ that contain $a$ and the number of members
of $F_i$ that do not). The proof follows Poonen (to most extent) and
is given in the Appendix.

Note that once the sequence of families $F_0, \ldots, F_k$ and the
sequence of numbers $c_0, \ldots, c_k$ are known, the formal proof is
much easier than in the FC case, as it need not use any search (all
conditions of Theorem \ref{thm:nonFC} can be directly
checked). Finding those sequences is not trivial, but it can be done
outside Isabelle/HOL.

\section{Procedure for checking FC-status of a given family and
  finding witnesses}
\label{sec:procedureFCstatus}

To prove that a family is FC based on Theorem
\ref{thm:FC_uce_shares_nonneg} one requires a witnessing weight
function $w$. To prove that a family is nonFC based on Theorem
\ref{thm:nonFC} one requires a witnessing sequence of families $F_i$
and numbers $c_i$. For the final formal proof of the FC-status of
characteristic families it is not important how those witnesses are
obtained. It is very desired to have a procedure that can obtain them
fully automatically. Pulaj suggested the first algorithm capable of
checking the FC-status of an arbitrary family based on the cutting
planes method and linear (integer) programming implemented in SCIP
\cite{pulaj}, and it can easily be modified to provide required
witnesses (both for the FC and the nonFC case). Note that such
procedure need not be implemented within Isabelle/HOL -- its purpose
is to determine the status and give witnesses that can be used for
Theorem \ref{thm:FC_uce_shares_nonneg} or \ref{thm:nonFC}, which are
formally checked within Isabelle/HOL.

Assume that a family $F_c$ is given. The procedure alternates two
phases. In the first one a candidate weight function is constructed,
and in the second it is checked if it satisfies the condition of
Theorem \ref{thm:FC_uce_shares_nonneg}.

In the first phase, the candidate weight function (represented by
unkwowns $w_i$, for $0 \leq i < n$) is constructed by solving a system
of linear integer inequalities (as we use only natural numbers in our
framework). In the beginning the system contains only conditions
required for a weight function ($w_i \geq 0$ and $\sum_i w_i > 0$),
but as new families are constructed in the second phase, it is
extended by the condition $\fs{F_i}{w}{(\Union{F_c})} \geq 0$, for
each family $F_i$ obtained in the second phase. If the current system
becomes unsatisfiable, than $F_c$ is not FC-family, the current set of
families $F_i$ can be used as a witness for Theorem \ref{thm:nonFC}
and the coefficients $c_i$ are obtained by solving its dual system.
Otherwise, its solution is the candidate weight function used in the
second phase.

In the second phase it is checked if the weight function $w$ satisfies
the conditions of Theorem \ref{thm:FC_uce_shares_nonneg} i.e., that
there is no union-closed extension of $F_c$ with a negative share
wrt.~$w$. For this, either one of the two approaches described in
Section \ref{sec:search} (either on the $\ssn{F_c}{w}$ procedure or
solving the system of linear inequalities) can be used. If all shares
are non-negative, then $F_c$ is an FC-family and the current weight
function $w$ is used as a witness to formally prove that using Theorem
\ref{thm:FC_uce_shares_nonneg}. If it does not, than the procedure
constructs a family $F_i$ that is in the union-closed extension of
$F_c$ and has a negative share. That family is then added to the
current set of such families and fed into the first phase again.

Unlike in the final Isabelle/HOL proofs, in the experimentation phase
non-verified implementations can be used (since the final witnesses
are checked again, using Isabelle/HOL). Therefore, in our
implementation we have used the ILP package SCIP (the same one used in
\cite{pulaj}) in all three cases (solving the system for finding a
candidate weight, solving the system to find coefficients $c_i$ based
on the sequence of families $F_i$ for which finding the weight
function was shown to be impossible, and for solving the system that
finds a family with a negative share wrt.~the current weight function
$w$), as our preliminary experiments indicated that it gives results
faster then the SMT solver Z3 (when run outside
Isabelle/HOL). Interestingly, the $\ssn{F_c}{w}$ procedure often gave
a family $F_i$ faster then SCIP, but the overall procedure required
more iterations (we assume that this can be attributed to a very
regular order in which $\ssn{F_c}{w}$ enumerates families). One
additional technique for which we noticed that significantly speeds up
the convergence is to favor smaller weights i.e., to require that the
weight function $w$ is minimal wrt. its sum of the weights $w_i$ (this
was possible to obtain in SCIP by using its built-in optimization
features and the objective function $\sum_{i}w_i$).

\section{Characteristic families}
\label{sec:characteristic}

In this section we introduce the notion of \emph{FC-covering} and
\emph{nonFC-covering} that enables to determine the FC-status of all
families from $\nfam{n}$ from the status of just a small number of FC
and nonFC-families that are characteristic in some sense (that we
shall precisely define). Our goal is to give a full characterization
of all $2^{2^n}$ families from $\nfam{n}$ (i.e., for each family to
determine whether it is an FC-family or a nonFC-family), and in theory
that can be done by explicitly checking the status for each of
them. In practice that is almost impossible since even for $n=6$ there
are $2^{2^6} = 2^{64} \approx 2\cdot 10^{19}$ families.  However, (i)
many of them are isomorphic and (ii) many have the same closure and
(iii) many include smaller FC-families or are included in larger nonFC
families -- we shall show that in all those cases the FC-status can be
deduced from the already known status of other families, so we base
our definitions of characteristic families and covering on those
facts.  We shall devise methods that explicitly check the FC-status
for only a minimal set of characteristic families, and after that
enable us to easily get the status of every family from $\nfam{n}$ by
checking if they are covered by the characteristic ones.

\subsection{Isomorphic families. Representing collections. Bases.}
\label{sec:iso}

Bijective changes of the domain of a family do not affect if the
family is FC. 

\begin{definition}
  Two families $F$ and $F'$ are \emph{isomorphic} (denoted by
  $\iso{F}{F'}$) if there is a bijective function $f$ between
  $\Union{F}$ and $\Union{F'}$ such that $\mapfam{f}{F} = F'$. 
\end{definition}

If we consider families $\{\{a\}, \{a, b, c\}, \{a, c\}\}$ and
$\{\{0\}, \{0, 1, 2\}, \{0, 2\}\}$, they are clearly isomorphic, so we
consider only families over $\setn{n}$. The family
$\{\{0, 1, 2\}, \{1, 2\}, \{2\}\}$ also shares the same structure with
the previous two (although, that might not be so obvious, consider the
bijection $0 \mapsto 2, 1 \mapsto 0, 2\mapsto 1$)), so there are also
many isomorphic families over $\setn{n}$.

Obviously, isomorphism is an equivalence relation and isomorphic
families share all structural properties relevant to us ($F$ is union
closed iff and only if $F'$ is, $F$ satisfies the Frankl's condition
iff $F'$ does, the same holds for FC-family condition etc.).

\begin{proposition}
  If $\iso{F}{F'}$ then $F$ is an FC-family iff $F'$ is an FC-family.
\end{proposition}

\paragraph{Checking if the two families in $\nfam{n}$ are isomorphic}
One (naive) method to check if the two given families are isomorphic
is to check if the second family is among the families obtained by
applying all the permutations of $\listn{n}$ to the first family.

Another approach can be on defining the \emph{canonical
  representative} for each family. It can be the minimal family among
the families obtained by applying all permutations in $\listn{n}$ to
that family, where families are compared based on some fixed ordering
(i.e., a lexicographic ordering, where the sets are also ordered
lexicographically). Then two families are isomorphic iff they have the
same canonical representative.

There are more efficient orderings and methods of finding the
canonical representative, which avoid considering all permutations of
$\listn{n}$ \cite{generationUCMoore}, but since we only consider
the case $n=6$ where the number of permutations is rather small, we
use only the naive methods.

\paragraph{Iso-representatives and iso-bases}

If a collection of families contains many families whose structural
properties should be checked, it suffices to focus only on a single
representative from each isomorphism equivalence class.

\begin{definition}
  A collection $\mathcal{F}_b$ \emph{iso-represents} the collection
  $\mathcal{F}$ if for every $F \in \mathcal{F}$ there exists an
  $F_b \in \mathcal{F}_b$ such that $\iso{F}{F_b}$. If there are no
  $F_1 \in \mathcal{F}_b$ and $F_2 \in \mathcal{F}_b$ such that
  $\iso{F_1}{F_2}$, then $\mathcal{F}_b$ is a \emph{iso-base} of
  $\mathcal{F}$.
\end{definition}

Iso-base of a given collection can be found
algorithmically. Computation can start from the given collection
$\mathcal{F}$, choose its arbitrary member for a representative, move
it to the resulting collection, remove it and all its permuted
variants from the original collection (under a given set of
permutations), and repeat this sieving process until the list becomes
empty. Isabelle/HOL implementation of this procedure will be denoted
by $\base{\mathcal{F}}{\mathcal{P}}$ and its implementation is
available in our formal proof documents.

\begin{proposition}
  \label{lemma:nef}
  If $\mathcal{P}$ is a list of permutations of $\listn{n}$ and if
  $\mathcal{F}$ is a collection of families from $\nfam{n}$, then
  $\base{\mathcal{F}}{\mathcal{P}}$ iso-represents $\mathcal{F}$. If
  $\mathcal{P}$ contains all permutations of $\listn{n}$, then
  $\base{\mathcal{F}}{\mathcal{P}}$ is an iso-base of $\mathcal{F}$.
\end{proposition}

If an ordering of families is defined, another way to obtain an
iso-base is to find the canonical representative of each family, and
form the set of all different canonical representatives.

\subsection{Irreducible families}

Another technique that reduces the number of $\nfam{n}$ families for
which the FC-status explicitly needs to be checked is based on the
fact that the FC-status of a family depends only on its closure (and
not the family itself). From Proposition \ref{prop:FC_family_closure}
the following immediately follows.

\begin{proposition}
  If $\closure{F} = \closure{F'}$ then $F$ is an FC-family iff $F'$ is
  an FC-family.
\end{proposition}

\begin{definition}
  A set $A$ is \emph{dependent} on a family $F$ (denoted by
  $\expressible{A}{F}$) if it is a union of some of its members (i.e.,
  if $\exists F'.\ F' \subseteq F \And F' \neq \emptyset \And \ A =
  \Union{F'}$).
\end{definition}

\begin{proposition}\ \\[-1.5em]
  \begin{enumerate}
  \item If a set $A$ is dependent on a family $F$ then $A \in
    \closure{F}$.
  \item If a set $A$ is dependent on a family $F$, then $\closure{F
      \union \{A\}} = \closure{F}$.
  \end{enumerate}
\end{proposition}

Therefore, sets that can be expressed as unions of other sets of a
family do not affect its closure. \emph{Irreducible} family is
obtained if all dependent sets are removed (so this family is minimal
in some sense and it is a \emph{basis} of its closure \cite{kim}).

\begin{definition}
  A family is \emph{irreducible} if none of its sets can be expressed
  as a union of some of its other members (i.e., if
  $\nexists\,A\in F.\, \expressible{A}{(F \setminus \{A\})}$).
\end{definition}

For each family, an irreducible family can be obtained by removing all
expressible sets, one by one until there are no more such sets. This
procedure is guaranteed to terminate for finite sets.

\begin{proposition}
  \label{prop:ex_indep}
  Each family $F$ has an irreducible subfamily $F'$ such that
  $\closure{F} = \closure{F'}$.
\end{proposition}

The following interesting (and non-trivial) lemma, proved in
\cite{kim} and formally proved in the Appendix, shows that for all
families having a same closure there is a unique irreducible family,
and the previous procedure will always yield the same final answer in
whatever order the sets are removed.

\begin{lemma}\label{lemma:unique_irreducible} 
  If $F$ and $F'$ are irreducible families and
  $\closure{F} = \closure{F'}$, then $F = F'$.
\end{lemma}

\subsection{Covering}
Total FC characterization of all families in $\nfam{n}$ is done by
defining two collections $\FC$ and $\nonFC$ such that all families in
$\FC$ are FC and that all families in $\nonFC$ are nonFC, and such
that the status of each given $\nfam{n}$ family can easily be
determined by an element of $\FC$ or $\nonFC$ (we say that the
given $\nfam{n}$ family is \emph{covered} by $\FC$ and $\nonFC$). The
following definition formalizes the notion of covering and relies on
Proposition \ref{prop:FC_family_mono} and Proposition
\ref{prop:FC_family_empty}, Proposition \ref{prop:FC_family_closure},
and some trivial properties of isomorphic families.

\begin{definition}\ \\[-1em]
  \begin{enumerate}
  \item A family $F$ is \emph{FC-covered} by a family $F_c$ (denoted
    by $\FCcovered{F}{F_c}$) if there exists $F_c'$ such that
    $\iso{F_c'}{F_c}$ and $\closure{F} \supseteq F_c'$.  A family $F$
    is \emph{FC-covered} by a collection of families $\FC$ (denoted by
    $\FCcovered{F}{\FC}$) if there is an $F_c \in \FC$ such that $F$
    is FC-covered by $F_c$ (i.e.,
    $\exists\ F_c \in \FC.\ \FCcovered{F}{F_c}$). A collection of
    families $\mathcal{F}$ is \emph{FC-covered} by a collection of
    families $\FC$ (denoted by $\FCcovered{\mathcal{F}}{\FC}$) if all
    families $F\in \mathcal{F}$ are covered by $\FC$ (i.e.,
    $\forall F\in\mathcal{F}.\ \FCcovered{F}{\FC}$).
  \item A family $F$ is \emph{nonFC-covered} by $N_c$ (denoted by
    $\nonFCcovered{F}{N_c}$) if there is an $N_c'$ such that
    $\iso{N_c'}{N_c}$ and
    $\closure{F} \subseteq \closure{N_c'} \union \{\emptyset\}$.  A
    family $F$ is \emph{nonFC-covered} by a collection of families
    $\nonFC$ (denoted by $\nonFCcovered{F}{\nonFC}$) if there is an
    $N_c \in \nonFC$ such that $F$ is nonFC-covered by $N_c$ (i.e.,
    $\exists\ N_c \in \nonFC.\ \nonFCcovered{F}{N_c}$). A collection
    of families $\mathcal{F}$ is \emph{nonFC-covered} by a collection
    of families $\nonFC$ (denoted by
    $\nonFCcovered{\mathcal{F}}{\nonFC}$) if all families
    $F\in \mathcal{F}$ are covered by $\nonFC$ (i.e.,
    $\forall F\in\mathcal{F}.\ \nonFCcovered{F}{\nonFC}$).
  \item A family $F$ is \emph{covered} by $\FC$ and $\nonFC$ (denoted
    by $\covered{F}{\FC}{\nonFC}$) if it is FC-covered by $\FC$ or it
    is nonFC-covered by $\nonFC$. A collection of families
    $\mathcal{F}$ is \emph{covered} by $\FC$ and $\nonFC$ (denoted by
    $\covered{\mathcal{F}}{\FC}{\nonFC}$) if all its families are
    FC-covered by $\FC$ or nonFC-covered by $\nonFC$.
  \end{enumerate}
\end{definition}

The next lemma (proved in the Appendix) shows that our notion of
covering guarantees that FC-covered families are FC and that
nonFC-covered families are not FC.

\begin{lemma}\ \\[-5mm]
  \label{lemma:FCcoveredFC}
  \begin{enumerate}
  \item Any family $F$ that is FC-covered by an FC-family $F_c$ is an
    FC-family.
  \item Any family $F$ that is nonFC-covered by a nonFC-family $N_c$
    is not an FC-family.
  \end{enumerate}
\end{lemma}

An important aspect of our definition of the notion of covering is
that its condition can be checked easily. First, there is only a
relatively small number of possible families from $\nfam{n}$ that are
isomorphic to a given family in $\nfam{n}$ (for example, for
$\nfam{6}$, these are generated by $720$ permutations of the domain).
As the closure of a family can be easily effectively computed, it
remains to check only if it contains one of these isomorphs, and this
can be performed easily.

The following proposition gives some other easy consequences of the
covering definition.

\begin{proposition}\ \\[-1em]
  \label{prop:covered}
  \begin{enumerate}
  \item Let $F \subseteq F'$. If $\FCcovered{F}{\FC}$ then
    $\FCcovered{F'}{\FC}$. If $\nonFCcovered{F'}{\nonFC}$ then
    $\nonFCcovered{F}{\nonFC}$.
  \item If $\iso{F}{F'}$ and $\FCcovered{F}{F_c}$ then
    $\FCcovered{F'}{F_c}$. If $\iso{F}{F'}$ and
    $\nonFCcovered{F}{N_c}$ then $\nonFCcovered{F'}{N_c}$.
  \item If $\FCcovered{F - \{\emptyset\}}{F_c}$, then
    $\FCcovered{F}{F_c}$.  If $\nonFCcovered{F - \{\emptyset\}}{N_c}$,
    then $\nonFCcovered{F}{N_c}$. If
    $\covered{F-\{\emptyset\}}{\FC}{\nonFC}$, then
    $\covered{F}{\FC}{\nonFC}$.
  \end{enumerate}
\end{proposition}

Since covering is preserved by isomorphisms, to show that a collection
is covered it suffices to show that its iso-base is covered, as
shown by the following lemma proved in Appendix.
\begin{lemma}
  \label{lemma:generates_covered}
  Assume that $\mathcal{F}_b$ iso-represents $\mathcal{F}$. If
  $\FCcovered{\mathcal{F}_b}{\FC}$, then
  $\FCcovered{\mathcal{F}}{\FC}$. If
  $\nonFCcovered{\mathcal{F}_b}{\nonFC}$, then
  $\nonFCcovered{\mathcal{F}}{\nonFC}$. If
  $\covered{\mathcal{F}_b}{\FC}{\nonFC}$, then
  $\covered{\mathcal{F}}{\FC}{\nonFC}$.
\end{lemma}

Similarly, if two families have the same closure, even if they are not
equivalent, one is covered iff the other one is.
\begin{proposition}
  \label{prop:covered_closure}
  If $\closure{F} = \closure{F'}$ then $F$ is covered by $\FC$ and
  $\nonFC$ iff $F'$ is.
\end{proposition}

Therefore, the following lemma, proved in Appendix, reduces the
problem of checking all families in $\nfam{n}$ to checking just the
irreducible ones.

\begin{lemma}
  \label{lemma:indep}
  If all irreducible families in $\nfam{n}$ are covered by $\FC$ and
  $\nonFC$, then all families in $\nfam{n}$ are covered by $\FC$ and
  $\nonFC$.
\end{lemma}

\subsection{Minimal FC-familes and maximal nonFC-families}
\label{sec:minimalFC}

We want to have as few as possible characteristic families, so we want
all our characteristic families to be extreme in some sense.

\begin{definition}\ \\[-1em]
  \begin{enumerate}
  \item An FC-family is \emph{minimal} if it is irreducible and
    removing each of its sets yields a nonFC-family.
  \item A nonFC-family is \emph{maximal} if it is union-closed and
    every new set addeded yields an FC-family.
  \end{enumerate}
\end{definition}

It can be easily shown that minimal and maximal families are exactly
those that cannot be covered by smaller or larger families.

\begin{proposition}\ \\[-1em]
  \begin{enumerate}
  \item A FC-family is minimal iff it is not FC-covered by any other
    family.
  \item A nonFC-family is maximal iff it is not nonFC-covered by any
    other family.
  \end{enumerate}
\end{proposition}

\section{Enumerating families}
\label{sec:enum}

Next we describe efficient generic procedures for enumerating all
families with certain properties. Note that all concepts in this
section are generic and can be used in a wider context than checking
the FC-status.

\subsection{L-partitioning}
\label{sec:lpart}

In this section we develop efficient methods to enumerate all families
in $\nfam{n}$ that have certain properties. As it is usually the case,
an inductive construction gives good results. Larger families, can be
obtained from the smaller ones, by adding new sets. A good attribute
of a family that can be used to control the inductive construction is
the number of its members of each cardinality.

\begin{definition}
  Let $L$ be the list $[l_0, l_1, \ldots, l_{m}]$. A family $F$ is
  \emph{$L$-partitioned} if it consists of $l_0$ empty sets, $l_1$
  sets with 1 element, \ldots, and $l_{m}$ sets with $m$ elements,
  (i.e.,
  $(\forall A \in F.\ \card{A} \le m) \ \And\ (\forall i.\ 0 \leq i
  \leq m \longrightarrow \card{\{A \in F.\ \card{A} = i\}} = l_i)$).
\end{definition}

\begin{example}
  The family $\{$ $\{\}$, $\{0,1,2\}$, $\{0,1,2,3\}$, $\{0,1,2,4,5\}$,
  $\{0,1,3,4\}$, $\{0,1,3,5\}$, $\{0,2,3,4,5\}$, $\{0,3,4,5\}$ $\}$,
  is $[1, 0, 0, 1, 4, 2]$-partitioned, since it contains the empty
  set, one 3-element set, four 4-element sets and two 5-element sets.
\end{example}

The number of possible members of each cardinality in a family is
bounded.
\begin{proposition}
  \label{prop:lpart}
  If a family $F \in \nfam{n}$ is
  $[l_0, l_1, \ldots, l_{m}]$-partitioned, then for all $i \le n$,
  $l_i \le \binom{n}{i}$, and for all $i > n$, $l_i = 0$.
\end{proposition}

There is a natural partial order between lists ($\preceq$) that
corresponds to subfamily ($\subseteq$) relation of $L$-partitioned
families.
\begin{definition}
  A list $L' = [l_0', \ldots, l_m']$ is \emph{pointwise less or equal}
  to the list $L = [l_0, \ldots, l_m]$ (denoted by $L' \preceq L$)
  if for all $0 \le i \le m$, it holds that $l_i' \le l_i$.
\end{definition}

The relation $\preceq$ is a partial order (reflexive, antisymmetric,
and transitive), tightly connected with subfamilies (e.g., if $F$ and
$F'$ are $L$ and $L'$ partitioned families, then $F' \subseteq F$ implies
$L' \preceq L$).

We are often interested in enumerating all families of $\nfam{n}$
that are $L$-partitioned for some given list $L$ and that satisfy some
given property $P$.

\begin{definition} For a given list $L$, a number $n$, and a property
  $P$, the \emph{collection of all $L$-partitioned families of
    $\nfam{n}$ satisfying $P$} is denoted by $\LnP{L}{n}{P}$.
\end{definition}

Our goal is to define an inductive procedure for enumerating all
elements in $\LnP{L}{n}{P}$. It will rely on the fact that each family
in $\nfam{n}$ that is $[l_0, \ldots, l_m + 1]$ partitioned is of the
form $F \union \{A\}$ where $F$ is $[l_0, \ldots, l_m]$-partitioned,
$\card{A} = m$ and $A \subseteq \setn{n}$. Note that any other
position (before $m$) could be used, but the last position has some
nice properties that we shall exploit.

Instead of checking whether $F \union \{A\}$ has the property $P$, for
efficiency reasons we use the incremental approach and introduce
another predicate $\Pinc{P}$, checking relationship between $F$ and
$A$ that guarantees that $F \union \{A\}$ will have the property
$P$. Note that in our inductive construction we always extend a
family $F$ by a set $A$ that is not already contained in family, and
that has a greater or equal cardinality than all the family members
(as we always choose the last position $m$ for induction). The
following definition (and the corresponding incremental predicates)
will use this condition.

\begin{definition}
  Predicate $\Pinc{P}$ \emph{incrementally checks} predicate $P$ if for every
  family $F$ and a set $A$ such that $\forall A' \in F.\ \card{A} \ge
  \card{A'}$ and $A \notin F$, it holds
\begin{equation}
\label{eq:PQ}
P\ (F \union \{A\}) \iff P\ F\ \And\ \Pinc{P}\ F\ A.
\end{equation}
\end{definition}

\begin{example}
  If we know that all elements of a family $F$ have less than $k$
  elements, to check if all elements of $F \union \{A\}$ have less
  than $k$ elements it suffices to check only if $A$ has less than $k$
  elements. Therefore, the predicate $\Pinc{P} = \lambda\ F\
  A.\ \card{A} \leq k$ incrementally checks the predicate $P = \lambda
  F.\ \forall A \in F.\ \card{A} \leq k$.
\end{example}

The following construction extends all families in a collection by a
given set $A$, filtering out families that already contain $A$ and
families that do not satisfy the given incremental predicate
$\Pinc{P}$.

\begin{definition}
  \emph{$\Pinc{P}$-filtered product of a collection $\mathcal{F}$ and
    a set $A$} is defined by:
  $$\mult{\mathcal{F}}{A}{\Pinc{P}} = \{F \union \{A\}.\ F \in \mathcal{F} \And A \notin F \And \Pinc{P}\ F\ A\}$$

  \emph{$\Pinc{P}$-filtered product of a collection $\mathcal{F}$ and
    a family $F$} is defined by
  $$\mult{\mathcal{F}}{F}{\Pinc{P}} = \{\mult{\mathcal{F}}{A}{\Pinc{P}}.\ A \in F\}$$
\end{definition}

\begin{definition}
  $\families{n}{m}$ denotes a collection of all $A \subseteq \setn{n}$
  such that $\card{A} = m$.
\end{definition}

The following theorem is the basis for an inductive construction of
$\LnP{L}{n}{P}$ (its proof is given in the Appendix).

\begin{theorem}
  \label{thm:mult}
  Assume that the predicate $\Pinc{P}$ incrementally checks $P$. Then
  $$\LnP{[l_0, \ldots, l_m+1]}{n}{P} = \mult{\LnP{[l_0, \ldots, l_m]}{n}{P}}{\families{n}{m}}{\Pinc{P}}.$$
\end{theorem}

In many cases instead of enumerating whole $\LnP{L}{n}{P}$ it suffices
to enumerate its iso-representing subcollection. Again, an inductive
construction can be used.

\begin{definition}
  Predicate $\Pinc{P}$ is \emph{preserved by injective functions}, if
  for all functions $f$ injective on $\Union F \union A$, if
  $\Pinc{P}\ F\ A$ holds, then $\Pinc{P}\ (\mapfam{f}{F})\
  (\mapset{f}{A})$ also holds.
\end{definition}

\begin{theorem}
  \label{thm:mult_generates}
  Assume that the predicate $\Pinc{P}$ incrementally checks $P$ and is
  preserved by injective functions. If $m \le n$ and $\mathcal{F}_b$
  is an iso-representing subcollection of
  $\LnP{[l_0,\ldots, l_m]}{n}{P}$, then
  $\mathcal{F}'_b \equiv
  \mult{\mathcal{F}_b}{\families{n}{m}}{\Pinc{P}}$
  is an iso-representing subcollection of
  $\LnP{[l_0, \ldots, l_m + 1]}{n}{P}$.
\end{theorem}

The Theorem \ref{thm:mult_generates} yields an iso-representing
collection (not necessarily an iso-base), so the algorithm for finding
an iso-base (described in Section \ref{sec:iso}) should be applied, if
we are interested to find an iso-base (and it is usually better to
work with an iso-base since this removes redundancies).

\subsection{Simple recursive enumeration}
For a specific list $L$, Theorem \ref{thm:mult_generates} inspires the
following recursive procedure.

\texttt{
\begin{tabbing}
\hspace{0.5cm}\=\hspace{1cm}\=\hspace{1cm}\=\hspace{1cm}\=\kill
function $\enumrecname$ where\\
\>"$\enumrec{L}{v_{\emptylist}}{upd} =$\\
\>\>(if $L = \emptylist$ then $v_{\emptylist}$\\
\>\>\ else if $\last{L} = 0$ then $\enumrec{(\butlast{L})}{v_{\emptylist}}{upd}$\\
\>\>\ else $upd$ $(\enumrec{(\declast{L})}{v_{\emptylist}}{upd})$ $L$ )"\\
\end{tabbing}
}

The recursion goes trough a sequence of lists, decreasing the last
element if it is not zero and removing it otherwise. For example, if
called for a list $[1, 2, 2]$, it would make a sequence of recursive
calls for the lists $[1, 2, 1]$, $[1, 2, 0]$, $[1, 2]$, $[1, 1]$,
$[1, 0]$, $[1]$, $[0]$, and $\emptylist$. The function returns the
value corresponding to the list given as its input parameter. Values
corresponding to each list are reconstructed backwards (in the return
of the recursion call). In the given example, it would generate
families with one empty set $[1]$, then one empty and one singleton
set $[1,1]$, then one empty and two singletons $[1, 2]$, then one
empty, two singletons and one doubleton $[1, 2, 1]$, and finally the
required families with one empty, two singletons, and two doubletons
$[1, 2, 2]$. The fixed parameter $v_{\emptylist}$ is the value for the
empty list (the base case of the recursion), and the other fixed
parameter $upd$ is the function that is used to update the value
corresponding to the next list in the sequence (the return value of
the recursive call) and to obtain the value corresponding to the
current list (the return value of the current call). When updating the
value, the function $upd$ can also take the current list into account.

The following theorem shows how can we use $\enumrecname$ to find an
iso-base of $\LnP{{L}}{n}{P}$ for some given list $L$, number $n$, and
a predicate $P$ (its proof is outlined in the Appendix).

\begin{theorem}
  \label{lemma:enum_rec_mult}
  Let $L$ be a list such that $\length{L} \leq n + 1$.  Assume
  that:
  \begin{enumerate}
  \item $P\ \{\}$ holds and $v_{\emptylist} = \{\emptyset\}$,
  \item $\Pinc{P}$ incrementally checks $P$ and is preserved by injective
  functions,
  \item $\mathcal{P}$ contains all permutations of $\listn{n}$, and
  $upd = \lambda\ \mathcal{F}\ L.\
  \base{(\mult{\mathcal{F}}{\families{n}{\length{L}-1}}{\Pinc{P}})}{\mathcal{P}}$.
  \end{enumerate}
  Then $\enumrec{L}{v_{\emptylist}}{upd}$ is an iso-base of
  $\LnP{{L}}{n}{P}$.
\end{theorem}

\subsection{Dynamic programming enumeration}
Theorem \ref{lemma:enum_rec_mult} gives us a way to compute an
iso-base of $\LnP{L}{n}{P}$ for a single given list $L$. However, we
often want to enumerate iso-bases of $\LnP{{L}}{n}{P}$ for many
different lists $L$ (in an extreme case, for all possible lists
$L$). In that case, a much better solution can be obtained by using
dynamic programming, as many subproblems will overlap. For example,
calculating both $\enumrec{[1, 2, 3]}{v_{\emptylist}}{upd}$ and
$\enumrec{[1, 2, 2, 1]}{v_{\emptylist}}{upd}$ will require calculation
of $\enumrec{[1, 2, 2]}{v_{\emptylist}}{upd}$.

Next we present a dynamic programming algorithm that traverses all
lists $L$ such that $L \preceq L_{max}$ for a given list $L_{max}$ and
gathers a list of values corresponding to each of those lists. The
uniform upper bound (given by the list $L_{max}$) allows easy
termination proof, but in some cases we want to be able to terminate
the traversal of some branches earlier (before the bound of $L_{max}$
is reached). 

For example, when examining different $L$-partitions of the collection
$\nfam{n}$ it can be noticed that there are some partitions that
contain only nonFC-families, some that contain both FC and
nonFC-families and some that contain only FC-families. From
Proposition \ref{prop:covered} it follows that this property is
monotonic over the relation $\preceq$.

\begin{proposition}\ \\[-1em]
  \label{prop:covered_mono}
  \begin{enumerate}
  \item If for all $F$ in $\nfam{n}$ that are $L$-partitioned it holds
    $\FCcovered{F}{\FC}$, and if $L' \succeq L$, then for all $F$ in
    $\nfam{n}$ that are $L'$-partitioned it holds that
    $\FCcovered{F}{\FC}$.
  \item If for all $F$ in $\nfam{n}$ that are $L$-partitioned it holds
    that $\nonFCcovered{F}{\nonFC}$, and if $L' \preceq L$, then for
    all $F$ in $\nfam{n}$ that are $L'$-partitioned it holds that
    $\nonFCcovered{F}{\nonFC}$.
  \end{enumerate}
\end{proposition}

When showing that all families in $\nfam{n}$ are covered by $\FC$ and
$\nonFC$, a good approach is to identify minimal lists $L_{F}$ such
that all $L_F$-partitioned families in $\nfam{n}$ are FC-covered by
$\FC$ and maximal lists $L_{N}$ such that all $L_N$-partitioned
families in $\nfam{n}$ are nonFC-partitioned. Then it remains only to
show that all $L$-partitioned families in $\nfam{n}$ are covered for
lists $L$ that are between all lists $L_N$ and $L_F$. However, since
our inductive construction must start from the empty family, we
respect only the upper bound and show the result for lists $L$ that
are below all lists in $L_F$.

Therefore, in our algorithm it is also allowed to exclude all lists
$L'$ such that $L' \succeq L$ for lists $L$ for which the given
predicate $stop\ L$ holds. All lists encountered during the traversed
will have the same length (unlike the lists traversed by the
$\enumrecname$ function, here the shorter lists will be padded by
zeros). All successive lists during the traversal must differ only by
one on the last non-zero entry. If $m$ is an index such that all
elements after the position $m$ in the list $L$ are zero, then the
list $L$ will be used to reach lists obtained by incrementing values
of list $L$ on positions greater or equal to $m$. For example, the
list $[3, 0, 1, 0, 0]$ will be used to reach the lists
$[3, 0, 2, 0, 0]$, $[3, 0, 1, 1, 0]$ and $[3, 0, 1, 0, 1]$.

The function $\enumdppname$ has three fixed, and four changing
parameters. First two changing parameters are the current list $L$ and
the value $v$ corresponding to it. Next parameter $m$ is an index of
last non-zero entry in $L$, and the final parameter $res$ is the
accumulating parameter that stores the result (the list of all values
corresponding to lists previously encountered during the
traversal). First two fixed parameters are criteria for traversal
termination (the list $L_{max}$ and the predicate $stop$), and the
final fixed parameter $upd$ is the function used to calculate the
value corresponding to a next lists encountered during the traversal,
based on the value corresponding to the current list. If a list is
obtained from the list $L$ by incrementing the value on the position
$m$, its corresponding value is calculated by $upd\ v\ m$, where $v$
is the value corresponding to the list $L$.

\texttt{
\begin{tabbing}
  \hspace{0.5cm}\=\hspace{1cm}\=\hspace{1cm}\=\hspace{1cm}\=\kill
  function $\enumdppname$ where \\
  \>"$\enumdpp{L}{m}{v}{res}{upd}{stop}{L_{max}}$ = foldl\\
  \>\>($\lambda$ $res'$ $m'$. let $L'$ = $\incnth{L}{m'}$ in\\
  \>\>\>if $stop\ L'$ $\vee$ $L' \succ L_{max}$ then \\
  \>\>\>\>$res'$ \\
  \>\>\>else \\
  \>\>\>\>$\enumdpp{L'}{m'}{(upd\ v\ m')}{res'}{upd}{stop}{L_{max}}$\\
  \>\>)\\
  \>\>$(v \# res)$\\
  \>\>$[m, m+1, \ldots, \length{L} - 1 ]$"
\end{tabbing}
}

The traversal usually starts from the empty list (appropriately padded
by zeros), the value $v_{\emptylist}$ corresponding to it, the value
$m=0$, and the empty accumulating parameter $res$. The wrapper
function $\enumdpname$ performs such initial function call for
$\enumdppname$.

\texttt{
\begin{tabbing}
  \hspace{0.5cm}\=\hspace{1cm}\=\hspace{1cm}\=\hspace{1cm}\=\kill
  definition $\enumdpname$ where\\
  \>"$\enumdp{v_{\emptylist}}{upd}{stop}{L_{max}} = \enumdpp{[0, ..., 0]}{0}{v_{\emptylist}}{\emptylist}{upd}{stop}{L_{max}}$"
\end{tabbing}
}

Finally, we can use the dynamic programming enumeration to collect
iso-bases of $\LnP{L}{n}{P}$ for all list $L \preceq L_{max}$ that
were not excluded by the given $stop$ predicate.

\begin{theorem}
  \label{lemma:enum_dp_mult}
  Assume that
  \begin{enumerate}
  \item $\Pinc{P}$ incrementally checks $P$ and is preserved by injective
    functions, 
  \item $P\ \{\}$ holds and $v_{\emptylist} = \{\emptyset\}$,
  \item $L_{max}$ is a list such that $\length{L_{max}} \le n + 1$,
    $\mathcal{L}_{s}$ contain lists (that all have the same length
    $\length{L_{max}}$), and $stop = \lambda\ L.\ (\exists L_s \in
    \mathcal{L}_s.\ L \succeq L_s)$,
  \item $\mathcal{P}$ contains all permutations of $\listn{n}$, $upd =
    \lambda\ \mathcal{F}\ m.\
    \base{\mult{(\mathcal{F}}{\families{n}{m}}{\Pinc{P}})}{\mathcal{P}}$.
  \end{enumerate}
  Let
  $X = \{L.\ L \preceq L_{max} \And (\nexists L_s \in
  \mathcal{L}_{s}.\ L \succeq L_s)\}$.
  Then, for all $L \in X$, there exists an
  $\mathcal{F}_b \in \enumdp{v_{\emptylist}}{upd}{stop}{L_{max}}$ such
  that $\mathcal{F}_b$ is an iso-base of $\LnP{L}{n}{P}$.
\end{theorem}

\subsection{Finding characteristic families}
Next, we describe a fully automated procedure that finds all
characteristic families --- all minimal FC-families and maximal
nonFC-families that are canonical wrt. lexicographic ordering of
families. Note again that this procedure needs not to be verified (we
implemented it outside Isabelle/HOL). The procedure is based on the
dynamic programing enumeration. During the enumeration, a collection
$\FC$ of canonical, minimal FC-families, and a collection of canonical
nonFC-families $\nonFC'$ discovered so far is maintained (both are
empty in the beginning). For each $L$ encountered during enumeration,
a list of all canonical, irreducible, $L$-partitioned families that
are not FC-covered by any family in $\FC$ is calculated (the
enumeration uses the predicate
$P = \lambda F.\ \indep{F}\ \And\ \neg (\FCcovered{F}{\FC})$ and
$\Pinc{P} = \lambda\ F\ A.\ \neg\ \expressible{A}{F} \And\ \neg
(\FCcovered{F \union \{A\}}{\FC})$
--- $\Pinc{P}$ incrementally checks $P$ and is preserved by injective
functions). The FC-status of each family in that list is examined. All
newly discovered FC-families are added to $\FC$, nonFC-families to
$\nonFC'$, and the procedure continues with the next list $L$.

Since the enumeration of lists $L$ is in the lexicographic order, all
discovered FC-families will be minimal (an $L$-partitioned family can
be covered only by a $L'$-partitioned familiy only if $L \succeq L'$,
and if both $L$ and $L'$ are the same length, then $L'$ must be
lexicographically smaller than $L$).

At the end, the collection $\nonFC'$ will be an iso-base of all
irreducible nonFC-families in $\nfam{n}$. The collection $\nonFC$ of
all irreducible, cannonical, maximal nonFC-families is obtained by
filtering all families in $\nonFC'$ that are covered by some other
family in $\nonFC'$ and calculating union-closures.

\section{FC(6) families}
\label{sec:fc6}

Applying previous procedure on $\nfam{6}$ led us to the following
definitions (these are the characteristic families for which we shall
prove that they cover all families in $\nfam{6}$).

\begin{definition}
$\FCsix = $

\begin{footnotesize}
\begin{longtable}{|p{\textwidth}|}
\hline
\{\{0,1,2,3\},\{0,1,2,4\},\{0,1,3,4\},\{0,2,3,4\},\{0,1,2,3,5\},\{1,2,3,4,5\}\}\\\hline
\{\{0,1,2,3\},\{0,1,2,4\},\{0,1,3,4\},\{0,1,2,5\},\{0,1,3,4,5\},\{0,2,3,4,5\},\{1,2,3,4,5\}\}\\\hline
\{\{0,1,2,3\},\{0,1,2,4\},\{0,1,3,4\},\{0,2,3,4\},\{1,2,3,4\}\}\\\hline
\{\{0,1,2,3\},\{0,1,2,4\},\{0,1,3,4\},\{0,2,3,4\},\{1,2,3,5\}\}\\\hline
\{\{0,1,2,3\},\{0,1,2,4\},\{0,1,3,4\},\{0,1,2,5\},\{0,2,3,5\},\{1,2,3,4,5\}\}\\\hline
\{\{0,1,2,3\},\{0,1,2,4\},\{0,1,3,4\},\{0,2,3,5\},\{0,2,4,5\},\{1,2,3,4,5\}\}\\\hline
\{\{0,1,2,3\},\{0,1,2,4\},\{0,1,3,4\},\{0,1,2,5\},\{0,3,4,5\},\{1,2,3,4,5\}\}\\\hline
\{\{0,1,2,3\},\{0,1,2,4\},\{0,1,3,4\},\{0,2,3,5\},\{1,2,3,5\},\{0,1,2,4,5\}\}\\\hline
\{\{0,1,2,3\},\{0,1,2,4\},\{0,1,3,4\},\{0,2,3,5\},\{1,2,4,5\},\{0,2,3,4,5\}\}\\\hline
\{\{0,1,2,3\},\{0,1,2,4\},\{0,1,3,4\},\{0,2,3,5\},\{2,3,4,5\},\{1,2,3,4,5\}\}\\\hline
\{\{0,1,2,3\},\{0,1,2,4\},\{0,1,3,5\},\{0,2,4,5\},\{1,3,4,5\},\{0,2,3,4,5\}\}\\\hline
\{\{0,1,2,3\},\{0,1,2,4\},\{0,1,3,4\},\{0,2,3,5\},\{2,3,4,5\},\{0,1,2,4,5\}\}\\\hline
\{\{0,1,2,3\},\{0,1,2,4\},\{0,1,2,5\},\{0,3,4,5\},\{1,3,4,5\},\{0,2,3,4,5\}\}\\\hline
\{\{0,1,2,3\},\{0,1,2,4\},\{0,1,3,5\},\{0,1,4,5\},\{2,3,4,5\},\{0,2,3,4,5\}\}\\\hline
\{\{0,1,2,3\},\{0,1,2,4\},\{0,1,3,5\},\{0,2,4,5\},\{0,3,4,5\},\{1,2,3,4,5\}\}\\\hline
\{\{0,1,2,3\},\{0,1,2,4\},\{0,1,3,4\},\{0,1,2,5\},\{2,3,4,5\},\{0,2,3,4,5\},\{1,2,3,4,5\}\}\\\hline
\{\{0,1,2,3\},\{0,1,2,4\},\{0,1,3,4\},\{0,1,2,5\},\{2,3,4,5\},\{0,1,3,4,5\},\{0,2,3,4,5\}\}\\\hline
\{\{0,1,2,3\},\{0,1,2,4\},\{0,1,3,4\},\{0,2,3,4\},\{0,1,2,5\},\{0,1,3,5\}\}\\\hline
\{\{0,1,2,3\},\{0,1,2,4\},\{0,1,3,4\},\{0,1,2,5\},\{0,1,3,5\},\{2,3,4,5\}\}\\\hline
\{\{0,1,2,3\},\{0,1,2,4\},\{0,1,3,4\},\{0,2,3,4\},\{0,1,2,5\},\{0,3,4,5\}\}\\\hline
\{\{0,1,2,3\},\{0,1,2,4\},\{0,1,3,4\},\{0,1,2,5\},\{0,2,3,5\},\{1,2,4,5\}\}\\\hline
\{\{0,1,2,3\},\{0,1,2,4\},\{0,1,3,4\},\{0,1,2,5\},\{0,1,3,5\},\{0,1,4,5\},\{0,2,3,4,5\}\}\\\hline
\{\{0,1,2,3\},\{0,1,2,4\},\{0,1,3,4\},\{0,1,2,5\},\{0,1,3,5\},\{0,2,4,5\},\{0,2,3,4,5\}\}\\\hline
\{\{0,1,2\},\{0,1,3,4\},\{0,2,3,5\},\{0,1,3,4,5\},\{1,2,3,4,5\}\}\\\hline
\{\{0,1,2\},\{0,1,3,4\},\{0,2,3,5\},\{0,1,2,4,5\},\{1,2,3,4,5\}\}\\\hline
\{\{0,1,2\},\{0,1,3,4\},\{0,2,3,4\},\{0,1,2,3,5\},\{1,2,3,4,5\}\}\\\hline
\{\{0,1,2\},\{0,1,3,4\},\{0,1,3,5\},\{0,2,3,4,5\},\{1,2,3,4,5\}\}\\\hline
\{\{0,1,2\},\{0,1,3,4\},\{0,3,4,5\},\{0,1,2,3,5\},\{0,1,2,4,5\},\{1,2,3,4,5\}\}\\\hline
\{\{0,1,2\},\{0,1,3,4\},\{2,3,4,5\},\{0,1,2,3,5\},\{0,1,2,4,5\},\{0,2,3,4,5\},\{1,2,3,4,5\}\}\\\hline
\{\{0,1,2\},\{0,3,4,5\},\{1,3,4,5\},\{0,1,2,3,4\},\{0,1,2,3,5\},\{0,1,2,4,5\},\{0,2,3,4,5\},\{1,2,3,4,5\}\}\\\hline
\{\{0,1,2\},\{0,1,2,3\},\{0,1,3,4\},\{0,1,2,3,5\},\{0,1,2,4,5\},\{0,1,3,4,5\},\{0,2,3,4,5\},\{1,2,3,4,5\}\}\\\hline
\{\{0,1,2\},\{0,1,3,4\},\{0,2,3,4\},\{1,2,3,4\}\}\\\hline
\{\{0,1,2\},\{0,1,3,4\},\{0,2,3,4\},\{1,2,3,5\}\}\\\hline
\{\{0,1,2\},\{0,1,3,4\},\{0,2,3,5\},\{1,2,4,5\}\}\\\hline
\{\{0,1,2\},\{0,1,3,4\},\{0,2,3,5\},\{1,3,4,5\}\}\\\hline
\{\{0,1,2\},\{0,1,2,3\},\{0,1,3,4\},\{0,2,3,4\},\{1,2,3,4,5\}\}\\\hline
\{\{0,1,2\},\{0,1,3,4\},\{0,1,3,5\},\{0,2,4,5\},\{0,2,3,4,5\}\}\\\hline
\{\{0,1,2\},\{0,1,3,4\},\{0,2,3,4\},\{1,3,4,5\},\{0,1,2,3,5\}\}\\\hline
\{\{0,1,2\},\{0,1,3,4\},\{0,2,3,4\},\{1,3,4,5\},\{1,2,3,4,5\}\}\\\hline
\{\{0,1,2\},\{0,1,3,4\},\{0,2,3,4\},\{0,3,4,5\},\{1,2,3,4,5\}\}\\\hline
\{\{0,1,2\},\{0,1,3,4\},\{0,1,3,5\},\{2,3,4,5\},\{0,1,2,4,5\}\}\\\hline
\{\{0,1,2\},\{0,1,3,4\},\{0,1,3,5\},\{2,3,4,5\},\{0,2,3,4,5\}\}\\\hline
\{\{0,1,2\},\{0,1,3,4\},\{0,1,3,5\},\{0,1,4,5\},\{0,2,3,4,5\}\}\\\hline
\{\{0,1,2\},\{0,1,2,3\},\{0,1,4,5\},\{0,2,4,5\},\{1,2,3,4,5\}\}\\\hline
\{\{0,1,2\},\{0,1,2,3\},\{0,1,3,4\},\{0,2,3,5\},\{1,2,3,4,5\}\}\\\hline
\{\{0,1,2\},\{0,1,2,3\},\{0,1,3,4\},\{0,2,4,5\},\{1,2,3,4,5\}\}\\\hline
\{\{0,1,2\},\{0,1,3,4\},\{0,2,3,4\},\{0,1,3,5\},\{0,1,2,4,5\}\}\\\hline
\{\{0,1,2\},\{0,1,3,4\},\{0,1,3,5\},\{0,3,4,5\},\{1,2,3,4,5\}\}\\\hline
\{\{0,1,2\},\{0,1,2,3\},\{0,1,3,4\},\{0,3,4,5\},\{0,1,2,4,5\},\{1,2,3,4,5\}\}\\\hline
\{\{0,1,2\},\{0,1,3,4\},\{0,3,4,5\},\{2,3,4,5\},\{0,1,2,3,5\},\{0,1,2,4,5\}\}\\\hline
\{\{0,1,2\},\{0,1,3,4\},\{0,3,4,5\},\{1,3,4,5\},\{0,1,2,3,5\},\{0,1,2,4,5\}\}\\\hline
\{\{0,1,2\},\{0,1,3,4\},\{0,3,4,5\},\{1,3,4,5\},\{0,1,2,3,5\},\{0,2,3,4,5\}\}\\\hline
\{\{0,1,2\},\{0,1,3,4\},\{0,3,4,5\},\{2,3,4,5\},\{0,1,2,3,5\},\{1,2,3,4,5\}\}\\\hline
\{\{0,1,2\},\{0,1,2,3\},\{0,1,3,4\},\{0,3,4,5\},\{0,1,2,3,5\},\{0,2,3,4,5\},\{1,2,3,4,5\}\}\\\hline
\{\{0,1,2\},\{0,1,2,3\},\{0,1,4,5\},\{0,3,4,5\},\{0,1,2,3,4\},\{0,2,3,4,5\},\{1,2,3,4,5\}\}\\\hline
\{\{0,1,2\},\{0,1,2,3\},\{0,1,3,4\},\{2,3,4,5\},\{0,1,2,4,5\},\{0,2,3,4,5\},\{1,2,3,4,5\}\}\\\hline
\{\{0,1,2\},\{0,1,2,3\},\{0,1,2,4\},\{0,1,3,4\},\{0,1,2,3,5\},\{0,1,2,4,5\},\{0,2,3,4,5\},\{1,2,3,4,5\}\}\\\hline
\{\{0,1,2\},\{0,1,2,3\},\{0,1,2,4\},\{0,1,3,4\},\{0,1,2,3,5\},\{0,1,3,4,5\},\{0,2,3,4,5\},\{1,2,3,4,5\}\}\\\hline
\{\{0,1,2\},\{0,1,2,3\},\{0,1,3,4\},\{2,3,4,5\},\{0,1,2,3,5\},\{0,1,3,4,5\},\{0,2,3,4,5\},\{1,2,3,4,5\}\}\\\hline
\{\{0,1,2\},\{0,1,2,3\},\{0,1,4,5\},\{2,3,4,5\},\{0,1,2,3,4\},\{0,1,3,4,5\},\{0,2,3,4,5\},\{1,2,3,4,5\}\}\\\hline
\{\{0,1,2\},\{0,1,2,3\},\{0,1,3,4\},\{0,2,3,4\},\{0,1,3,5\}\}\\\hline
\{\{0,1,2\},\{0,1,2,3\},\{0,1,3,4\},\{0,2,3,4\},\{0,1,4,5\}\}\\\hline
\{\{0,1,2\},\{0,1,3,4\},\{0,2,3,4\},\{0,1,3,5\},\{0,2,3,5\}\}\\\hline
\{\{0,1,2\},\{0,1,3,4\},\{0,1,3,5\},\{0,1,4,5\},\{0,3,4,5\}\}\\\hline
\{\{0,1,2\},\{0,1,3,4\},\{0,2,3,4\},\{0,1,3,5\},\{0,3,4,5\}\}\\\hline
\{\{0,1,2\},\{0,1,2,3\},\{0,1,3,4\},\{0,2,3,4\},\{1,3,4,5\}\}\\\hline
\{\{0,1,2\},\{0,1,2,3\},\{0,1,4,5\},\{0,2,4,5\},\{1,3,4,5\}\}\\\hline
\{\{0,1,2\},\{0,1,3,4\},\{0,1,3,5\},\{0,3,4,5\},\{1,3,4,5\}\}\\\hline
\{\{0,1,2\},\{0,1,3,4\},\{0,2,3,4\},\{0,3,4,5\},\{1,3,4,5\}\}\\\hline
\{\{0,1,2\},\{0,1,2,3\},\{0,1,3,4\},\{0,1,4,5\},\{0,2,4,5\},\{0,2,3,4,5\}\}\\\hline
\{\{0,1,2\},\{0,1,2,3\},\{0,1,2,4\},\{0,1,3,4\},\{0,3,4,5\},\{1,2,3,4,5\}\}\\\hline
\{\{0,1,2\},\{0,1,2,3\},\{0,1,3,4\},\{0,2,3,4\},\{0,3,4,5\},\{0,1,2,4,5\}\}\\\hline
\{\{0,1,2\},\{0,1,2,3\},\{0,1,3,4\},\{0,3,4,5\},\{1,3,4,5\},\{0,2,3,4,5\}\}\\\hline
\{\{0,1,2\},\{0,1,2,3\},\{0,1,4,5\},\{0,3,4,5\},\{1,3,4,5\},\{0,1,2,3,4\}\}\\\hline
\{\{0,1,2\},\{0,1,2,3\},\{0,1,4,5\},\{0,3,4,5\},\{1,3,4,5\},\{0,2,3,4,5\}\}\\\hline
\{\{0,1,2\},\{0,1,2,3\},\{0,1,3,4\},\{0,3,4,5\},\{1,3,4,5\},\{0,1,2,3,5\}\}\\\hline
\{\{0,1,2\},\{0,1,2,3\},\{0,1,3,4\},\{0,3,4,5\},\{1,3,4,5\},\{0,1,2,4,5\}\}\\\hline
\{\{0,1,2\},\{0,1,2,3\},\{0,1,3,4\},\{0,3,4,5\},\{2,3,4,5\},\{0,1,2,3,5\}\}\\\hline
\{\{0,1,2\},\{0,1,2,3\},\{0,1,3,4\},\{0,3,4,5\},\{2,3,4,5\},\{0,1,2,4,5\}\}\\\hline
\{\{0,1,2\},\{0,1,2,3\},\{0,1,3,4\},\{0,3,4,5\},\{2,3,4,5\},\{1,2,3,4,5\}\}\\\hline
\{\{0,1,2\},\{0,1,2,3\},\{0,1,4,5\},\{0,3,4,5\},\{2,3,4,5\},\{0,1,2,3,4\}\}\\\hline
\{\{0,1,2\},\{0,1,2,3\},\{0,1,4,5\},\{0,3,4,5\},\{2,3,4,5\},\{1,2,3,4,5\}\}\\\hline
\{\{0,1,2\},\{0,1,2,3\},\{0,1,3,4\},\{0,1,3,5\},\{0,3,4,5\},\{0,1,2,4,5\},\{0,2,3,4,5\}\}\\\hline
\{\{0,1,2\},\{0,1,2,3\},\{0,1,3,4\},\{0,1,4,5\},\{0,3,4,5\},\{0,1,2,3,5\},\{0,2,3,4,5\}\}\\\hline
\{\{0,1,2\},\{0,1,2,3\},\{0,3,4,5\},\{1,3,4,5\},\{2,3,4,5\},\{0,1,2,3,4\},\{0,1,2,4,5\}\}\\\hline
\{\{0,1,2\},\{0,1,2,3\},\{0,1,2,4\},\{0,1,3,4\},\{2,3,4,5\},\{0,2,3,4,5\},\{1,2,3,4,5\}\}\\\hline
\{\{0,1,2\},\{0,1,2,3\},\{0,1,2,4\},\{0,1,3,5\},\{2,3,4,5\},\{0,2,3,4,5\},\{1,2,3,4,5\}\}\\\hline
\{\{0,1,2\},\{0,1,2,3\},\{0,1,2,4\},\{0,1,3,5\},\{2,3,4,5\},\{0,1,2,4,5\},\{0,1,3,4,5\},\{0,2,3,4,5\}\}\\\hline
\{\{0,1,2\},\{0,1,2,3\},\{0,1,2,4\},\{0,3,4,5\},\{1,3,4,5\},\{0,1,2,3,5\},\{0,2,3,4,5\},\{1,2,3,4,5\}\}\\\hline
\{\{0,1,2\},\{0,1,2,3\},\{0,1,2,4\},\{0,1,3,4\},\{2,3,4,5\},\{0,1,2,3,5\},\{0,1,2,4,5\},\{0,1,3,4,5\},\{0,2,3,4,5\}\}\\\hline
\{\{0,1,2\},\{0,1,2,3\},\{0,1,2,4\},\{0,1,3,4\},\{0,2,3,5\},\{0,2,4,5\}\}\\\hline
\{\{0,1,2\},\{0,1,2,3\},\{0,1,2,4\},\{0,1,3,4\},\{0,1,3,5\},\{0,2,4,5\}\}\\\hline
\{\{0,1,2\},\{0,1,2,3\},\{0,1,2,4\},\{0,1,3,4\},\{0,2,3,4\},\{0,3,4,5\}\}\\\hline
\{\{0,1,2\},\{0,1,2,3\},\{0,1,2,4\},\{0,1,3,5\},\{0,2,4,5\},\{0,3,4,5\}\}\\\hline
\{\{0,1,2\},\{0,1,2,3\},\{0,1,2,4\},\{0,1,3,4\},\{0,1,3,5\},\{2,3,4,5\}\}\\\hline
\{\{0,1,2\},\{0,1,2,3\},\{0,1,2,4\},\{0,1,3,4\},\{0,3,4,5\},\{1,3,4,5\}\}\\\hline
\{\{0,1,2\},\{0,1,2,3\},\{0,1,2,4\},\{0,1,3,4\},\{0,3,4,5\},\{2,3,4,5\}\}\\\hline
\{\{0,1,2\},\{0,1,2,3\},\{0,1,2,4\},\{0,3,4,5\},\{1,3,4,5\},\{2,3,4,5\}\}\\\hline
\{\{0,1,2\},\{0,1,2,3\},\{0,1,2,4\},\{0,1,3,4\},\{0,1,3,5\},\{0,3,4,5\},\{0,2,3,4,5\}\}\\\hline
\{\{0,1,2\},\{0,1,2,3\},\{0,1,2,4\},\{0,1,3,4\},\{0,2,3,5\},\{0,3,4,5\},\{0,1,2,4,5\}\}\\\hline
\{\{0,1,2\},\{0,1,2,3\},\{0,1,2,4\},\{0,1,3,4\},\{0,1,2,5\},\{2,3,4,5\},\{0,2,3,4,5\}\}\\\hline
\{\{0,1,2\},\{0,1,2,3\},\{0,1,2,4\},\{0,1,3,4\},\{0,1,2,5\},\{0,1,3,5\},\{0,3,4,5\}\}\\\hline
\{\{0,1,2\},\{0,1,3\},\{0,2,3,4,5\},\{1,2,3,4,5\}\}\\\hline
\{\{0,1,2\},\{0,3,4\},\{0,1,2,3,5\},\{0,1,2,4,5\},\{0,1,3,4,5\},\{1,2,3,4,5\}\}\\\hline
\{\{0,1,2\},\{3,4,5\},\{0,1,2,3,4\},\{0,1,2,3,5\},\{0,1,2,4,5\},\{0,1,3,4,5\},\{0,2,3,4,5\},\{1,2,3,4,5\}\}\\\hline
\{\{0,1,2\},\{0,3,4\},\{1,2,3,4\}\}\\\hline
\{\{0,1,2\},\{0,1,3\},\{0,2,3,4\}\}\\\hline
\{\{0,1,2\},\{0,1,3\},\{0,1,4,5\},\{0,2,3,4,5\}\}\\\hline
\{\{0,1,2\},\{0,1,3\},\{0,2,4,5\},\{1,2,3,4,5\}\}\\\hline
\{\{0,1,2\},\{0,1,3\},\{2,3,4,5\},\{0,2,3,4,5\}\}\\\hline
\{\{0,1,2\},\{0,3,4\},\{1,2,3,5\},\{0,1,3,4,5\}\}\\\hline
\{\{0,1,2\},\{0,3,4\},\{0,1,3,5\},\{0,1,2,4,5\}\}\\\hline
\{\{0,1,2\},\{0,3,4\},\{0,1,3,5\},\{1,2,3,4,5\}\}\\\hline
\{\{0,1,2\},\{0,1,3\},\{0,2,4,5\},\{0,1,3,4,5\},\{0,2,3,4,5\}\}\\\hline
\{\{0,1,2\},\{0,3,4\},\{0,1,2,3\},\{0,1,2,4,5\},\{1,2,3,4,5\}\}\\\hline
\{\{0,1,2\},\{0,3,4\},\{0,1,2,3\},\{0,1,3,4,5\},\{1,2,3,4,5\}\}\\\hline
\{\{0,1,2\},\{0,3,4\},\{1,2,3,5\},\{0,1,2,4,5\},\{1,2,3,4,5\}\}\\\hline
\{\{0,1,2\},\{0,1,3\},\{0,2,4,5\},\{0,1,2,3,4\},\{0,1,2,3,5\},\{0,2,3,4,5\}\}\\\hline
\{\{0,1,2\},\{0,1,3\},\{0,2,4,5\},\{0,1,2,3,4\},\{0,1,2,3,5\},\{0,1,3,4,5\}\}\\\hline
\{\{0,1,2\},\{0,1,3\},\{2,3,4,5\},\{0,1,2,3,4\},\{0,1,2,4,5\},\{0,1,3,4,5\}\}\\\hline
\{\{0,1,2\},\{0,3,4\},\{0,1,2,5\},\{0,1,2,3,5\},\{0,1,3,4,5\},\{1,2,3,4,5\}\}\\\hline
\{\{0,1,2\},\{0,3,4\},\{0,1,2,5\},\{0,1,3,4,5\},\{0,2,3,4,5\},\{1,2,3,4,5\}\}\\\hline
\{\{0,1,2\},\{3,4,5\},\{0,1,3,4\},\{0,1,2,3,5\},\{0,1,2,4,5\},\{0,2,3,4,5\}\}\\\hline
\{\{0,1,2\},\{0,1,3\},\{0,1,4,5\},\{0,2,4,5\}\}\\\hline
\{\{0,1,2\},\{0,1,3\},\{0,2,4,5\},\{1,2,4,5\}\}\\\hline
\{\{0,1,2\},\{0,1,3\},\{0,2,4,5\},\{1,3,4,5\}\}\\\hline
\{\{0,1,2\},\{0,1,3\},\{0,1,4,5\},\{2,3,4,5\}\}\\\hline
\{\{0,1,2\},\{0,3,4\},\{0,1,2,3\},\{1,3,4,5\}\}\\\hline
\{\{0,1,2\},\{0,3,4\},\{0,1,2,3\},\{0,1,4,5\}\}\\\hline
\{\{0,1,2\},\{3,4,5\},\{0,1,3,4\},\{0,2,3,5\}\}\\\hline
\{\{0,1,2\},\{0,1,3\},\{0,1,2,4\},\{0,2,4,5\},\{0,1,3,4,5\}\}\\\hline
\{\{0,1,2\},\{0,1,3\},\{0,1,2,4\},\{0,3,4,5\},\{0,1,2,3,5\}\}\\\hline
\{\{0,1,2\},\{0,1,3\},\{0,1,2,4\},\{0,3,4,5\},\{0,1,2,4,5\}\}\\\hline
\{\{0,1,2\},\{0,1,3\},\{0,1,2,4\},\{0,3,4,5\},\{0,2,3,4,5\}\}\\\hline
\{\{0,1,2\},\{0,1,3\},\{0,1,2,4\},\{0,2,4,5\},\{0,2,3,4,5\}\}\\\hline
\{\{0,1,2\},\{0,1,3\},\{0,1,2,4\},\{2,3,4,5\},\{0,1,3,4,5\}\}\\\hline
\{\{0,1,2\},\{0,3,4\},\{0,1,2,5\},\{1,3,4,5\},\{0,2,3,4,5\}\}\\\hline
\{\{0,1,2\},\{0,3,4\},\{0,1,2,5\},\{1,3,4,5\},\{1,2,3,4,5\}\}\\\hline
\{\{0,1,2\},\{0,3,4\},\{0,1,2,3\},\{1,2,3,5\},\{0,1,2,4,5\}\}\\\hline
\{\{0,1,2\},\{0,3,4\},\{0,1,2,3\},\{1,2,4,5\},\{0,1,2,3,5\}\}\\\hline
\{\{0,1,2\},\{0,3,4\},\{0,1,2,5\},\{1,2,3,5\},\{0,1,2,4,5\}\}\\\hline
\{\{0,1,2\},\{0,3,4\},\{0,1,2,5\},\{1,2,3,5\},\{1,2,3,4,5\}\}\\\hline
\{\{0,1,2\},\{0,3,4\},\{0,1,2,3\},\{0,1,2,4\},\{1,2,3,4,5\}\}\\\hline
\{\{0,1,2\},\{0,3,4\},\{0,1,2,3\},\{0,1,3,4\},\{1,2,3,4,5\}\}\\\hline
\{\{0,1,2\},\{3,4,5\},\{0,1,3,4\},\{0,2,3,4\},\{1,2,3,4,5\}\}\\\hline
\{\{0,1,2\},\{3,4,5\},\{0,1,3,4\},\{0,2,3,4\},\{0,1,2,3,5\}\}\\\hline
\{\{0,1,2\},\{0,1,3\},\{0,1,2,4\},\{2,3,4,5\},\{0,1,2,3,5\},\{0,1,2,4,5\}\}\\\hline
\{\{0,1,2\},\{3,4,5\},\{0,1,2,3\},\{0,1,4,5\},\{0,2,3,4,5\},\{1,2,3,4,5\}\}\\\hline
\{\{0,1,2\},\{3,4,5\},\{0,1,2,3\},\{0,1,3,4\},\{0,1,2,4,5\},\{0,2,3,4,5\}\}\\\hline
\{\{0,1,2\},\{3,4,5\},\{0,1,2,3\},\{0,1,3,4\},\{0,2,3,4,5\},\{1,2,3,4,5\}\}\\\hline
\{\{0,1,2\},\{0,3,4\},\{0,1,2,3\},\{0,1,3,4\},\{0,1,2,3,5\},\{0,1,2,4,5\},\{0,1,3,4,5\},\{0,2,3,4,5\}\}\\\hline
\{\{0,1,2\},\{3,4,5\},\{0,1,2,3\},\{0,1,2,4\},\{0,1,2,3,5\},\{0,1,3,4,5\},\{0,2,3,4,5\},\{1,2,3,4,5\}\}\\\hline
\{\{0,1,2\},\{3,4,5\},\{0,1,2,3\},\{0,3,4,5\},\{0,1,2,3,4\},\{0,1,2,3,5\},\{0,1,2,4,5\},\{0,1,3,4,5\},\{1,2,3,4,5\}\}\\\hline
\{\{0,1,2\},\{0,1,3\},\{0,1,2,4\},\{0,1,3,4\},\{2,3,4,5\}\}\\\hline
\{\{0,1,2\},\{0,1,3\},\{0,1,2,4\},\{0,1,2,5\},\{0,2,4,5\}\}\\\hline
\{\{0,1,2\},\{0,1,3\},\{0,1,2,4\},\{0,1,3,4\},\{0,2,4,5\}\}\\\hline
\{\{0,1,2\},\{0,1,3\},\{0,1,2,4\},\{0,1,3,5\},\{2,3,4,5\}\}\\\hline
\{\{0,1,2\},\{0,3,4\},\{0,1,2,3\},\{0,1,3,5\},\{0,3,4,5\}\}\\\hline
\{\{0,1,2\},\{0,3,4\},\{0,1,2,3\},\{0,1,2,5\},\{0,1,3,5\}\}\\\hline
\{\{0,1,2\},\{0,3,4\},\{0,1,2,3\},\{1,2,3,5\},\{0,3,4,5\}\}\\\hline
\{\{0,1,2\},\{0,3,4\},\{0,1,2,3\},\{1,2,4,5\},\{0,3,4,5\}\}\\\hline
\{\{0,1,2\},\{0,3,4\},\{0,1,2,5\},\{1,2,3,5\},\{0,3,4,5\}\}\\\hline
\{\{0,1,2\},\{0,3,4\},\{0,1,2,3\},\{0,1,3,4\},\{0,1,3,5\}\}\\\hline
\{\{0,1,2\},\{0,3,4\},\{0,1,2,3\},\{0,1,2,4\},\{1,2,3,5\}\}\\\hline
\{\{0,1,2\},\{0,3,4\},\{0,1,2,3\},\{0,1,2,5\},\{1,2,3,5\}\}\\\hline
\{\{0,1,2\},\{3,4,5\},\{0,1,2,3\},\{0,1,3,4\},\{0,1,4,5\}\}\\\hline
\{\{0,1,2\},\{3,4,5\},\{0,1,2,3\},\{0,1,3,4\},\{0,2,3,4\}\}\\\hline
\{\{0,1,2\},\{3,4,5\},\{0,1,2,3\},\{0,1,4,5\},\{0,2,4,5\}\}\\\hline
\{\{0,1,2\},\{0,3,4\},\{0,1,2,3\},\{0,1,2,4\},\{0,1,3,4\},\{0,2,3,4,5\}\}\\\hline
\{\{0,1,2\},\{0,3,4\},\{0,1,2,3\},\{0,1,2,5\},\{0,3,4,5\},\{1,2,3,4,5\}\}\\\hline
\{\{0,1,2\},\{3,4,5\},\{0,1,2,3\},\{0,1,3,4\},\{2,3,4,5\},\{0,1,2,4,5\}\}\\\hline
\{\{0,1,2\},\{3,4,5\},\{0,1,2,3\},\{0,1,2,4\},\{0,1,3,5\},\{0,2,3,4,5\}\}\\\hline
\{\{0,1,2\},\{3,4,5\},\{0,1,2,3\},\{0,1,3,4\},\{0,3,4,5\},\{0,1,2,4,5\}\}\\\hline
\{\{0,1,2\},\{0,3,4\},\{0,1,2,3\},\{0,1,3,4\},\{0,1,2,5\},\{0,1,2,4,5\},\{0,2,3,4,5\}\}\\\hline
\{\{0,1,2\},\{3,4,5\},\{0,1,2,3\},\{0,1,3,4\},\{2,3,4,5\},\{0,1,2,3,5\},\{0,2,3,4,5\}\}\\\hline
\{\{0,1,2\},\{3,4,5\},\{0,1,2,3\},\{0,1,4,5\},\{2,3,4,5\},\{0,1,2,3,4\},\{0,2,3,4,5\}\}\\\hline
\{\{0,1,2\},\{0,3,4\},\{0,1,2,3\},\{0,1,2,4\},\{0,1,3,4\},\{0,2,3,4\}\}\\\hline
\{\{0,1,2\},\{3,4,5\},\{0,1,2,3\},\{0,1,2,4\},\{0,1,3,5\},\{0,3,4,5\}\}\\\hline
\{\{0,1,2\},\{3,4,5\},\{0,1,2,3\},\{0,1,2,4\},\{0,1,3,4\},\{2,3,4,5\}\}\\\hline
\{\{0,1,2\},\{3,4,5\},\{0,1,2,3\},\{0,1,2,4\},\{0,1,3,5\},\{2,3,4,5\}\}\\\hline
\{\{0,1,2\},\{0,3,4\},\{0,1,2,3\},\{0,1,2,4\},\{0,1,2,5\},\{0,3,4,5\},\{0,1,3,4,5\},\{0,2,3,4,5\}\}\\\hline
\{\{0,1,2\},\{3,4,5\},\{0,1,2,3\},\{0,1,2,4\},\{0,3,4,5\},\{1,3,4,5\},\{0,1,2,3,5\},\{0,1,2,4,5\}\}\\\hline
\{\{0,1,2\},\{3,4,5\},\{0,1,2,3\},\{0,1,2,4\},\{0,3,4,5\},\{1,3,4,5\},\{0,1,2,3,5\},\{0,2,3,4,5\}\}\\\hline
\{\{0,1,2\},\{0,3,4\},\{1,3,5\}\}\\\hline
\{\{0,1,2\},\{0,1,3\},\{0,4,5\}\}\\\hline
\{\{0,1,2\},\{0,1,3\},\{2,3,4\}\}\\\hline
\{\{0,1,2\},\{0,1,3\},\{0,2,4\}\}\\\hline
\{\{0,1,2\},\{0,1,3\},\{0,1,4\}\}\\\hline
\{\{0,1,2\},\{0,1,3\},\{0,2,3\}\}\\\hline
\{\{0,1,2\},\{0,1,3\},\{2,4,5\},\{0,2,3,4,5\}\}\\\hline
\{\{0,1,2\},\{0,1,3\},\{2,4,5\},\{0,1,2,3,4\},\{0,1,3,4,5\}\}\\\hline
\{\{0,1,2\},\{0,1,3\},\{2,4,5\},\{0,2,4,5\}\}\\\hline
\{\{0,1,2\},\{0,1,3\},\{2,4,5\},\{0,1,2,4\},\{0,1,2,3,5\}\}\\\hline
\{\{0,1,2\},\{0,1,3\},\{2,4,5\},\{2,3,4,5\},\{0,1,2,3,4\}\}\\\hline
\{\{0,1,2\},\{0,1,3\},\{2,4,5\},\{2,3,4,5\},\{0,1,3,4,5\}\}\\\hline
\{\{0,1\}\}\\\hline
\{\{0\}\}\\\hline
\end{longtable}
\end{footnotesize}

$\nonFCsix = $

\begin{center}
\begin{footnotesize}
\begin{longtable}{|p{\textwidth}|}
\hline
\{\{0,1,2,3\},\{0,1,2,4\},\{0,1,3,4\},\{0,2,3,4\},\{0,1,2,3,4\},\{1,2,3,4,5\},\{0,1,2,3,4,5\}\}\\\hline
\{\{0,1,2,3\},\{0,1,2,4\},\{0,1,3,4\},\{0,1,2,5\},\{0,1,2,3,4\},\{0,1,2,3,5\},\{0,1,2,4,5\},\{0,2,3,4,5\},\{1,2,3,4,5\},\\
\{0,1,2,3,4,5\}\}\\\hline
\{\{0,1,2,3\},\{0,1,2,4\},\{0,1,3,4\},\{0,2,3,5\},\{0,1,2,3,4\},\{0,1,2,3,5\},\{0,1,2,4,5\},\{0,1,3,4,5\},\{0,2,3,4,5\},\\
\{1,2,3,4,5\},\{0,1,2,3,4,5\}\}\\\hline
\{\{0,1,2,3\},\{0,1,2,4\},\{0,1,3,4\},\{2,3,4,5\},\{0,1,2,3,4\},\{0,1,2,3,5\},\{0,1,2,4,5\},\{0,1,3,4,5\},\{0,2,3,4,5\},\\
\{1,2,3,4,5\},\{0,1,2,3,4,5\}\}\\\hline
\{\{0,1,2,3\},\{0,1,2,4\},\{0,1,3,5\},\{0,1,4,5\},\{0,1,2,3,4\},\{0,1,2,3,5\},\{0,1,2,4,5\},\{0,1,3,4,5\},\{0,2,3,4,5\},\\
\{1,2,3,4,5\},\{0,1,2,3,4,5\}\}\\\hline
\{\{0,1,2,3\},\{0,1,2,4\},\{0,1,3,5\},\{0,2,4,5\},\{0,1,2,3,4\},\{0,1,2,3,5\},\{0,1,2,4,5\},\{0,1,3,4,5\},\{0,2,3,4,5\},\\
\{1,2,3,4,5\},\{0,1,2,3,4,5\}\}\\\hline
\{\{0,1,2,3\},\{0,1,2,4\},\{0,1,3,5\},\{2,3,4,5\},\{0,1,2,3,4\},\{0,1,2,3,5\},\{0,1,2,4,5\},\{0,1,3,4,5\},\{0,2,3,4,5\},\\
\{1,2,3,4,5\},\{0,1,2,3,4,5\}\}\\\hline
\{\{0,1,2,3\},\{0,1,2,4\},\{0,3,4,5\},\{1,3,4,5\},\{0,1,2,3,4\},\{0,1,2,3,5\},\{0,1,2,4,5\},\{0,1,3,4,5\},\{0,2,3,4,5\},\\
\{1,2,3,4,5\},\{0,1,2,3,4,5\}\}\\\hline
\{\{0,1,2,3\},\{0,1,2,4\},\{0,1,3,4\},\{0,2,3,5\},\{1,2,3,5\},\{0,1,2,3,4\},\{0,1,2,3,5\},\{0,1,2,3,4,5\}\}\\\hline
\{\{0,1,2,3\},\{0,1,2,4\},\{0,1,3,4\},\{0,2,3,5\},\{2,3,4,5\},\{0,1,2,3,4\},\{0,1,2,3,5\},\{0,2,3,4,5\},\{0,1,2,3,4,5\}\}\\\hline
\{\{0,1,2,3\},\{0,1,2,4\},\{0,1,3,4\},\{0,1,2,5\},\{2,3,4,5\},\{0,1,2,3,4\},\{0,1,2,3,5\},\{0,1,2,4,5\},\{0,2,3,4,5\},\\
\{0,1,2,3,4,5\}\}\\\hline
\{\{0,1,2,3\},\{0,1,2,4\},\{0,1,3,4\},\{0,2,3,5\},\{1,2,4,5\},\{0,1,2,3,4\},\{0,1,2,3,5\},\{0,1,2,4,5\},\{0,1,3,4,5\},\\
\{0,1,2,3,4,5\}\}\\\hline
\{\{0,1,2,3\},\{0,1,2,4\},\{0,1,3,5\},\{0,1,4,5\},\{2,3,4,5\},\{0,1,2,3,4\},\{0,1,2,3,5\},\{0,1,2,4,5\},\{0,1,3,4,5\},\\
\{0,1,2,3,4,5\}\}\\\hline
\{\{0,1,2,3\},\{0,1,2,4\},\{0,1,3,5\},\{0,2,4,5\},\{1,3,4,5\},\{0,1,2,3,4\},\{0,1,2,3,5\},\{0,1,2,4,5\},\{0,1,3,4,5\},\\
\{0,1,2,3,4,5\}\}\\\hline
\{\{0,1,2,3\},\{0,1,2,4\},\{0,1,3,4\},\{0,1,2,5\},\{0,1,3,5\},\{0,2,4,5\},\{0,1,2,3,4\},\{0,1,2,3,5\},\{0,1,2,4,5\},\\
\{0,1,3,4,5\},\{0,1,2,3,4,5\}\}\\\hline
\{\{0,1,2,3\},\{0,1,2,4\},\{0,1,3,4\},\{0,1,2,5\},\{0,2,3,5\},\{0,3,4,5\},\{0,1,2,3,4\},\{0,1,2,3,5\},\{0,1,2,4,5\},\\
\{0,1,3,4,5\},\{0,2,3,4,5\},\{0,1,2,3,4,5\}\}\\\hline
\{\{0,1,2,3\},\{0,1,2,4\},\{0,1,3,4\},\{0,2,3,5\},\{0,2,4,5\},\{0,3,4,5\},\{0,1,2,3,4\},\{0,1,2,3,5\},\{0,1,2,4,5\},\\
\{0,1,3,4,5\},\{0,2,3,4,5\},\{0,1,2,3,4,5\}\}\\\hline
\{\{0,1,2\},\{0,1,3,4\},\{0,2,3,5\},\{0,1,2,3,4\},\{0,1,2,3,5\},\{1,2,3,4,5\},\{0,1,2,3,4,5\}\}\\\hline
\{\{0,1,2\},\{0,1,3,4\},\{0,2,3,4\},\{0,1,2,3,4\},\{0,1,3,4,5\},\{0,2,3,4,5\},\{1,2,3,4,5\},\{0,1,2,3,4,5\}\}\\\hline
\{\{0,1,2\},\{0,1,2,3\},\{0,1,3,4\},\{0,1,2,3,4\},\{0,1,2,4,5\},\{0,1,3,4,5\},\{0,2,3,4,5\},\{1,2,3,4,5\},\{0,1,2,3,4,5\}\}\\\hline
\{\{0,1,2\},\{0,1,3,4\},\{0,3,4,5\},\{0,1,2,3,4\},\{0,1,2,3,5\},\{0,1,3,4,5\},\{0,2,3,4,5\},\{1,2,3,4,5\},\{0,1,2,3,4,5\}\}\\\hline
\{\{0,1,2\},\{0,1,3,4\},\{2,3,4,5\},\{0,1,2,3,4\},\{0,1,2,3,5\},\{0,1,3,4,5\},\{0,2,3,4,5\},\{1,2,3,4,5\},\{0,1,2,3,4,5\}\}\\\hline
\{\{0,1,2\},\{0,1,2,3\},\{0,1,4,5\},\{0,1,2,3,4\},\{0,1,2,3,5\},\{0,1,2,4,5\},\{0,1,3,4,5\},\{0,2,3,4,5\},\{1,2,3,4,5\},\\
\{0,1,2,3,4,5\}\}\\\hline
\{\{0,1,2\},\{0,1,3,4\},\{0,2,3,4\},\{1,3,4,5\},\{0,1,2,3,4\},\{0,1,3,4,5\},\{0,2,3,4,5\},\{0,1,2,3,4,5\}\}\\\hline
\{\{0,1,2\},\{0,1,2,3\},\{0,1,2,4\},\{0,1,3,4\},\{0,1,2,3,4\},\{0,1,2,3,5\},\{0,2,3,4,5\},\{1,2,3,4,5\},\{0,1,2,3,4,5\}\}\\\hline
\{\{0,1,2\},\{0,1,2,3\},\{0,1,2,4\},\{0,1,3,4\},\{0,1,2,3,4\},\{0,1,3,4,5\},\{0,2,3,4,5\},\{1,2,3,4,5\},\{0,1,2,3,4,5\}\}\\\hline
\{\{0,1,2\},\{0,1,2,3\},\{0,1,3,4\},\{0,3,4,5\},\{0,1,2,3,4\},\{0,1,3,4,5\},\{0,2,3,4,5\},\{1,2,3,4,5\},\{0,1,2,3,4,5\}\}\\\hline
\{\{0,1,2\},\{0,1,2,3\},\{0,1,3,4\},\{2,3,4,5\},\{0,1,2,3,4\},\{0,1,2,3,5\},\{0,2,3,4,5\},\{1,2,3,4,5\},\{0,1,2,3,4,5\}\}\\\hline
\{\{0,1,2\},\{0,1,2,3\},\{0,1,3,4\},\{2,3,4,5\},\{0,1,2,3,4\},\{0,1,3,4,5\},\{0,2,3,4,5\},\{1,2,3,4,5\},\{0,1,2,3,4,5\}\}\\\hline
\{\{0,1,2\},\{0,1,2,3\},\{0,1,4,5\},\{0,3,4,5\},\{0,1,2,4,5\},\{0,1,3,4,5\},\{0,2,3,4,5\},\{1,2,3,4,5\},\{0,1,2,3,4,5\}\}\\\hline
\{\{0,1,2\},\{0,1,2,3\},\{0,1,4,5\},\{2,3,4,5\},\{0,1,2,3,4\},\{0,1,2,4,5\},\{0,2,3,4,5\},\{1,2,3,4,5\},\{0,1,2,3,4,5\}\}\\\hline
\{\{0,1,2\},\{0,1,2,3\},\{0,1,4,5\},\{2,3,4,5\},\{0,1,2,4,5\},\{0,1,3,4,5\},\{0,2,3,4,5\},\{1,2,3,4,5\},\{0,1,2,3,4,5\}\}\\\hline
\{\{0,1,2\},\{0,1,3,4\},\{0,2,3,4\},\{0,1,3,5\},\{0,1,2,3,4\},\{0,1,2,3,5\},\{0,1,3,4,5\},\{0,2,3,4,5\},\{0,1,2,3,4,5\}\}\\\hline
\{\{0,1,2\},\{0,1,3,4\},\{0,3,4,5\},\{2,3,4,5\},\{0,1,2,3,4\},\{0,1,2,3,5\},\{0,1,3,4,5\},\{0,2,3,4,5\},\{0,1,2,3,4,5\}\}\\\hline
\{\{0,1,2\},\{0,1,2,3\},\{0,1,2,4\},\{0,1,3,5\},\{0,1,2,3,4\},\{0,1,2,3,5\},\{0,1,2,4,5\},\{0,2,3,4,5\},\{1,2,3,4,5\},\\
\{0,1,2,3,4,5\}\}\\\hline
\{\{0,1,2\},\{0,1,2,3\},\{0,1,2,4\},\{0,1,3,5\},\{0,1,2,3,4\},\{0,1,2,3,5\},\{0,1,3,4,5\},\{0,2,3,4,5\},\{1,2,3,4,5\},\\
\{0,1,2,3,4,5\}\}\\\hline
\{\{0,1,2\},\{0,1,2,3\},\{0,1,3,4\},\{2,3,4,5\},\{0,1,2,3,4\},\{0,1,2,3,5\},\{0,1,2,4,5\},\{0,1,3,4,5\},\{0,2,3,4,5\},\\
\{0,1,2,3,4,5\}\}\\\hline
\{\{0,1,2\},\{0,1,2,3\},\{0,1,4,5\},\{2,3,4,5\},\{0,1,2,3,4\},\{0,1,2,3,5\},\{0,1,2,4,5\},\{0,1,3,4,5\},\{0,2,3,4,5\},\\
\{0,1,2,3,4,5\}\}\\\hline
\{\{0,1,2\},\{0,1,2,3\},\{0,3,4,5\},\{1,3,4,5\},\{0,1,2,3,4\},\{0,1,2,4,5\},\{0,1,3,4,5\},\{0,2,3,4,5\},\{1,2,3,4,5\},\\
\{0,1,2,3,4,5\}\}\\\hline
\{\{0,1,2\},\{0,1,3,4\},\{0,1,3,5\},\{0,3,4,5\},\{0,1,2,3,4\},\{0,1,2,3,5\},\{0,1,2,4,5\},\{0,1,3,4,5\},\{0,2,3,4,5\},\\
\{0,1,2,3,4,5\}\}\\\hline
\{\{0,1,2\},\{0,1,2,3\},\{0,1,4,5\},\{0,3,4,5\},\{1,3,4,5\},\{0,1,2,4,5\},\{0,1,3,4,5\},\{0,1,2,3,4,5\}\}\\\hline
\{\{0,1,2\},\{0,1,2,3\},\{0,1,3,4\},\{0,1,3,5\},\{2,3,4,5\},\{0,1,2,3,4\},\{0,1,2,3,5\},\{0,1,3,4,5\},\{0,1,2,3,4,5\}\}\\\hline
\{\{0,1,2\},\{0,1,2,3\},\{0,1,3,4\},\{0,1,4,5\},\{0,2,4,5\},\{0,1,2,3,4\},\{0,1,2,4,5\},\{0,1,3,4,5\},\{0,1,2,3,4,5\}\}\\\hline
\{\{0,1,2\},\{0,1,2,3\},\{0,1,3,4\},\{0,1,4,5\},\{2,3,4,5\},\{0,1,2,3,4\},\{0,1,2,4,5\},\{0,1,3,4,5\},\{0,1,2,3,4,5\}\}\\\hline
\{\{0,1,2\},\{0,1,2,3\},\{0,1,3,4\},\{0,3,4,5\},\{2,3,4,5\},\{0,1,2,3,4\},\{0,1,3,4,5\},\{0,2,3,4,5\},\{0,1,2,3,4,5\}\}\\\hline
\{\{0,1,2\},\{0,1,2,3\},\{0,1,4,5\},\{0,3,4,5\},\{2,3,4,5\},\{0,1,2,4,5\},\{0,1,3,4,5\},\{0,2,3,4,5\},\{0,1,2,3,4,5\}\}\\\hline
\{\{0,1,2\},\{0,1,2,3\},\{0,1,2,4\},\{0,1,3,4\},\{2,3,4,5\},\{0,1,2,3,4\},\{0,1,2,3,5\},\{0,1,2,4,5\},\{0,2,3,4,5\},\\
\{0,1,2,3,4,5\}\}\\\hline
\{\{0,1,2\},\{0,1,2,3\},\{0,1,2,4\},\{0,1,3,4\},\{2,3,4,5\},\{0,1,2,3,4\},\{0,1,2,3,5\},\{0,1,3,4,5\},\{0,2,3,4,5\},\\
\{0,1,2,3,4,5\}\}\\\hline
\{\{0,1,2\},\{0,1,2,3\},\{0,1,2,4\},\{0,1,3,5\},\{0,3,4,5\},\{0,1,2,3,4\},\{0,1,2,3,5\},\{0,1,3,4,5\},\{1,2,3,4,5\},\\
\{0,1,2,3,4,5\}\}\\\hline
\{\{0,1,2\},\{0,1,2,3\},\{0,1,2,4\},\{0,1,3,5\},\{2,3,4,5\},\{0,1,2,3,4\},\{0,1,2,3,5\},\{0,1,2,4,5\},\{0,2,3,4,5\},\\
\{0,1,2,3,4,5\}\}\\\hline
\{\{0,1,2\},\{0,1,2,3\},\{0,1,2,4\},\{0,1,3,5\},\{2,3,4,5\},\{0,1,2,3,4\},\{0,1,2,3,5\},\{0,1,3,4,5\},\{0,2,3,4,5\},\\
\{0,1,2,3,4,5\}\}\\\hline
\{\{0,1,2\},\{0,1,2,3\},\{0,1,2,4\},\{0,3,4,5\},\{1,3,4,5\},\{0,1,2,3,4\},\{0,1,3,4,5\},\{0,2,3,4,5\},\{1,2,3,4,5\},\\
\{0,1,2,3,4,5\}\}\\\hline
\{\{0,1,2\},\{0,1,2,3\},\{0,1,3,4\},\{0,1,3,5\},\{0,2,4,5\},\{0,1,2,3,4\},\{0,1,2,3,5\},\{0,1,2,4,5\},\{0,1,3,4,5\},\\
\{0,1,2,3,4,5\}\}\\\hline
\{\{0,1,2\},\{0,1,2,3\},\{0,1,3,4\},\{0,1,3,5\},\{0,3,4,5\},\{0,1,2,3,4\},\{0,1,2,3,5\},\{0,1,3,4,5\},\{0,2,3,4,5\},\\
\{0,1,2,3,4,5\}\}\\\hline
\{\{0,1,2\},\{0,1,2,3\},\{0,1,3,4\},\{0,1,4,5\},\{0,3,4,5\},\{0,1,2,3,4\},\{0,1,2,4,5\},\{0,1,3,4,5\},\{0,2,3,4,5\},\\
\{0,1,2,3,4,5\}\}\\\hline
\{\{0,1,2\},\{0,1,2,3\},\{0,1,3,4\},\{0,2,3,5\},\{0,1,4,5\},\{0,1,2,3,4\},\{0,1,2,3,5\},\{0,1,2,4,5\},\{0,1,3,4,5\},\\
\{0,1,2,3,4,5\}\}\\\hline
\{\{0,1,2\},\{0,1,2,3\},\{0,1,2,4\},\{0,3,4,5\},\{1,3,4,5\},\{0,1,2,3,4\},\{0,1,2,3,5\},\{0,1,2,4,5\},\{0,1,3,4,5\},\\
\{0,2,3,4,5\},\{0,1,2,3,4,5\}\}\\\hline
\{\{0,1,2\},\{0,1,2,3\},\{0,1,3,4\},\{0,2,3,5\},\{0,3,4,5\},\{0,1,2,3,4\},\{0,1,2,3,5\},\{0,1,2,4,5\},\{0,1,3,4,5\},\\
\{0,2,3,4,5\},\{0,1,2,3,4,5\}\}\\\hline
\{\{0,1,2\},\{0,1,2,3\},\{0,1,3,4\},\{0,2,4,5\},\{0,3,4,5\},\{0,1,2,3,4\},\{0,1,2,3,5\},\{0,1,2,4,5\},\{0,1,3,4,5\},\\
\{0,2,3,4,5\},\{0,1,2,3,4,5\}\}\\\hline
\{\{0,1,2\},\{0,1,2,3\},\{0,1,4,5\},\{0,2,4,5\},\{0,3,4,5\},\{0,1,2,3,4\},\{0,1,2,3,5\},\{0,1,2,4,5\},\{0,1,3,4,5\},\\
\{0,2,3,4,5\},\{0,1,2,3,4,5\}\}\\\hline
\{\{0,1,2\},\{0,1,2,3\},\{0,1,2,4\},\{0,1,2,5\},\{0,3,4,5\},\{0,1,2,3,4\},\{0,1,2,3,5\},\{0,1,2,4,5\},\{0,1,3,4,5\},\\
\{0,2,3,4,5\},\{1,2,3,4,5\},\{0,1,2,3,4,5\}\}\\\hline
\{\{0,1,2\},\{0,1,2,3\},\{0,1,2,4\},\{0,1,2,5\},\{0,3,4,5\},\{1,3,4,5\},\{0,1,2,3,4\},\{0,1,2,3,5\},\{0,1,2,4,5\},\\
\{0,1,3,4,5\},\{0,1,2,3,4,5\}\}\\\hline
\{\{0,1,2\},\{0,1,2,3\},\{0,1,2,4\},\{0,1,3,4\},\{0,1,2,5\},\{2,3,4,5\},\{0,1,2,3,4\},\{0,1,2,3,5\},\{0,1,2,4,5\},\\
\{0,1,3,4,5\},\{0,1,2,3,4,5\}\}\\\hline
\{\{0,1,2\},\{0,1,2,3\},\{0,1,2,4\},\{0,1,3,4\},\{0,1,3,5\},\{0,3,4,5\},\{0,1,2,3,4\},\{0,1,2,3,5\},\{0,1,2,4,5\},\\
\{0,1,3,4,5\},\{0,1,2,3,4,5\}\}\\\hline
\{\{0,1,2\},\{0,1,2,3\},\{0,1,2,4\},\{0,1,3,4\},\{0,2,3,5\},\{0,3,4,5\},\{0,1,2,3,4\},\{0,1,2,3,5\},\{0,1,3,4,5\},\\
\{0,2,3,4,5\},\{0,1,2,3,4,5\}\}\\\hline
\{\{0,1,2\},\{0,1,2,3\},\{0,1,2,4\},\{0,1,3,5\},\{0,1,4,5\},\{0,3,4,5\},\{0,1,2,3,4\},\{0,1,2,3,5\},\{0,1,2,4,5\},\\
\{0,1,3,4,5\},\{0,1,2,3,4,5\}\}\\\hline
\{\{0,1,2\},\{0,1,2,3\},\{0,1,2,4\},\{0,1,3,4\},\{0,1,2,5\},\{0,2,3,5\},\{0,1,2,3,4\},\{0,1,2,3,5\},\{0,1,2,4,5\},\\
\{0,1,3,4,5\},\{0,2,3,4,5\},\{0,1,2,3,4,5\}\}\\\hline
\{\{0,1,2\},\{0,1,2,3\},\{0,1,2,4\},\{0,1,3,4\},\{0,1,2,5\},\{0,3,4,5\},\{0,1,2,3,4\},\{0,1,2,3,5\},\{0,1,2,4,5\},\\
\{0,1,3,4,5\},\{0,2,3,4,5\},\{0,1,2,3,4,5\}\}\\\hline
\{\{0,1,2\},\{0,1,2,3\},\{0,1,2,4\},\{0,1,3,4\},\{0,2,3,4\},\{0,1,2,5\},\{0,1,2,3,4\},\{0,1,2,3,5\},\{0,1,2,4,5\},\\
\{0,1,3,4,5\},\{0,2,3,4,5\},\{0,1,2,3,4,5\}\}\\\hline
\{\{0,1,2\},\{0,3,4\},\{0,1,2,3,4\},\{0,1,2,3,5\},\{0,1,3,4,5\},\{1,2,3,4,5\},\{0,1,2,3,4,5\}\}\\\hline
\{\{0,1,2\},\{0,3,4\},\{1,2,3,5\},\{0,1,2,3,4\},\{0,1,2,3,5\},\{0,1,2,4,5\},\{0,1,2,3,4,5\}\}\\\hline
\{\{0,1,2\},\{0,3,4\},\{0,1,2,5\},\{0,1,2,3,4\},\{0,1,2,3,5\},\{0,1,2,4,5\},\{1,2,3,4,5\},\{0,1,2,3,4,5\}\}\\\hline
\{\{0,1,2\},\{3,4,5\},\{0,1,2,3\},\{0,1,2,3,4\},\{0,1,2,4,5\},\{0,1,3,4,5\},\{0,2,3,4,5\},\{1,2,3,4,5\},\\
\{0,1,2,3,4,5\}\}\\\hline
\{\{0,1,2\},\{0,1,3\},\{0,1,2,3\},\{2,3,4,5\},\{0,1,2,4,5\},\{0,1,3,4,5\},\{0,1,2,3,4,5\}\}\\\hline
\{\{0,1,2\},\{0,3,4\},\{0,1,2,3\},\{1,2,4,5\},\{0,1,2,3,4\},\{0,1,2,4,5\},\{0,1,2,3,4,5\}\}\\\hline
\{\{0,1,2\},\{0,3,4\},\{0,1,2,5\},\{1,3,4,5\},\{0,1,2,3,4\},\{0,1,3,4,5\},\{0,1,2,3,4,5\}\}\\\hline
\{\{0,1,2\},\{0,1,3\},\{0,1,2,3\},\{0,2,4,5\},\{0,1,2,3,4\},\{0,1,2,4,5\},\{0,1,3,4,5\},\{0,1,2,3,4,5\}\}\\\hline
\{\{0,1,2\},\{0,1,3\},\{0,1,2,3\},\{0,2,4,5\},\{0,1,2,3,4\},\{0,1,2,4,5\},\{0,2,3,4,5\},\{0,1,2,3,4,5\}\}\\\hline
\{\{0,1,2\},\{0,1,3\},\{0,1,2,3\},\{2,3,4,5\},\{0,1,2,3,4\},\{0,1,2,3,5\},\{0,1,2,4,5\},\{0,1,2,3,4,5\}\}\\\hline
\{\{0,1,2\},\{0,3,4\},\{0,1,2,3\},\{0,1,2,5\},\{0,1,2,3,4\},\{0,1,2,3,5\},\{1,2,3,4,5\},\{0,1,2,3,4,5\}\}\\\hline
\{\{0,1,2\},\{0,3,4\},\{0,1,2,3\},\{0,1,3,5\},\{0,1,2,3,4\},\{0,1,2,3,5\},\{0,1,3,4,5\},\{0,1,2,3,4,5\}\}\\\hline
\{\{0,1,2\},\{0,3,4\},\{0,1,2,3\},\{0,3,4,5\},\{0,1,2,3,4\},\{0,1,2,3,5\},\{1,2,3,4,5\},\{0,1,2,3,4,5\}\}\\\hline
\{\{0,1,2\},\{0,3,4\},\{0,1,2,3\},\{1,2,3,5\},\{0,1,2,3,4\},\{0,1,2,3,5\},\{1,2,3,4,5\},\{0,1,2,3,4,5\}\}\\\hline
\{\{0,1,2\},\{0,3,4\},\{0,1,2,5\},\{0,3,4,5\},\{0,1,2,3,4\},\{0,1,2,3,5\},\{1,2,3,4,5\},\{0,1,2,3,4,5\}\}\\\hline
\{\{0,1,2\},\{0,3,4\},\{0,1,2,3\},\{0,1,3,4\},\{0,1,2,3,4\},\{0,1,2,3,5\},\{0,1,2,4,5\},\{0,2,3,4,5\},\\
\{0,1,2,3,4,5\}\}\\\hline
\{\{0,1,2\},\{3,4,5\},\{0,1,2,3\},\{0,1,4,5\},\{0,1,2,3,4\},\{0,1,2,4,5\},\{0,1,3,4,5\},\{0,2,3,4,5\},\\
\{0,1,2,3,4,5\}\}\\\hline
\{\{0,1,2\},\{0,1,3\},\{0,1,2,3\},\{0,1,2,4\},\{0,3,4,5\},\{0,1,2,3,4\},\{0,1,3,4,5\},\\
\{0,1,2,3,4,5\}\}\\\hline
\{\{0,1,2\},\{0,1,3\},\{0,1,2,3\},\{0,1,2,4\},\{2,3,4,5\},\{0,1,2,3,4\},\{0,1,2,3,5\},\\
\{0,1,2,3,4,5\}\}\\\hline
\{\{0,1,2\},\{0,1,3\},\{0,1,2,3\},\{0,1,2,4\},\{2,3,4,5\},\{0,1,2,3,4\},\{0,1,2,4,5\},\\
\{0,1,2,3,4,5\}\}\\\hline
\{\{0,1,2\},\{0,1,3\},\{0,1,2,3\},\{0,1,2,4\},\{0,2,4,5\},\{0,1,2,3,4\},\{0,1,2,3,5\},\{0,1,2,4,5\},\\
\{0,1,2,3,4,5\}\}\\\hline
\{\{0,1,2\},\{0,3,4\},\{0,1,2,5\},\{0,1,3,5\},\{0,3,4,5\},\{0,1,2,3,4\},\{0,1,2,3,5\},\{0,1,3,4,5\},\\
\{0,1,2,3,4,5\}\}\\\hline
\{\{0,1,2\},\{3,4,5\},\{0,1,2,3\},\{0,1,2,4\},\{0,1,3,4\},\{0,1,2,3,4\},\{0,1,3,4,5\},\{0,2,3,4,5\},\\
\{0,1,2,3,4,5\}\}\\\hline
\{\{0,1,2\},\{3,4,5\},\{0,1,2,3\},\{0,1,3,4\},\{0,1,3,5\},\{0,1,2,3,4\},\{0,1,2,3,5\},\{0,1,3,4,5\},\\
\{0,1,2,3,4,5\}\}\\\hline
\{\{0,1,2\},\{3,4,5\},\{0,1,2,3\},\{0,1,3,4\},\{2,3,4,5\},\{0,1,2,3,4\},\{0,1,2,3,5\},\{0,1,3,4,5\},\\
\{0,1,2,3,4,5\}\}\\\hline
\{\{0,1,2\},\{3,4,5\},\{0,1,2,3\},\{0,1,3,4\},\{2,3,4,5\},\{0,1,2,3,4\},\{0,1,3,4,5\},\{0,2,3,4,5\},\\
\{0,1,2,3,4,5\}\}\\\hline
\{\{0,1,2\},\{3,4,5\},\{0,1,2,3\},\{0,1,4,5\},\{2,3,4,5\},\{0,1,2,3,4\},\{0,1,2,4,5\},\{0,1,3,4,5\},\\
\{0,1,2,3,4,5\}\}\\\hline
\{\{0,1,2\},\{3,4,5\},\{0,1,2,3\},\{0,1,2,4\},\{0,3,4,5\},\{0,1,2,3,4\},\{0,1,2,3,5\},\{0,1,3,4,5\},\\
\{1,2,3,4,5\},\{0,1,2,3,4,5\}\}\\\hline
\{\{0,1,2\},\{3,4,5\},\{0,1,2,3\},\{0,1,2,4\},\{0,3,4,5\},\{0,1,2,3,4\},\{0,1,3,4,5\},\{0,2,3,4,5\},\\
\{1,2,3,4,5\},\{0,1,2,3,4,5\}\}\\\hline
\{\{0,1,2\},\{3,4,5\},\{0,1,2,3\},\{0,1,3,4\},\{0,3,4,5\},\{0,1,2,3,4\},\{0,1,2,3,5\},\{0,1,3,4,5\},\\
\{0,2,3,4,5\},\{0,1,2,3,4,5\}\}\\\hline
\{\{0,1,2\},\{0,3,4\},\{0,1,2,3\},\{0,1,2,4\},\{0,1,2,5\},\{0,1,2,3,4\},\{0,1,2,3,5\},\{0,1,2,4,5\},\\
\{0,1,3,4,5\},\{0,2,3,4,5\},\{0,1,2,3,4,5\}\}\\\hline
\{\{0,1,2\},\{0,3,4\},\{0,1,2,3\},\{0,1,2,4\},\{0,3,4,5\},\{0,1,2,3,4\},\{0,1,2,3,5\},\{0,1,2,4,5\},\\
\{0,1,3,4,5\},\{0,2,3,4,5\},\{0,1,2,3,4,5\}\}\\\hline
\{\{0,1,2\},\{0,3,4\},\{0,1,2,3\},\{0,1,2,5\},\{0,3,4,5\},\{0,1,2,3,4\},\{0,1,2,3,5\},\{0,1,2,4,5\},\\
\{0,1,3,4,5\},\{0,2,3,4,5\},\{0,1,2,3,4,5\}\}\\\hline
\{\{0,1,2\},\{3,4,5\},\{0,1,2,3\},\{0,1,2,4\},\{0,1,3,4\},\{0,3,4,5\},\{0,1,2,3,4\},\{0,1,3,4,5\},\\
\{0,1,2,3,4,5\}\}\\\hline
\{\{0,1,2\},\{3,4,5\},\{0,1,2,3\},\{0,1,2,4\},\{0,3,4,5\},\{1,3,4,5\},\{0,1,2,3,4\},\{0,1,2,3,5\},\\
\{0,1,3,4,5\},\{0,1,2,3,4,5\}\}\\\hline
\{\{0,1,2\},\{3,4,5\},\{0,1,2,3\},\{0,1,2,4\},\{0,1,2,5\},\{0,3,4,5\},\{0,1,2,3,4\},\{0,1,2,3,5\},\\
\{0,1,2,4,5\},\{1,2,3,4,5\},\{0,1,2,3,4,5\}\}\\\hline
\{\{0,1,2\},\{3,4,5\},\{0,1,2,3\},\{0,1,2,4\},\{0,1,3,4\},\{0,1,2,5\},\{0,1,2,3,4\},\{0,1,2,3,5\},\\
\{0,1,2,4,5\},\{0,1,3,4,5\},\{0,1,2,3,4,5\}\}\\\hline
\{\{0,1,2\},\{3,4,5\},\{0,1,2,3\},\{0,1,2,4\},\{0,1,2,5\},\{0,3,4,5\},\{0,1,2,3,4\},\{0,1,2,3,5\},\\
\{0,1,2,4,5\},\{0,1,3,4,5\},\{0,2,3,4,5\},\{0,1,2,3,4,5\}\}\\\hline
\{\{0,1,2\},\{0,3,4\},\{0,1,2,3\},\{0,1,2,4\},\{0,1,3,4\},\{0,1,2,5\},\{0,3,4,5\},\{0,1,2,3,4\},\\
\{0,1,2,3,5\},\{0,1,2,4,5\},\{0,1,3,4,5\},\{0,1,2,3,4,5\}\}\\\hline
\{\{0,1,2\},\{0,1,3\},\{0,1,2,3\},\{0,1,2,4\},\{0,1,3,4\},\{0,1,2,5\},\{0,1,3,5\},\{0,1,2,3,4\},\\
\{0,1,2,3,5\},\{0,1,2,4,5\},\{0,1,3,4,5\},\{0,2,3,4,5\},\{0,1,2,3,4,5\}\}\\\hline
\{\{0,1,2\},\{0,1,3\},\{0,1,2,3\},\{0,1,2,4\},\{0,1,3,4\},\{0,1,2,5\},\{0,1,3,5\},\{0,1,4,5\},\\
\{0,1,2,3,4\},\{0,1,2,3,5\},\{0,1,2,4,5\},\{0,1,3,4,5\},\{0,1,2,3,4,5\}\}\\\hline
\{\{0,1,2\},\{0,1,3\},\{2,4,5\},\{0,1,2,3\},\{0,1,2,4,5\},\{0,1,3,4,5\},\{0,1,2,3,4,5\}\}\\\hline
\{\{0,1,2\},\{0,1,3\},\{2,4,5\},\{0,1,2,3\},\{0,1,2,3,4\},\{0,1,2,3,5\},\{0,1,2,4,5\},\{0,1,2,3,4,5\}\}\\\hline
\{\{0,1,2\},\{0,1,3\},\{2,4,5\},\{0,1,2,3\},\{2,3,4,5\},\{0,1,2,4,5\},\{0,1,2,3,4,5\}\}\\\hline
\{\{0,1,2\},\{0,1,3\},\{2,4,5\},\{0,1,2,3\},\{0,1,2,4\},\{0,1,2,3,4\},\{0,1,2,4,5\},\{0,1,2,3,4,5\}\}\\\hline
\{\{0,1,2\},\{0,1,3\},\{2,4,5\},\{0,1,2,3\},\{0,1,3,4\},\{0,1,2,3,4\},\{0,1,2,4,5\},\{0,1,2,3,4,5\}\}\\\hline
\end{longtable}
\end{footnotesize}
\end{center}
\end{definition}

\begin{theorem}
  \label{thm:FCsixNonFCsix}
  All families of $\FCsix$ are FC-families. All families of
  $\nonFCsix$ are nonFC-families.
\end{theorem}
\begin{proof}
  The first part is proved by direct computation based on Theorem
  \ref{thm:FC_uce_shares_nonneg} and Proposition
  \ref{lemma:ssn_correct} (due to the lack of space, weights are not
  printed). The other part is proved by direct computation based on
  Theorem \ref{thm:nonFC} (due to the lack of space, families $F_i$
  and coefficients $c_i$ are not printed)\footnote{All data used for
    proofs can be found in our original proof documents, available
    online at \url{http://argo.matf.bg.ac.rs/formalizations/FCFamilies.zip}}.
\end{proof}

Collections $\FCsix$ and $\nonFCsix$ have some other nice properties
(that we do not formally prove within Isabelle/HOL as they are not
necessary for our main theorem). Every family is canonical (the
smallest family in lexicographic order among the families obtained by
applying all permutations of $\setn{n}$ to it, where the sets are
compared first by their number of ellements, and if the number of
elements is the same, then lexicographically). All families in
$\FCsix$ and all families in $\nonFCsix$ are irreducible. All families
in $\FCsix$ are minimal FC-families i.e., no family is FC-covered by
other families in $\FCsix$. All families in $\nonFCsix$ are maximal
nonFC-families i.e., no family is nonFC-covered by other families in
$\nonFCsix$. Families in the table are printed as they are discovered
--- in lexicographic order of their $L$-partition lists.

The next lemma gives a full characterization of families of $\nfam{6}$
that we call \emph{semi-uniform}. These are the families such that
their FC-status (whether they are FC or nonFC) is known only from the
number of members of certain cardinality (and it does not depend on
the arrangement of elements in these family members).

\begin{definition}\ \\
  $\mathcal{L}_F = \{$ 
  $[0, 0, 0, 0, 5, 6, 0]$,
  $[0, 0, 0, 0, 7, 0, 0]$,
  $[0, 0, 0, 1, 6, 5, 0]$,
  $[0, 0, 0, 2, 0, 6, 0]$,\\
  $[0, 0, 0, 3, 0, 4, 0]$,
  $[0, 0, 0, 3, 2, 3, 0]$,
  $[0, 0, 0, 3, 3, 0, 0]$,
  $[0, 0, 0, 4, 0, 0, 0]$,\\
  $[0, 0, 1, 0, 0, 0, 0]$,
  $[0, 1, 0, 0, 0, 0, 0]$ $\}$.

  \noindent $\mathcal{L}_N = \{$
  $[0, 0, 0, 0, 3, 6, 1]$,
  $[0, 0, 0, 0, 4, 1, 1]$,
  $[0, 0, 0, 1, 1, 6, 1]$,
  $[0, 0, 0, 1, 2, 1, 1]$,\\
  $[0, 0, 0, 2, 0, 1, 1]$  $\}$.
\end{definition}

For example, by Theorem \ref{thm:FCsixNonFCsix}, Lemma
\ref{lemma:FCcoveredFC}, and the following lemma (Lemma
\ref{lemma:allFCcovered}), since
$[0, 0, 0, 0, 5, 6, 0] \in \mathcal{L}_F$, it holds that if a family
contains $5$ four-element, and $6$ five-element sets (all contained in
a six-element set), then it is an FC-family. Similarly, since
$[0, 0, 0, 0, 3, 6, 1] \in \mathcal{L}_N$ it holds that that if all
sets of a family (all contained in a six-element set) have at least
four elements and the family contains only up to $3$ four-element
sets, $6$ five-element sets and $1$ six-element sets, then it is not
an FC-family.

\begin{lemma}
  \label{lemma:allFCcovered}
  For every list $L \in \mathcal{L}_F$ all $L$-partitioned
  families of $\nfam{6}$ are FC-covered by $\FCsix$. For every
  list $L \in \mathcal{L}_N$, all $L$-partitioned families of
  $\nfam{6}$ are nonFC-covered by $\nonFCsix$.
\end{lemma}

The proof of this lemma is available in the Appendix, and is based on
generating (by applying the recursive enumeration) an iso-base of all
irreducible, $L$-partitioned families (for all $L \in \mathcal{L}_F$)
that are not covered by $\FCsix$ and showing that it is empty, and on
generating (again by applying the recursive enumeration) an iso-base
of all irreducible, $L$-partitioned families (for all
$L \in \mathcal{L}_N$) and showing that all its elements are covered
by $\nonFCsix$.

Lists in $\mathcal{L}_F$ give sufficient conditions for a family to
satisfy Frankl's condition and whenever a family extends some of these
families it is know that it is an FC-family. Therefore, we can focus
our attention only the families that are $L$-partitioned for lists
that are less then lists in $\mathcal{L}_F$.
\begin{definition}
  $\mathfrak{L}_F = \{L.\ L \preceq [1, 6, 15, 20, 15, 6, 1]\ \And\
  \not\exists L' \in \mathcal{L}_F.\ L \succeq L'\}$
\end{definition}

Note that pruning is very efficient and only a very small percentage
of possible lists belongs to $\mathfrak{L}_F$ (out of 4704 lists that
do not allow empty sets, singletons and doubletons, only 296 are in
$\mathfrak{L}_F$ --- if singletons and doubletons are allowed, then
there are more than a million possible lists). Therefore, most
families in $\nfam{6}$ are FC-families.

Moreover, as shown by the following lemma proved in the Appendix, it
suffices to consider only a iso-representing set of irreducible,
L-partitioned families, for $L \in \mathfrak{L}_F$ that contain no
empty set.

\begin{lemma}
  \label{lemma:nonAllFCpartitions}
  If for all $L \in \mathfrak{L}_F$, there exists a collection
  $\mathcal{F}^L_b$ that iso-represents $\LnPirsix$ such that
  $\covered{\mathcal{F}^L_b}{\FCsix}{\nonFCsix}$, then
  $\covered{\nfam{6}}{\FCsix}{\nonFCsix}$.
\end{lemma}

Finally, we can show that all families of $\nfam{6}$ are covered by
our collections $\FCsix$ and $\nonFCsix$. The proof is given in the
Appendix, and relies on using dynamic programming enumeration to
enumerate all elements of iso-bases of irreducible, L-partitioned
families for $L \in \mathfrak{L}_F$ that contain no empty set and are
not covered by $\FCsix$, and then showing that all of them are covered
by $\nonFCsix$.

\begin{theorem}
\label{thm:allcovered}
  $\covered{\nfam{6}}{\FCsix}{\nonFCsix}$
\end{theorem}

\section{Experimental results}
\label{sec:experiments}

All experiments have been done on a notebook computer with Intel(R)
Core(TM) 2.3GHz CPU with 4MB RAM memory, running Linux.

In the first phase minimal FC and maximal nonFC-families were
automatically detected, as described in Section
\ref{sec:procedureFCstatus}. The process took around 150 minutes, and
most of the time was consumed by the SCIP ILP solver for checking if
there is a family with the negative share wrt. the current candidate
weight function. During the process, the status of 197 minimal FC
families was checked, along with the status of 1125 nonFC-families
(these families form an iso-base of all irreducible nonFC-families
with up to six elements). Note that in the region bounded from above
by $\mathcal{L}_F$ there are 12877 FC-families that form an iso-base
of all irreducible FC-families in that region, confirming that even in
that region there are much more FC than nonFC-families, and that it is
very important that during enumeration all families covered by smaller
FC-families are excluded, so that only minimal FC-families are
explicitly checked. Possible optimizations might include switching
from integer to rational weights and reducing the number of
nonFC-families that need to be explicitly checked (e.g., by carefully
walking along the line between maximal nonFC and minimal FC-families).

The Isabelle/HOL formalization consumes around 1,2MB organized into
around 20,000 lines of Isabelle/Isar proof text (approximately a half
of that are automatically generated proofs for 197 minimal FC-families
and 115 maximal nonFC-families). Total proof checking time by
Isabelle/HOL takes around 30 minutes. The major fraction of this time
goes to checking the proofs of 197 minimal FC-families (around 6
minutes) and 115 maximal nonFC-families (around 16 minutes), and for
proving that all other families are covered (around 4 minutes).

This is significantly slower than unverified programs that preform the
same calculations. The big difference is due to the use of
machine-integers supporting atomic bitwise-or for finding set unions
(and that operation is heavily trough out the whole
formalization). Therefore, the proof checking time could be
significantly reduced if machine-integers were also used in
Isabelle/ML (a support for this has been added to Isabelle recently
\cite{isabelle-machine-words}).

Interestingly, although proving the status of nonFC-families does not
involve search and proving the status of FC-families does, it turned
out that nonFC-families consume more time and that currently, the most
demanding part was to prove that all witness families belong to the
union-closed extension (the first point of Theorem \ref{thm:nonFC}).

There is much room for improving the proof checking efficiency, but we
did not do that since most of time is consumed by automated
classification procedure and it should be the main focus for further
optimization.

\section{Conclusions and Further Work}
\label{sec:conclusions}

In this paper, we have described a fully automated and mechanically
verified method for classifying families into Frankl-complete (FC) and
non Frankl-complete (nonFC), and applied it to obtain a full
characterization of all families over a six element universe.

We have shown that status of any family over the six-element universe
can be easily determined by knowing the status of only a very small
number of characteristic families (FC-minimal and nonFC-maximal
families) and we have shown that our list of 197 FC and 115
nonFC-families covers all $2^{2^6}$ families over the six-element
universe (their vast majority being FC). All known FC-families are
confirmed and a new uniform FC-family is discovered (as a simple
corollary of our classification we have that FC($4$, $6$) = $7$).

Compared to the prior pen-and-paper work \cite{frankl-morris}, the
computer assisted approach significantly reduces the complexity of
mathematical arguments behind the proof and employs
computing-machinery in doing its best --- quickly enumerating and
checking a large search space. This enables formulation of a general
framework for checking various FC-families (and finite cases of
Frankl's conjecture), without the need of employing human intellectual
resources in analyzing features of specific families.

The method fully is formalized (within Isabelle/HOL), and all our
results are fully mechanically verified. Apart from achieving the
highest level of trust possible, the significant contribution of the
formalization is the clear separation of mathematical background and
combinatorial search algorithms, not present in earlier work. Also,
separation of abstract properties of search algorithms and technical
details of their implementation significantly simplifies reasoning
about their correctness and brings them much closer to classic
mathematical audience, not inclined towards computer science. We have
also shown that efficient unverified procedures (such as ILP packages
or SMT solvers) can freely be used during search if they are able to
produce results and certificates that are independently checked and
verified by proof-assistant.

Some formalized concepts about set families (e.g., concept related to
family isomorphisms or irreducibility) might be useful in other
applications, out of the context of Frankl's conjecture. The same
holds for procedures for efficient enumeration of all families
satisfying certain properties, that are described in Section
\ref{sec:enum}.

We assume that techniques introduced in this work can be adapted to
obtain a full characterization of all families over a 7-element
universe, but that would require a significantly higher computing
power (a cluster computer working more days). We also assume that the
full classification of all families over the 8-element universe is not
possible with the current approach and technology.

Methods used in this paper could be adapted to formally and
automatically prove finite cases of Frankl's conjecture. For example,
\v Zivkovi\' c and Vu\v ckovi\' c have informally shown that Frankl's
conjecture holds for families $F$ such that $\card{\Union{F}} \le 12$
\cite{frankl-zivkovic-vuckovic}, and now their results can be
confirmed fully automatically, within a proof assistant. We also
assume that the automated formalized methods developed in this paper
might enable us the check the conjecture for the case
$\card{\Union{F}} \le 13$ (also assuming a high computing power).


\appendix

\section{Proofs of lemmas and theorems}

\subsection{Proof of Theorem \ref{thm:FC_uce_shares_nonneg}}

In this section we describe proof of Theorem
\ref{thm:FC_uce_shares_nonneg}. First we introduce some auxiliary
notions.

\paragraph{Hypercubes}
Sets of a family can be grouped into so called hypercubes.
\begin{definition}
  An $S$-\emph{hypercube} with a base $K$, denoted by $\hc{K}{S}$, is
  the family $\{A.\ K \subseteq A \And A \subseteq K \union
  S\}$. Alternatively, a hypercube can be characterized by $\hc{K}{S}
  = \{K \union A.\ A \in \pow{S}\}$.
\end{definition}

\begin{example}
  \label{ex:hypercube}
  Let $S \equiv \{s_0, s_1\}$, and $K \equiv \{k_0, k_1\}$. If $K'
  \subseteq K$, then all $S$-hypercubes with a base $K'$ are:
  \begin{small}
  \begin{eqnarray*}
    \hc{\{\}}{S} &=& \{\{\}, \{s_0\}, \{s_1\}, \{s_0, s_1\}\}\\
    \hc{\{k_0\}}{S} &=& \{\{k_0\}, \{k_0, s_0\}, \{k_0, s_1\}, \{k_0, s_0, s_1\}\}\\
    \hc{\{k_1\}}{S} &=& \{\{k_1\}, \{k_1, s_0\}, \{k_1, s_1\}, \{k_1, s_0, s_1\}\}\\
    \hc{\{k_0, k_1\}}{S} &=& \{\{k_0, k_1\}, \{k_0, k_1, s_0\}, \{k_0, k_1, s_1\}, \{k_0, k_1, s_0, s_1\}\}\\
  \end{eqnarray*}
  \end{small}
\end{example}
\vspace{-5mm}

Previous example indicates that (disjoint) $S$-hypercubes can span the
whole $\pow{(K \union S)}$. Indeed, this is generally the case.
\begin{proposition} 
(i) $\pow{(K \union S)} = \bigcup_{K' \subseteq K} \hc{K'}{S}$.
(ii) If $K_1$ and $K_2$ are different and disjoint with $S$, then
    $\hc{K_1}{S}$ and $\hc{K_2}{S}$ are disjoint.
\end{proposition}

Families of sets can be separated into (disjoint) parts belonging to
different hypercubes (formed as $\hc{K}{S} \inter F$).
\begin{definition}
  A \emph{hyper-share of a family $F$ wrt.~weight function $w$, the
    hypercube $\hc{K}{S}$ and the set $X$}, denoted by
  $\hs{K}{S}{F}{w}{X}$, is the value $\sum_{A \in \hc{K}{S} \inter
    F}\ss{A}{w}{X}$.
\end{definition}

\begin{example}
  \label{ex:hypershare}
  Let $S$ and $K$ be as in the Example \ref{ex:hypercube}, let $X
  \equiv K \union S$, let $F \equiv \{\{s_0\}, \{s_1\}, \{k_0, s_0\},
  \{k_0, k_1, s_0, s_1\}\}$, and $w(a) = 1$ for all $a \in X$.  Then,
  $\hs{\{\}}{S}{F}{w}{X} = \ss{\{s_0\}}{w}{X} + \ss{\{s_1\}}{w}{X} =
  -4$, $\hs{\{k_0\}}{S}{F}{w}{X} = \ss{\{k_0, s_0\}}{w}{X} = 0$,
  $\hs{\{k_1\}}{S}{F}{w}{X} = 0$, and $\hs{\{k_0, k_1\}}{S}{F}{w}{X}$
  $=$ $\ss{\{k_0, k_1, s_0, s_1\}}{w}{X} = 4$.
\end{example}

Share of a family can be expressed in terms of sum of hyper-shares.
\begin{proposition}
  \label{prop:Family_share_Hyper_sher}
  If $K \union S = \Union{F}$ and $K \inter S = \emptyset$,
  then $$\fs{F}{w}{(\Union{F})} = \sum_{K' \subseteq K}
  \hs{K'}{S}{F}{w}{(\Union{F})}.$$
\end{proposition}

\begin{proposition}
\label{lemma:Frankl_all_Hyper_share_ge_0}
Let $w$ be a weight function on $\Union{F}$.  If $K \union S =
\Union{F}$, $K \inter S = \emptyset$, and $\forall K' \subseteq
K.\ \hs{K'}{S}{F}{w}{(\Union{F})} \ge 0$, then $\frankl{F}$.
\end{proposition}

\begin{definition}
  \emph{Projection of a family $F$ onto a hypercube $\hc{K}{S}$},
  denoted by $\hcprj{K}{S}{F}$, is the set $\{A - K.\ A \in \hc{K}{S}
  \inter F\}$.
\end{definition}

\begin{example}
  \label{ex:hcprj}
  Let $K$, $S$ and $F$ be as in Example \ref{ex:hypershare}. Then
  $\hcprj{\{\}}{S}{F} = \{\{s_0\}, \{s_1\}\}$, $\hcprj{\{k_0\}}{S}{F}
  = \{\{s_0\}\}$, $\hcprj{\{k_1\}}{S}{F} = \{\}$, and $\hcprj{\{k_0,
    k_1\}}{S}{F} = \{\{s_0, s_1\}\}$.
\end{example}

\begin{proposition}\label{prop:hcprj}
\hfill
\begin{enumerate}
\item If $K \inter S = \emptyset$ and $K' \subseteq K$, then
  $\hcprj{K'}{S}{F} \subseteq \pow{S}$
\item If $\uc{F}$, then $\uc{(\hcprj{K}{S}{F})}$.
\item If $\uc{F}$, $F_c \subseteq F$, $S = \Union F_c$, $K \inter S =
  \emptyset$, then $\uca{F_c}{(\hcprj{K}{S}{F})}$.
\item If $\forall x \in K.\ w(x) = 0$, then $\hs{K}{S}{F}{w}{X} =
  \fs{\hcprj{K}{S}{F}}{w}{X}$.
\end{enumerate}
\end{proposition}

\begin{lemma}
  \label{lemma:Frankl_Min_Family_share_ge_0}
  Let $F$ be a non-empty union-closed family, and let $F_c$ be a
  subfamily (i.e., $F_c \subseteq F$). Let $w$ be a weight function on
  $\Union{F}$, that is zero for all elements of
  $\Union{F} - \Union{F_c}$.  If shares of all union-closed extensions
  of $F_c$ are nonnegative, then there is an element
  $a \in \Union{F_c} \subseteq \Union{F}$ such that it satisfies the
  Frankl's condition for $F$, i.e., if
  $\forall F' \in \uce{F_c}.\ \fs{F'}{w}{(\Union{F_c})} \ge 0$, then
  $\frankl{F}$.
\end{lemma}

\begin{proof}
  Let $S$ denote $\Union{F_c}$, and let $K$ denote
  $\Union{F} - \Union{F_c}$. Since, $K \union S = \Union F$ and
  $K \inter S = \emptyset$, by Proposition
  \ref{lemma:Frankl_all_Hyper_share_ge_0}, it suffices to show that
  $\forall K' \subseteq K.\ \hs{K'}{S}{F}{w}{(\Union{F})} \ge 0$.  Fix
  $K'$ and assume that $K' \subseteq K$. Since $w$ is zero on $K$, by
  Proposition \ref{prop:hcprj}, it holds that
  $\hs{K'}{S}{F}{w}{(\Union{F})} =
  \fs{\hcprj{K'}{S}{F}}{w}{(\Union{F})}$.
  On the other hand, since $\uc{F}$, $F_c \subseteq F$, and
  $K \inter S = \emptyset$, by Proposition \ref{prop:hcprj} it holds
  that $\uca{F_c}{(\hcprj{K'}{S}{F})}$.  Moreover,
  $\hcprj{K'}{S}{F} \subseteq \pow{S}$, so
  $\hcprj{K'}{S}{F} \in \uce{F_c}$. Then,
  $\fs{\hcprj{K'}{S}{F}}{w}{(\Union{F_c})} \ge 0$ holds from the
  assumption. However, since $w$ is zero on $K$, it holds that
  $\sw{w}{\Union{F_c}} = \sw{w}{\Union{F}}$ and
  $\fs{\hcprj{K'}{S}{F}}{w}{(\Union{F})} =
  \fs{\hcprj{K'}{S}{F}}{w}{(\Union{F_c})} \ge 0$
\end{proof}

Finally, we can easily prove Theorem \ref{thm:FC_uce_shares_nonneg}.

\setcounter{theorem}{0}
\begin{theorem}
A family $F_c$ is an FC-family if there is a weight function $w$ such
that shares (wrt.~$w$ and $\Union F_c$) of all union-closed extension
of $F_c$ are nonnegative.
\end{theorem}
\begin{proof}
  Consider an arbitrary union-closed family $F \supseteq F_c$. Let $w$
  be the weight function such that $\forall F' \in \uce{F_c}.\
  \fs{F'}{w}{(\Union{F_c})} \ge 0$. Let $w'$ be a function equal to
  $w$ on $\Union F_c$ and 0 on other elements. Since $\forall F' \in
  \uce{F_c}.\ \fs{F'}{w'}{(\Union{F_c})} = \fs{F'}{w}{(\Union{F_c})}$,
  Lemma \ref{lemma:Frankl_Min_Family_share_ge_0} applies to $F$ and
  there is an element $a \in \Union{F_c}$ that satisfies the Frankl's
  condition for $F$. Therefore, $F_c$ is an FC-family.
\end{proof}

\subsection{Proof of Theorem \ref{thm:nonFC}}

\begin{theorem}
  Assume that $F_c$ is a union-closed family. If there exists a
  sequence of families $F_0, \ldots, F_k$, and a sequence of natural
  numbers $c_0, \ldots, c_k$ that:
  \begin{enumerate}
  \item for all $0 \leq i \leq k$ it holds that $F_i \in \uce{F_c}$,
  \item for every $a \in \Union{F_c}$ it holds that $$\sum_{i=0}^k c_i \cdot (2 \cdot \cnt{a}{F_i} - \card{F_i}) < 0,$$
  \item not all $c_i$ are zero (i.e., $\exists i.\ 0 \le i \le k \wedge c_i > 0$),
  \end{enumerate}
  then the family $F_c$ is not an FC-family.
\end{theorem}

\begin{proof}
  Let $c = \sum_{i=0}^kc_i$. From the assumptions, it holds that $c >
  0$. For each natural number $d > 0$, let $B^d = \{b_0, \ldots,
  b_{c\cdot d}\}$ be a set containing $c\cdot d + 1$ elements, having
  no common elements with $\Union{F_c}$. For each $0 \leq s < c\cdot
  d$, let $B^d_s = B^d \setminus \{b_s\}$. Let $G^d$ be a sequence of
  sets $[F_1, \ldots, F_1, F_2, \ldots, F_2, \ldots, F_k, \ldots,
  F_k]$, where each $F_i$ is repeated exactly $c_i\cdot d$ times, and
  for each $0 \leq s < c\cdot d$, let $G^d_s$ be the $s$-the member of
  the sequence $G^d$. For each $0 \leq s < c\cdot d$, let $H^d_s = \{A
  \union B^d_s.\ A \in G^d_s\}$. Let $H^d = \{A \union B^d.\ A \in
  \pow{\Union{F_c}}\}$. Finally, we form a family $F^d = F_c \union
  (\bigcup_{0\leq s < c\cdot d} H^d_s) \union H^d$. For an
  appropriately chosen $d$ it will be a counterexample that $F_c$ is
  an FC-family. 

  For each $d$, the family $F^d$ is union-closed ($F_c$ is union
  closed, so the union of each two sets from $F_c$ is in $F_c$, $H^d$
  is also union-closed, so the union of each two sets from $F_c$ is in
  $H^d$, the union of each two sets from $\bigcup_{0\leq s < c\cdot d}
  H^d_s$ is in $H^d$, unions of sets from $F_c$ and $\bigcup_{0\leq s
    < c\cdot d} H^d_s$ are in $\bigcup_{0\leq s < c\cdot d} H^d_s$
  since for some $i$, it holds that $G^d_s = F_i \in \uce{F_c}$,
  unions of sets from $F_c$ and $H^d$ are in $H^d$ and unions of sets
  from $\bigcup_{0\leq s < c\cdot d} H^d_s$ and $H^d$ are in $H^d$).

  Let $f$ be a function defined by $f\ a\ F = 2\cdot \cnt{a}{F} -
  \card{F}$.  Since $B^d$ and $\Union{F_c}$ are disjoint, for each $a
  \in \Union{F_c}$, $d > 0$ and $0 \leq s < c\cdot d$ it holds that
  $f\ a\ H^d_s = f\ a\ G^d_s$. As $H^d$ is built around the whole
  $\pow{\Union{F_c}}$, it can be easilly shown that for each $a \in
  \Union{F_c}$ and $d>0$, it holds that $f\ a\ H^d = 0$. For each $d >
  0$, the families $F_c$, $\bigcup_{0\leq s < c\cdot d} H^d_s$, and
  $H^d$ are mutually disjoint (all $B^d_s$ are non-empty, sets in
  $F_c$ contain no element from $B^d$, sets from $\bigcup_{0\leq s <
    c\cdot d} H^d_s$ contain all but one element from $B^d$, while
  sets from $H^d$ contain the whole $B^d$). Also, for each $0\leq s_1
  \neq s_2 < c\cdot d$, the families $H^d_{s_1}$ and $H^d_{s_2}$ are
  disjoint (as all sets from $H^d_{s_1}$ contain $b_{s_2}$, while none
  of the sets from $H^d_{s_2}$ does). Therefore, for each $a \in
  \Union{F_c}$ and $d > 0$, $f\ a\ F^d = f\ a\ F_c + f\ a\
  (\bigcup_{0\leq s < c\cdot d} H^d_s) + f\ a\ H^d = f\ a\ F_c +
  \sum_{0\leq s < c\cdot d} f\ a\ G^d_s = f\ a\ F_c + d \cdot
  \sum_{0\leq s < c} f\ a\ G^d_{s\cdot d}$. By construction of $G^d$,
  it holds that for each $0 \leq i \leq k$, the last sum has exactly
  $c_i$ terms $f\ a\ F_i$. Therefore, $f\ a\ F^d = f\ a\ F_c + d \cdot
  (\sum_{0 \leq i \leq k}c_i \cdot (f\ a\ F_i))$. By assumption,
  $\sum_{0 \leq i \leq k}c_i \cdot (f\ a\ F_i) = \sum_{0 \leq i \leq
    k}c_i \cdot (2 \cdot \cnt{a}{F_i} - \card{F_i})$ is
  negative. Therefore, for each $a$, there exists a $d_a > 0$ such
  that $f\ a\ F^{d_a} < 0$. Let $d$ maximal $d_a$ for $a \in
  \Union{F_c}$. It holds that $F^d$ is a union-closed family
  containing $F_c$, such that each $a \in \Union{F_c}$ it holds that
  $f\ a\ F^d = (2 \cdot \cnt{a}{F^d} - \card{F^d}) < 0$, i.e.,
  $\cnt{a}{F^d} < \card{F^d}/2$. This shows that $F_c$ cannot be an
  FC-family.
\end{proof}

\subsection{Proof of Theorem \ref{thm:mult}}

\begin{theorem}
  Assume that the predicate $\Pinc{P}$ incrementally checks $P$. Then
  $$\LnP{[l_0, \ldots, l_m+1]}{n}{P} = \mult{\LnP{[l_0, \ldots, l_m]}{n}{P}}{\families{n}{m}}{\Pinc{P}}.$$
\end{theorem}

\begin{proof}
  First we show that $\LnP{[l_0, \ldots, l_m+1]}{n}{P} \subseteq
  \mult{\LnP{[l_0, \ldots, l_m]}{n}{P}}{\families{n}{m}}{\Pinc{P}}$. Let $F'$
  be an arbitrary family in $\LnP{[l_0, \ldots, l_m+1]}{n}{P}$. Then
  $F'$ is $[l_0, \ldots, l_m+1]$-partitioned, $\Union F' \subseteq
  \setn{n}$, and $P\ F'$ holds. Since $l_m + 1 > 0$, there is an
  element $A\in F'$ such that $\card{A} = m$. Since $\Union F'
  \subseteq \setn{n}$ and $A \in F'$, it holds that $A \subseteq
  \setn{n}$, so $A \in \families{n}{m}$.  Let $F$ denote the set $F'
  \setminus \{A\}$. Since $A \in F'$, it holds that $F' = F \union
  \{A\}$. Since $\Pinc{P}$ incrementally checks $P$, and since for all $A'
  \in F'$, it holds that $\card{A} = m \geq \card{A'}$, and since $A
  \notin F$, from $P\ F'$, by relation (\ref{eq:PQ}) it holds that $P\
  F$ and $\Pinc{P}\ F\ A$.  The set $F$ is in $\LnP{[l_0, \ldots,
    l_m]}{n}{P}$. Indeed, it is $[l_0, \ldots, l_m]$-partitioned, it
  holds that $\Union F \subseteq \Union F' \subseteq \setn{n}$, and
  $P\ F$ holds. Since $A \in \families{n}{m}$, $A \notin F$ and $\Pinc{P}\ F\
  A$, by the definition of $\Pinc{P}$-filtered multiplication $F'$ is in
  $\mult{\LnP{[l_0, \ldots, l_m]}{n}{P}}{\families{n}{m}}{\Pinc{P}}$.

  To prove $\LnP{[l_0, \ldots, l_m+1]}{n}{P} \supseteq
  \mult{\LnP{[l_0, \ldots, l_m]}{n}{P}}{\families{n}{m}}{\Pinc{P}}$, assume
  that $F'$ is an arbitrary element of $\mult{\LnP{[l_0, \ldots,
      l_m]}{n}{P}}{\families{n}{m}}{\Pinc{P}}$. Then there is an element $F
  \in \LnP{[l_0, \ldots, l_m]}{n}{P}$ and $A \in \families{n}{m}$ such
  that $F' = F \union \{A\}$, $A \notin F$, and $\Pinc{P}\ F\ A$.  Therefore
  $F$ is $[l_0, \ldots, l_m]$-partitioned, $P\ F$ holds, $\Union F
  \subseteq \setn{n}$, $\card{A} = m$ and $A \subseteq
  \setn{n}$. Since $\Pinc{P}$ incrementally checks $P$, since for all $A' \in
  F$, it holds that $\card{A} = m \geq \card{A}'$, since $A \notin F$,
  and since $P\ F$ and $\Pinc{P}\ F\ A$ hold, by relation (\ref{eq:PQ}), $P\
  (F \union \{A\})$ must hold. Moreover, since $A \notin F$, since $F$
  is $[l_0, \ldots, l_m]$-partitioned, and since $\card{A} = m$, it
  holds that $F \union \{A\}$ is $[l_0, \ldots,
  l_m+1]$-partitioned. Finally, $\Union F' = \Union F \union A
  \subseteq \setn{n}$. Therefore, $F' = F \union \{A\} \in \LnP{[l_0,
    \ldots, l_m+1]}{n}{P}$.
\end{proof}

\subsection{Proof of Theorem \ref{thm:mult_generates}}

\begin{theorem}
  Assume that the predicate $\Pinc{P}$ incrementally checks $P$ and is
  preserved by injective functions. If $m \le n$ and $\mathcal{F}_b$
  is an iso-representing subcollection of
  $\LnP{[l_0,\ldots, l_m]}{n}{P}$, then
  $\mathcal{F}'_b \equiv
  \mult{\mathcal{F}_b}{\families{n}{m}}{\Pinc{P}}$
  is an iso-representing subcollection of
  $\LnP{[l_0, \ldots, l_m + 1]}{n}{P}$.
\end{theorem}

\begin{proof}
  $\mathcal{F}'_b$ is a subcollection of $\LnP{[l_0, \ldots, l_m +
    1]}{n}{P}$. Indeed, since $\mathcal{F}_b$ is a subcollection of
  $\LnP{[l_0,\ldots, l_m]}{n}{P}$, and $\mathcal{F}'_b \equiv
  \mult{\mathcal{F}_b}{\families{n}{m}}{\Pinc{P}}$, the statement holds by
  Theorem \ref{thm:mult}.

  Let us show that $\mathcal{F}'_b$ iso-represents
  $\LnP{[l_0, \ldots, l_m + 1]}{n}{P}$. Let $F'$ be an arbitrary
  element of $\LnP{[l_0, \ldots, l_m + 1]}{n}{P}$. By Theorem
  \ref{thm:mult}, there is a family
  $F \in \LnP{[l_0, \ldots, l_m]}{n}{P}$ and a set
  $A \in \families{n}{m}$ such that $F' = F \union {A}$, $A \notin F$,
  $\Pinc{P}\ A\ F$. Since $\mathcal{F}_b$ iso-represents
  $\LnP{[l_0, \ldots, l_m]}{n}{P}$, there is a family
  $F_b \in \mathcal{F}_b$ such that $\iso{F}{F_b}$, i.e., there is a
  bijection $f$ between $\Union F$ and $\Union F_b$ such that
  $F_b = \mapfam{f}{F}$.

  There is a function $f'$, extending $f$ from $\Union F$ to $\Union F
  \union A$, such that it is injective on $\Union F \union A$, that
  $\mapfam{f'}{F} = F_b$ and that $\mapfam{f'}{(F \union \{A\})}
  \subseteq \setn{n}$.

  Let $A_b$ denote the set $\mapset{f'}{A}$. The function $f'$
  establishes an isomorphism between $F \union \{A\}$ and
  $\mapfam{f'}{F} \union \{\mapset{f'}{A}\}$, i.e.,
  $F_b \union \{A_b\}$. Therefore $\iso{F'}{F_b \union
    \{A_b}\}$.
  Moreover, $F_b \union \{A_b\}$ is in $\mathcal{F}'_b$. Indeed, it
  holds that $F_b \in \mathcal{F}_b$. Also, since $f'$ is an injection
  from $\Union F \union A$ into $\setn{n}$, it holds that
  $\card{A_b} = \card{A} = m$, and $A_b \subseteq \setn{n}$, and,
  since $A \notin F$ it holds that
  $\mapset{f'}{A} \notin \mapset{f'}{F}$ i.e., $A_b \notin
  F_b$.
  Finally, since $\Pinc{P}$ is preserved by injective functions it
  also holds that $\Pinc{P}\ (\mapfam{f'}{F})\ (\mapset{f'}{A})$,
  i.e., $\Pinc{P}\ F_b\ A_b$. Therefore, $\mathcal{F}'_b$
  iso-represents $\LnP{[l_0, \ldots, l_m + 1]}{n}{P}$.
\end{proof}

\subsection{Proof of Theorem \ref{lemma:enum_rec_mult}}

The characterization of this procedure (in terms of relation that
connects the list given as its parameter and the returned value for
corresponding to that list) is given by the following
proposition. This proposition is proved using mathematical induction,
based on the definition of $\enumrecname$.

\begin{proposition}
  \label{prop:enum_rec}
  Let $L$ be a given list, and $R$ be a relation between a value and a
  list. If
  \begin{enumerate}
  \item $R\ v_{\emptylist}\ \emptylist$,
  \item for all $v'$ and $L'$, if $R\ v'\ L'$, then $R\ v'\ [L', 0]$,
  \item for all $v'$ and $L'$, if $L'$ is not empty, $\last{L'} > 0$,
    $\length{L'} \le \length{L}$, and\\ $R\ v'\ (\declast{L'})$, then
    $R\ (upd\ v'\ L')\ L'$,
  \end{enumerate}
  then $R\ (\enumrec{L}{v_{\emptylist}}{upd})\ L$.
\end{proposition}

Finally, we show how can we use $\enumrecname$ to find an iso-base of
$\LnP{{L}}{n}{P}$ for some given list $L$, number $n$, and a predicate
$P$.

\begin{theorem}
  Let $L$ be a list such that $\length{L} \leq n + 1$.  Assume
  that:
  \begin{enumerate}
  \item $P\ \{\}$ holds and $v_{\emptylist} = \{\emptyset\}$,
  \item $\Pinc{P}$ incrementally checks $P$ and is preserved by injective
  functions,
  \item $\mathcal{P}$ contains all permutations of $\listn{n}$, and
  $upd = \lambda\ \mathcal{F}\ L.\
  \base{(\mult{\mathcal{F}}{\families{n}{\length{L}-1}}{\Pinc{P}})}{\mathcal{P}}$.
  \end{enumerate}
  Then $\enumrec{L}{v_{\emptylist}}{upd}$ is an iso-base of
  $\LnP{{L}}{n}{P}$.
\end{theorem}
\begin{proof}
  The result follows by Proposition \ref{prop:enum_rec} ($R\ v\ L$
  holds iff $v$ is an iso-base of $\LnP{L}{n}{P}$). All the conditions
  of Proposition \ref{prop:enum_rec} are met.
  \begin{enumerate}
  \item Since $P\ \{\}$ holds,
    $\LnP{\emptylist}{n}{P} = \{\emptyset\}$, so
    $v_{\emptylist} = \{\emptyset\}$ is its iso-base.
  \item It holds that $\LnP{L}{n}{P} = \LnP{[L, 0]}{n}{p}$. So, if
    some $\mathcal{F}_b$ is an iso-base of $\LnP{L}{n}{P}$, then it is
    also an iso-base of $\LnP{[L, 0]}{n}{p}$.
  \item Let $L'$ be a nonempty list, such that $\last{L'} > 0$, and
    $\length{L'}\le \length{L}$. Then, for some $m \ge 0$,
    $L' = [l_0, \ldots, l_m]$, $l_m > 0$,
    $\declast{L'} = [l_0, \ldots, l_m-1]$, and $\length{L'} - 1 = m$.
    Let $v' \equiv \mathcal{F}_b$ be an iso-base of $\declast{L'}$. By
    the definition of $upd$, it holds that
    $upd\ \mathcal{F}_b\ L' =
    \base{(\mult{\mathcal{F}_b}{\families{n}{m}}{\Pinc{P}})}{\mathcal{P}}$.
    Since $m \le n$, by Theorem \ref{thm:mult_generates},
    $\mult{\mathcal{F}_b}{\families{n}{m}}{\Pinc{P}}$ is an
    iso-representing set of $\LnP{L'}{n}{P}$. Therefore, since
    $\mathcal{P}$ contains all permutations of $\listn{n}$ and
    $\mult{\mathcal{F}_b}{\families{n}{m}}{\Pinc{P}} \subseteq
    \nfam{n}$,
    $\base{(\mult{\mathcal{F}_b}{\families{n}{m}}{\Pinc{P}})}{\mathcal{P}}$
    is an iso-base of $\LnP{L'}{n}{P}$.
  \end{enumerate}
\end{proof}

\subsection{Proof of Theorem \ref{lemma:enum_dp_mult}}

The procedure $\enumdpp{L}{m} {v}{res}{upd}{stop}{L_{max}}$ is
characterized by the following proposition.

\begin{proposition}
  \label{prop:enum_dp}
  Let $R$ be a relation between a value and a list. Let $L_{max}$ be a
  fixed list, and $stop$ a predicate for lists. Let $L \equiv [l_0,
  \ldots, l_m, 0, \ldots, 0]$ be a given (initial) list (i.e., $m <
  \length{L}$, $\forall k.\ m < k < \length{L}\ \longrightarrow\ l_{k}
  = 0$) such that $\length{L} = \length{L_{max}}$. Let $v$ be a
  value. Assume that:
  \begin{enumerate}
  \item $R\ v\ L$,
  \item for all $v'$, $L'$ and $m'$, such that $m' < \length{L}$ and
    $L' \equiv [l'_0, \ldots, l'_{m'}, 0, \ldots, 0]$ (i.e., $m' <
    \length{L'}$, and $\forall k.\ m' < k < \length{L}\
    \longrightarrow\ l'_{k} = 0$) it holds that if $R\ v'\ L'$, then
    $R\ (upd\ v'\ m')\ (\incnth{L'}{m'})$.
  \end{enumerate}
  Then, for all $L'$ such that $L \preceq L' \preceq L_{max}$, such
  that $L$ and $L'$ agree up to the position $m$ (i.e., $\take{m}{L'}
  = \take{m}{L}$), and such that there is no $L_s$ such that $L \preceq
  L_s \preceq L'$, $\take{m}{L_s} = \take{m}{L}$ and $stop\ L_s$,
  there is a $v' \in (\enumdpp{L}{m} {v}{res}{upd}{stop}{L_{max}})$
  such that $R\ v'\ L'$
\end{proposition}

The proof of Theorem \ref{lemma:enum_dp_mult} relies on this Proposition.

\begin{theorem}
  Assume that
  \begin{enumerate}
  \item $\Pinc{P}$ incrementally checks $P$ and is preserved by injective
    functions, 
  \item $P\ \{\}$ holds and $v_{\emptylist} = \{\emptyset\}$,
  \item $L_{max}$ is a list such that $\length{L_{max}} \le n + 1$,
    $\mathcal{L}_{s}$ contain lists (that all have the same length
    $\length{L_{max}}$), and $stop = \lambda\ L.\ (\exists L_s \in
    \mathcal{L}_s.\ L \succeq L_s)$,
  \item $\mathcal{P}$ contains all permutations of $\listn{n}$, $upd =
    \lambda\ \mathcal{F}\ m.\
    \base{\mult{(\mathcal{F}}{\families{n}{m}}{\Pinc{P}})}{\mathcal{P}}$.
  \end{enumerate}
  Let
  $X = \{L.\ L \preceq L_{max} \And (\nexists L_s \in
  \mathcal{L}_{s}.\ L \succeq L_s)\}$.
  Then, for all $L \in X$, there exists an
  $\mathcal{F}_b \in \enumdp{v_{\emptylist}}{upd}{stop}{L_{max}}$ such
  that $\mathcal{F}_b$ is an iso-base of $\LnP{L}{n}{P}$.
\end{theorem}

\begin{proof}
  The proof relies on Proposition \ref{prop:enum_dp} ($R\ v\ L$ holds
  iff $v$ is an iso-base of $\LnP{L}{n}{P}$).  The function
  $\enumdppname$ is called for list $L=[0, \ldots, 0]$ so the initial
  assumptions are trivially satisfied.

  Next we need to show that $v_{\emptylist} = \{\emptyset\}$ is an
  iso-base of $\LnP{[0, \ldots, 0]}{n}{P}$. But, it holds that
  $\LnP{[0, \ldots, 0]}{n}{P} = \LnP{\emptylist}{n}{P}$, and since
  $P \{\}$ holds, $\LnP{\emptylist}{n}{P} = \{\emptyset\}$ so
  $\{\emptyset\}$ is its iso-base.

  Let us show the final assumption. Fix a collection
  $v' \equiv \mathcal{F}_b$, list $L'$ and $m' < \length{L'}$,
  $m' < \length{L}$ such that $L'$ is of the form
  $[l'_0, \ldots, l'_{m'}, 0, \ldots, 0]$ and assume that
  $\mathcal{F}_b$ is an iso-base of $\LnP{L'}{n}{P}$. Then
  $\mathcal{F}_b$ is also an iso-base of
  $\LnP{[l'_0, \ldots, l'_{m'}]}{n}{P}$. It holds that
  $upd\ \mathcal{F}_b\ m' =
  \base{\mult{(\mathcal{F}_b}{\families{n}{m'}}{\Pinc{P}})}{\mathcal{P}}$.
  Since $\Pinc{P}$ incrementally checks $P$ and is preserved by
  injective functions, and since
  $m' < \length{L} = \length{L_{max}} \le n + 1$, by Theorem
  \ref{thm:mult_generates},
  $\mult{\mathcal{F}_b}{\families{n}{m'}}{\Pinc{P}}$ is an
  iso-representing set of $\LnP{[l'_0, \ldots, l'_{m'} + 1]}{n}{P}$.
  Therefore, since $\mathcal{P}$ contains all permutations of
  $\listn{n}$ and
  $\mult{\mathcal{F}_b}{\families{n}{m'}}{\Pinc{P}} \subseteq
  \nfam{n}$,
  the collection
  $\base{\mult{(\mathcal{F}_b}{\families{n}{m}}{\Pinc{P}})}{\mathcal{P}}$
  is an iso-base of $\LnP{[l'_0, \ldots, l'_{m'} + 1]}{n}{P}$. But,
  since,
  $inc\_nth\ L'\ m' = [l'_0, \ldots, l'_{m'} + 1, 0, \ldots, 0]$, that
  is equal to $[l'_0, \ldots, l'_{m'} + 1]$, it holds that
  $\base{\mult{(\mathcal{F}_b}{\families{n}{m'}}{\Pinc{P}})}{\mathcal{P}}$
  is also an iso-base of $\LnP{(inc\_nth\ L'\ m')}{n}{P}$.

  Since $\enumdpname$ calls $\enumdppname$ for $L=[0, \ldots, 0]$ and
  $m=0$, and since $\succeq$ is transitive, the set of all $L'$ such
  that $L \preceq L' \preceq L_{max}$, $\take{m}{L'} = \take{m}{L}$
  and such that there is no $L_s$ such that $L \preceq L_s \preceq
  L'$, $\take{m}{L_s} = \take{m}{L}$ and $stop\ L_s$ is exactly the
  set $X$.  Therefore, the statement holds.
\end{proof}

\setcounter{lemma}{0}

\subsection{Proof of Lemma \ref{lemma:unique_irreducible}}

\begin{lemma}If $F$ and $F'$ are both irreducible families and $\closure{F}
    = \closure{F'}$, then $F = F'$.
\end{lemma}
\begin{proof}
  Let us first show that if $F'$ is an irreducible family and
  $\closure{F} = \closure{F'}$, then $F' \subseteq F$.

  Assume the opposite. Then there is a set $A$ such that $A \in F'$
  and $A \notin F$. Since
  $A \in F' \subseteq \closure{F'} = \closure{F}$, there is a nonempty
  family $F_A \subseteq F$ such that $\Union F_A = A$. The family
  $F_A$ can be split to $F_A^+ = F_A \inter F'$ and
  $F_A^- = F_A \setminus F'$. All elements in $F_A^-$ belong to
  $F_A \subseteq F \subseteq \closure{F} = \closure{F'}$, so for every
  $A' \in F_A^-$ there is a non-empty family $F^{A'} \subseteq F'$
  such that $A' = \Union{F^{A'}}$. Let $G$ be a family consisting of
  $F_A^+$ and union of all such families $F^{A'}$ for each element
  $A' \in F_A^-$.

  The family $G$ is a subfamily of $F'$. Indeed, $F_A^+ = F_A \inter
  F' \subseteq F'$ and for all $F^{A'}$ it holds that $F^{A'}
  \subseteq F'$.

  It holds that $\Union{G} = A$. Namely, it holds that the union of
  all families $F^{A'}$ over all elements $A' \in F_A^-$ is equal to
  the $\Union F_A^-$. Therefore, $\Union{G} = \Union{F_A^+} \union
  \Union{F_A^-} = \Union{F_A} = A$.

  The set $A$ is not in $G$. Assume the opposite. Then, since $A
  \notin F$ and $F_A \subseteq F$, it holds that $A \notin F_A^+
  \subseteq F_A$. Therefore, $A$ must belong to some family $F^{A'}$
  for some $A' \in F_A^-$. Hence, $A' \in F_A$ so $A' \subseteq
  \Union{F_A} = A$. Also, since $\Union{F^{A'}} = A'$ and $A \in
  F^{A'}$ it must be that $A \subseteq A'$. So $A = A'$ and $A \in F_A
  \subseteq F$ which contradicts that $A \notin F$.

  The family $G$ is not empty. Indeed, since $F_A$ is not empty it
  contains a set $A'$. If $A' \in F'$, then $A' \in F_A \inter F' =
  F_A^+ \subseteq G$ so $G$ is not empty. If $A' \notin F'$, then $A'
  \in F_A \setminus F' = F_A^-$. But, then there is a non-empty family
  $F^{A'}$ whose elements are in $G$, so $G$ is not empty.

  From all this, it follows that $A$ depends on elements of
  $F'\setminus\{A\}$. But, since $A \in F'$, this contradicts that
  $F'$ is irreducible, so the initial assumption was wrong and $F'
  \subseteq F$.

  The main statement is a trivial consequence of the one that we have
  just proved.
\end{proof}

\subsection{Proof of Lemma \ref{lemma:FCcoveredFC}}

\begin{lemma}\ \\[-5mm]
  \begin{enumerate}
  \item Any family $F$ that is FC-covered by an FC-family $F_c$ is an
    FC-family.
  \item Any family $F$ that is nonFC-covered by a nonFC-family $N_c$
    is not an FC-family.
  \end{enumerate}
\end{lemma}
\begin{proof}\ \\[-5mm]
\begin{enumerate}
\item Since $\FCcovered{F}{F_c}$, there is a family $F_c'$ such that
  $\iso{F_c}{F_c'}$ and $\closure{F} \supseteq F_c'$. Since $F_c$ is
  an FC-family, so is $F_c'$. Therefore, by Proposition
  \ref{prop:FC_family_mono}, $\closure{F}$ is an FC-family. But, then,
  by Proposition \ref{prop:FC_family_closure}, $F$ is also an $FC$
  family.
\item Since $\nonFCcovered{F}{N_c}$, there is a family $N_c'$ such
  that $\closure{F} \subseteq \closure{N_c'} \union \{\emptyset\}$. If
  $F$ were an FC-family, since
  $F \subseteq \closure{F} \subseteq \closure{N_c'} \union
  \{\emptyset\}$,
  by Proposition \ref{prop:FC_family_mono} and Proposition
  \ref{prop:FC_family_empty}, $N_c'$ would also be FC-family, which is
  a contradiction, as it isomorphic to a nonFC-family $N_c$.
\end{enumerate}
\end{proof}

\subsection{Proof of Lemma \ref{lemma:generates_covered}}
\begin{lemma}
  Assume that $\mathcal{F}_b$ iso-represents $\mathcal{F}$. If
  $\FCcovered{\mathcal{F}_b}{\FC}$, then
  $\FCcovered{\mathcal{F}}{\FC}$. If
  $\nonFCcovered{\mathcal{F}_b}{\nonFC}$, then
  $\nonFCcovered{\mathcal{F}}{\nonFC}$. If
  $\covered{\mathcal{F}_b}{\FC}{\nonFC}$, then
  $\covered{\mathcal{F}}{\FC}{\nonFC}$.
\end{lemma}

\begin{proof}
  Let $F \in X$. Since $\mathcal{F}_b$ iso-represents $X$, there is an
  $F' \in \mathcal{F}_b$ such that $\iso{F}{F'}$. If
  $\FCcovered{\mathcal{F}_b}{\FC}$, then $\FCcovered{F'}{\FC}$. But
  then, by Proposition \ref{prop:covered},
  $\FCcovered{F}{\FC}$. Similarly, if
  $\nonFCcovered{\mathcal{F}_b}{\nonFC}$, then
  $\nonFCcovered{F'}{\nonFC}$. But then, by Proposition
  \ref{prop:covered}, $\nonFCcovered{F}{\nonFC}$. If
  $\covered{\mathcal{F}_b}{\FC}{\nonFC}$, then either
  $\FCcovered{F'}{\FC}$ or $\nonFCcovered{F'}{\nonFC}$. But then, by
  Proposition \ref{prop:covered} either $\FCcovered{F}{\FC}$ or
  $\nonFCcovered{F}{\nonFC}$, so $\covered{F}{\FC}{\nonFC}$.
\end{proof}

\subsection{Proof of Lemma \ref{lemma:indep}}
\begin{lemma}
  If all irreducible families in $\nfam{n}$ are covered by $\FC$ and
  $\nonFC$, then all families in $\nfam{n}$ are covered by $\FC$ and
  $\nonFC$.
\end{lemma}
\begin{proof}
  Fix an arbitrary family $F \in \nfam{n}$. By Proposition
  \ref{prop:ex_indep}, there is an irreducible family $F'$ such that
  $F' \subseteq F$ and $\closure{F} = \closure{F'}$. Since $F' \in
  \nfam{n}$, by assumption it is covered by $\FC$ and $\nonFC$. But
  then, by Proposition \ref{prop:covered_closure}, so is $F$.
\end{proof}

\subsection{Proof of Lemma \ref{lemma:allFCcovered}}

\begin{lemma}
  For every list $L \in \mathcal{L}_F$ all $L$-partitioned
  families of $\nfam{6}$ are FC-covered by $\FCsix$. For every
  list $L \in \mathcal{L}_N$, all $L$-partitioned families of
  $\nfam{6}$ are nonFC-covered by $\nonFCsix$.
\end{lemma}
\begin{proof}
  Let $P = \lambda F.\ \neg (\FCcovered{F}{\FCsix})$ and
  $\Pinc{P} = \lambda\ F\ A.\ \neg (\FCcovered{F \union
    \{A\}}{\FCsix})$.
  Let $v_{\emptylist} = \{\emptyset\}$. Let $\mathcal{P}_6$ contain
  all permutations of $\listn{6}$ and
  $upd = \lambda\ \mathcal{F}\ L.\
  \base{(\mult{\mathcal{F}}{\families{6}{\length{L} -
        1}}{\Pinc{P}})}{\mathcal{P}_6}$.
  All conditions of Lemma \ref{lemma:enum_rec_mult} are met, so it
  holds that $\enumrec{L}{v_{\emptylist}}{upd}$ is an iso-base of
  $\LnP{{L}}{6}{P}$. Evaluating it for all $L \in \mathcal{L}_F$ gives
  $\{\}$. Therefore, for any $L \in \mathcal{L}_F$ there are no
  families in $\nfam{6}$ that are $L$-partitioned an are not
  FC-covered by $\FCsix$.
  
  Let $P$ and $\Pinc{P}$ be $\top$. Let
  $v_{\emptylist} = \{\emptyset\}$. Let $\mathcal{P}_6$ contain all
  permutations of $\listn{6}$ and
  $upd = \lambda\ \mathcal{F}\ L.\
  \base{(\mult{\mathcal{F}}{\families{6}{\length{L} -
        1}}{\Pinc{P}})}{\mathcal{P}_6}$.
  By Lemma \ref{lemma:enum_rec_mult},
  $\enumrec{L}{v_{\emptylist}}{upd}$ is an iso-base of all
  $L$-partitioned families of $\nfam{6}$. After evaluating it for any
  $L \in \mathcal{L}_N$, a direct computation shows that all its
  members are nonFC-covered by $\nonFCsix$.
\end{proof}

\subsection{Proof of Lemma \ref{lemma:nonAllFCpartitions}}

\begin{lemma}
  If for all $L \in \mathfrak{L}_F$, there exists a collection
  $\mathcal{F}^L_b$ that iso-represents $\LnPirsix$ such that
  $\covered{\mathcal{F}^L_b}{\FCsix}{\nonFCsix}$, then
  $\covered{\nfam{6}}{\FCsix}{\nonFCsix}$.
\end{lemma}
\begin{proof}
  First we prove that for all $L \in \mathfrak{L}_F$, it holds that
  $\covered{\LnPirsix}{\FCsix}{\nonFCsix}$. Let $L$ be an arbitrary
  list in $\mathfrak{L}_F$. By assumption there exists
  $\mathcal{F}^L_b$ that iso-represents $\LnPirsix$, such that
  $\covered{\mathcal{F}^L_b}{\FCsix}{\nonFCsix}$. Then, by Lemma
  \ref{lemma:generates_covered}, it holds that
  $\covered{\LnPirsix}{\FCsix}{\nonFCsix}$.
  
  To show that for all $F \in \nfam{6}$ it holds that
  $\covered{F}{\FCsix}{\nonFCsix}$, by Lemma \ref{lemma:indep} it
  suffices to show that all irreducible families in $\nfam{6}$ are
  covered by $\FCsix$ and $\nonFCsix$.

  Fix an arbitrary irreducible family $F$ in $\mathcal{F}_b$. By
  Proposition \ref{prop:lpart}, $F$ is $L$-partitioned for some
  $L = [l_0, \ldots, l_6]$ and $L \preceq [1, 6, 15, 20, 15, 6, 1]$.
  
  If there is a list $L' \in \mathcal{L}_F$ such that $L \succeq L'$,
  then, by Lemma \ref{lemma:allFCcovered} for all $F'$ in $\nfam{6}$
  that are $L'$-partitioned it holds that $\FCcovered{F'}{\FCsix}$, so
  by Proposition \ref{prop:covered_mono} it holds that
  $\FCcovered{F}{\FCsix}$ (as $F$ is $L$-partitioned and
  $L \succeq L'$) and therefore $\covered{F}{\FCsix}{\nonFCsix}$.

  Otherwise, $L \in \mathfrak{L}_F$, and by our first proved statement
  it holds that $\covered{\LnPirsix}{\FCsix}{\nonFCsix}$. It holds
  that $l_0 = 0$ or $l_0 = 1$.

  If $l_0 = 0$, then $F$ is irreducible, $L$-partitioned family, where
  $L \in \mathfrak{L}_F$, i.e., $F \in \LnPirsix$, so, since
  $\covered{\LnPirsix}{\FCsix}{\nonFCsix}$, it holds that
  $\covered{F}{\FCsix}{\nonFCsix}$.

  If $l_0 = 1$, then $F - \{\emptyset\}$ is irreducible,
  $[0, l_1, \ldots, l_6]$-partitioned family, where
  $[0, l_1, \ldots, l_6] \in \mathfrak{L}_F$, i.e.,
  $F - \{\emptyset\} \in \LnPirsix$, so, since
  $\covered{\LnPirsix}{\FCsix}{\nonFCsix}$, it holds that
  $F - \{\emptyset\}$ is covered by $\FCsix$ and $\nonFCsix$. But
  then, by Proposition \ref{prop:covered}, so is $F$.
\end{proof}

\subsection{Proof of Theorem \ref{thm:allcovered}}

\setcounter{theorem}{7}
\begin{theorem}
  It holds that $\covered{\nfam{6}}{\FCsix}{\nonFCsix}$.
\end{theorem}
\begin{proof}
  Let $P = \lambda F.\ \indep{F}\ \And\ \neg (\FCcovered{F}{\FCsix})$
  and $\Pinc{P} = \lambda\ F\ A.\ \neg\ \expressible{A}{F} \And\ \neg
  (\FCcovered{F \union \{A\}}{\FCsix})$. It holds that $\Pinc{P}$
  incrementally checks $P$ and is preserved by injective
  functions. Let $val_{\emptylist} = \{\emptyset\}$, $L_{max} = [0, 6, 15, 20,
  15, 6, 1]$, and $stop = \lambda\ L.\ (\exists L_s \in
  \mathcal{L}_F.\ L \succeq L_s)$. Let $\mathcal{P}$ contain all
  permutations of $\listn{6}$, and $upd = \lambda\ \mathcal{F}\ m.\
  \base{\mult{(\mathcal{F}}{\families{n}{m}}{\Pinc{P}})}{\mathcal{P}}$.
  
  By Theorem \ref{lemma:enum_dp_mult}, for all lists
  $L = [0, \ldots, l_6] \in \mathfrak{L}_F$, there exists a collection
  $\mathcal{F}^L_b \in \enumdp{v_{\emptylist}}{upd}{stop}{L_{max}}$
  such that $\mathcal{F}^L_b$ is an iso-base of $\LnP{L}{6}{P}$. By
  definition of $P$, $\mathcal{F}^L_b$ is an iso-base of all members
  of $\LnPirsix$ that are not covered by $\FCsix$. However, a direct
  computation shows that all families in
  $\enumdp{v_{\emptylist}}{upd}{stop}{L_{max}}$ are covered by
  $\nonFCsix$. Therefore all members of $\mathcal{F}^L_b$ are covered
  by $\nonFCsix$, so by Lemma \ref{lemma:generates_covered}, all
  elements of $\LnPirsix$ that are not covered by $\FCsix$ are covered
  by $\nonFCsix$. In other words, for all $F \in \LnPirsix$ it holds
  that $\covered{F}{\FCsix}{\nonFCsix}$, i.e.,
  $\covered{\LnPirsix}{\FCsix}{\nonFCsix}$. As $\LnPirsix$
  iso-represents itself, by Lemma \ref{lemma:nonAllFCpartitions}, all
  sets in $\nfam{6}$ are covered.
\end{proof}

\section{Statistics}

As a byproduct of our classification, we have counted the number of FC
and nonFC families. For each each list $L = [l_0, \ldots, l_6]$ we
have counted $L$-partititioned families and calculated: (a) the total
number of non-isomorphic FC-families and the total number of
non-isomorphic nonFC-families, (b) the total number of non-isomorphic
irreducible FC-families and the total number of non-isomorphic
nonFC-families, (c) the total number of minimal FC-families and the
total number of maximal nonFC-families. This data is summarized in a
spreadsheat available online
(\url{http://argo.matf.bg.ac.rs/downloads/formalizations/FCFamilies.xls}).

\end{document}